\documentclass[12pt,oneside]{book}
\usepackage{amsfonts,amsmath,amssymb,amsbsy,amsthm}
\usepackage{a4}
\usepackage{latexsym}
\usepackage{multibox}
\usepackage[dvips]{graphicx,epsfig}
\usepackage{color}
\usepackage{graphics}
\usepackage{graphicx}
\usepackage{enumerate}
\usepackage{epic}
\usepackage{eepic}
\usepackage{bbm,bm}
\usepackage{dsfont}
\usepackage[titletoc]{appendix}

\begin{document}
\oddsidemargin 6pt\evensidemargin 6pt\marginparwidth
48pt\marginparsep 10pt

\thispagestyle{empty}


\noindent \vskip3.3cm

\bigskip \bigskip

\begin{center}
\large {A. ALIKHANYAN NATIONAL SCIENCE LABORATORY
\\YEREVAN PHYSICS INSTITUTE}
\end{center}

\bigskip\bigskip\bigskip

\begin{center}
{DEPARTMENT OF THEORETICAL PHYSICS}
\end{center}

\bigskip\bigskip\bigskip\bigskip \bigskip\bigskip

\begin{center}
{\Large\bf Geometric Measure of Entanglement and Schmidt Decomposition of Multipartite Systems}
\end{center}

\bigskip\bigskip\bigskip\bigskip \bigskip\bigskip \bigskip\bigskip

\begin{center}
{\large
Levon Tamaryan}\\
\medskip
{\small\it Yerevan Physics Institute\\ Alikhanian Br. Str.
2, 0036 Yerevan, Armenia}\\
\medskip

\bigskip \bigskip \bigskip \bigskip \bigskip \bigskip\bigskip

{\large PhD Thesis}\\

\bigskip \bigskip\bigskip \bigskip

Advisor: Dr. Lekdar Gevorgyan \\

\bigskip \bigskip

\end{center}

\bigskip \bigskip \bigskip\bigskip


\begin{center}
{\sc Abstract}
\end{center}

\bigskip

\noindent The thesis includes the original results of our articles
~\cite{dual,analyt,jung08-1,shared,univers,inequal,toward}. These results
are described in a concise form  below.

\bigskip

A method is developed to compute analytically entanglement
measures of three-qubit pure states. The methods leans on the
theorem stating that entanglement measures of the n-party pure
state can be  expressed by the (n-1)-party reduced state density
operator directly. Owing to this theorem algebraic equations are derived
for the geometric measure of entanglement and solved explicitly in
the cases of most interest. The solutions give analytic
expressions for the geometric entanglement measure in a wide range
of three-qubit systems, including the general class of W-type
states and states which are symmetric under the permutation of two
qubits~\cite{analyt,jung08-1}.

\bigskip

The same method is used to find the geometric measure of
entanglement of generic three-qubit pure states. Closed-form
expressions are presented for the geometric measure of
entanglement for three-qubit states that are linear combinations
of four orthogonal product states. It turns out that the geometric
measure for these states has three different expressions depending
on the range of definition in parameter space. Each expression of
the measure has its own geometrically meaningful interpretation
and thus the Hilbert space of three-qubits consists of three
different entangled regions. The states that lie on joint surfaces
separating different entangled regions, designated as shared
states, have particularly interesting features and are dual
quantum channels for the perfect teleportation and superdense
coding~\cite{shared}.

\bigskip

A powerful method is developed to compute analytically
multipartite entanglement measures. The method uses the duality
concept and creates a bijection between highly entangled quantum
states and their nearest separable states. The bijection gives
explicitly the geometric entanglement measure of arbitrary
generalized W states of n qubits and singles out two critical
points of entanglement in quantum state parameter space. The first
critical value separates symmetric and asymmetric entangled
regions of highly entangled states, while the second one separates
highly and slightly entangled states~\cite{dual,toward}.

\bigskip

The behavior of the geometric entanglement measure of many-qubit W
states is analyzed and an interpolating formula is derived. The
importance of the interpolating formula in quantum information is
threefold. First, it connects  quantities that can be easily
estimated in experiments. Second, it is an example of how we
compute entanglement of a quantum state with many unknowns. Third,
one can prepare the W state with a given entanglement bringing
into the position a single quantity~\cite{univers}.

\bigskip

Generalized Schmidt decomposition of pure three-qubit states has four positive and one complex coefficients. In contrast to the bipartite case, they are not arbitrary and the largest Schmidt coefficient restricts severely other coefficients. It is derived a non-strict inequality between three-qubit Schmidt coefficients, where the largest coefficient defines the least upper bound for the three nondiagonal coefficients or, equivalently, the three nondiagonal coefficients together define the greatest lower bound for the largest coefficient. Besides, it is shown the existence of another inequality which should establish an upper bound for the remaining Schmidt coefficient~\cite{inequal}.

\newpage

\tableofcontents

\newpage


\chapter*{Introduction}\label{intro}
\addcontentsline{toc}{chapter}{Introduction}

\newcommand{\uv}{{\bm u}}
\newcommand{\rv}{{\bm r}}
\newcommand{\ra}{\rangle}
\newcommand{\la}{\langle}
\newcommand{\vp}{\varphi}
\newcommand{\tr}{\mathrm{tr}}
\renewcommand{\t}{\theta}

\setcounter{equation}{0}

Building quantum information processing devices is a great
challenge for scientists and engineers of the third millennium
~\cite{fey-82,deu-85}. Compound quantum systems have potential for
many quantum processes, including the following applications:
factoring of large composite numbers ~\cite{shor-94,shor-exp},
quantum cryptography~\cite{ek-91,exp-cr}, superdense
coding~\cite{dense-92,exp-den}, quantum
teleportation~\cite{tele-93,exp-tel} and exponential speedup of
quantum computers~\cite{vidal-03,joz-03,12qubit}. These remarkable
phenomena have provided a basis for the development of modern
quantum information science.

The superior performance of quantum systems in computation and
communication applications is rooted in a property of quantum
mechanical states called entanglement~\cite{epr,bel,wer-89}. Quantum
entanglement is a physical resource associated with the peculiar
nonclassical correlations that are possible between separated
quantum systems. It is a fundamental property of quantum systems
and a basic physical resource for quantum information science~\cite{hor-09}.
In general, any task involving distant parties and using up entangled states as
a resource benefits from a better understanding of entanglement. It is
increasingly realized that quantum entanglement is at the heart of
quantum physics and as such it may be of very broad importance for
modern science and future technologies.

Entanglement is usually created by direct interactions between
subatomic particles. If two particles are entangled, then there is
a correlation between the results of measurements performed on
entangled pairs, and this correlation is observed even though the
entangled pair may have been separated by arbitrarily large
distances. In the multipartite case the entanglement is more
complicated concept and to distinguish entangled and unentangled
quantum states in this case it is necessary to define product
states and separable states.

Consider multipartite systems. The Hilbert space of a such system
is the tensor product of the Hilbert spaces of single particles.
There is a simple definition of unentangled states in the case of
pure states. Indeed, the vector(pure state) belonging to the
Hilbert space of the multipartite system is called a product state
if it is a tensor product of vectors(pure states) belonging to
Hilbert spaces of single particles. In other words, a pure state
of a multi-particle system is a product state if and only if all
subsystems are pure states. Clearly, there is no correlation
between subsystems of product states and they are unentangled
states.

Consider now mixed states of a multipartite system. The
generalization of the definition of unentangled states to mixed
states leans on the local operations and classical
communication(LOCC). Local operations and classical communication
is a method in quantum information theory where a local operation
is performed on part of the system, and where the result of that
operation is communicated classically to another part where
usually another local operation is performed. Since no quantum
interaction occurs within these actions it is natural to assume
that no entanglement can be created by LOCC alone, which is to say
that LOCC can decrease, but never increases entanglement. Hence,
any mixed state that can be obtained from product states via LOCC
is unentangled. It is shown, that if a mixed state is a
probability distribution over product states, known as separable
states, can be created from product states by LOCC
alone~\cite{plenvir-07}. Then there are two definitions: first,
separable states are unentangled and second, non-separable states
are entangled.

One of the most difficult and at the same time fundamental
questions in entanglement theory is quantifying
entanglement~\cite{ben-conc,woot-98,ben-error}. The basic
requirements to an entanglement measures rely on LOCC and local
unitary transformations(LU) that is  unitary transformations which
act on single particles separately. These requirements can be
formulated as follows~\cite{plenvir-07}:

\begin{itemize}

\item {Separable states contain no entanglement.}

\item {All nonseparable states are entangled.}

\item {The entanglement of states does not increase under LOCC
operations.}

\item {Entanglement does not change under LU-transformations.}

\end{itemize}

Many entanglement measures have been proposed for the two-particle
as well as for the multi-particle case~\cite{plenvir-07}. They are
very difficult to compute as their definition contains
optimizations over certain quantum states or quantum information
protocols. In bipartite case entanglement is relatively well
understood, while in multipartite case quantifying entanglement of
pure states is a question of vital importance.

The geometric measure of entanglement(GM) is one of the most
reliable quantifiers of multipartite
entanglement~\cite{shim-95,vedr-97,barn-01,wei-03}. It measures
the distance of a given quantum state from the set of product
states and is a decreasing function of the maximal product overlap
of the quantum state. The maximal product overlap (MPO) of a
quantum pure state is the absolute value of the inner product of
the quantum state and its nearest separable state. It(or its
square) has several names and we list all of them for the
completeness: entanglement eigenvalue~\cite{wei-03}, injective
tensor norm~\cite{wern-02}, maximal probability of
success~\cite{biham-02}, maximum singular value~\cite{sud-geom}
and maximal product overlap~\cite{dual}.

The geometric measure of entanglement (GM) has the following
remarkable properties and applications:

\begin{enumerate}

\item {It has identified irregularity in channel capacity
additivity. Using this measure, one can show that a family of
quantities, which were thought to be additive in earlier papers,
actually are not}~\cite{wern-02}.

\item {It has an operational treatment and quantifies how
well a given state serves as an input state to Grover's search
algorithm}~\cite{grov,biham-02}.

\item {It has useful connections to other entanglement
measures and gives rise to a lower bound on the relative entropy
of entanglement~{\rm\cite{wei-04}} and generalized
robustness}~\cite{cav-06}.

\item {It quantifies the difficulty to distinguish
multipartite quantum states by local means}~\cite{local}.

\item {It exhibits interesting connections with
entanglement witnesses and can be efficiently estimated in
experiments}~\cite{guh-07}.

\item {It has been used to prove that one dimensional
quantum systems tend to be globally separable along
renormalization group flows by following a universal scaling law
in the correlation length of the system. Owing to this one can
understand the physical implication of Zamolodchikov's c-theorem
more deeply}~\cite{orus-08}.

\item {It has been used to study quantum phase transitions
in spin models}~\cite{orusand}.

\item {It singles out states that can be used as a quantum
channel for the perfect teleportation and superdence
coding}~\cite{analyt}.

\item {It gives the largest coefficient of the generalized
Schmidt decomposition and the corresponding nearest product state
uniquely defines the factorisable basis of the
decomposition}~\cite{hig}.

\item {It has been used to derive a single-parameter family
of the maximally entangled three-qubit  states, where the
paradigmatic Greenberger-Horne-Zeilinger and W states emerge as
the extreme members in this family of maximally entangled states}
~\cite{maxim}.
\end{enumerate}

\noindent Owing to these features, GM can play an important role in the
investigation of different problems related to entanglement. In
spite of its usefulness one obstacle to use GM fully in quantum
information theories is the that it is difficult to compute
it analytically for generic states. The usual maximization method
generates a system of nonlinear equations which are unsolvable in
general. Thus, it is important to develop a technique for the
computation of GM~ \cite{jung08-1,hay-wir,shared,wei-sev,wei-guh}.

For bipartite systems maim problems problems related to entanglement have been solved with the help of the Schmidt decomposition~\cite{schmidt1907,ek-sch}. Therefore its generalization to multipartite states can solve difficult problems related to multipartite entanglement. This generalization for three qubits is done by Ac\'in {\it et al}~\cite{acin}, where it is shown that an arbitrary pure state can be written as a linear combination of five product states. Independently, Carteret {\it et al} developed a method for such a generalization for pure states of arbitrary multipartite system, where the dimensions of the individual state spaces are finite but otherwise arbitrary~\cite{hig}.

However, for a given quantum state the canonical form is not unique and the same state can have different canonical forms and therefore different sets of such amplitudes. The reason is that the stationarity equations defining stationarity points are nonlinear equations and in general have several solutions of different types. Then the question is which of amplitude sets should be treated as Schmidt coefficients and which ones should be treated as insignificant mathematical solutions. A criterion should exist that can distinguish right Schmidt coefficients from false ones and we need such a criterion. It is unlikely that we can solve problems of three-qubit entanglement without knowledge of what quantities are the relevant entanglement parameters.

\bigskip

The main goals of the thesis are:
\begin{description}
 \item[\quad {\rm a)}] to develop methods that allow us to compute analytically multipartite entanglement measures,
 \item[\quad {\rm b)}] to derive analytic expressions for the geometric measure of entanglement of multi-particle systems,
 \item[\quad {\rm c)}] to analyze basic phenomena in quantum information theory using closed form solutions for geometric measure.
 \item[\quad {\rm d)}] to find inequalities which define a unique Schmidt decomposition for generic multipartite systems
\end{description}

We have developed two powerful methods to compute analytically
multipartite entanglement measures.

The first method, hereafter referred to as \textit{reduced density
method}, allows us to compute analytically entanglement measures
of three-qubit pure states. The three-qubit system is important in
the sense that it is the simplest system which gives a nontrivial
effect in the entanglement. Thus, we should understand the general
properties of the entanglement in this system as much as possible
to go further to more complicated higher-qubit systems. The
three-qubit system can be entangled in two inequivalent ways --
Greenberger-Horne-Zeilinger (GHZ)~\cite{ghz} and W -- and neither
form can be transformed into the other with any probability of
success~\cite{Chir}. This picture is complete: any fully entangled
three-qubit pure state can be obtained from either the GHZ or W
state via stochastic local operations and classical communication
(SLOCC).

The reduced density method leans on the theorem stating that any
reduced (n-1)-qubit state uniquely determines the entanglement of
the original n-qubit pure state~\cite{reduced}. This means that
two-qubit mixed states can be used to calculate the geometric
measure of three-qubit pure states. This idea converts the task
effectively into the maximization of the two-qubit mixed state
over product states and yields {linear eigenvalue
equations}. Owing to this substantial simplification closed form
expressions can be derived for the geometric measure of
three-qubit pure states. This is fully addressed in works
\cite{analyt,shared,jung08-1}.

The second method, hereafter referred to as \textit{duality
method}, allows us to compute analytically the entanglement
measures of highly entangled n-qubit pure states. The main point
of the method is the theorem stating that the nearest product
state is essentially unique if the quantum state is highly
entangled~\cite{dual}. This makes it possible to map highly
entangled state to its nearest product state and quickly obtain
its geometric measure of entanglement. More precisely, we
construct two bijections. The first one creates a map between
highly entangled n-qubit quantum states and n-dimensional unit
vectors. The second one does the same between n-dimensional unit
vectors and n-part product states. Thus we obtain a double map, or
{\it duality}, as follows

\bigskip
\centerline{{ n-qubit pure states $\leftrightarrow$ n-dimensional
spatial vectors $\leftrightarrow$ n-part product states.}}
\bigskip

The main advantage of the map is that if one knows any of the
three vectors, then one instantly finds the other two. Hence we
find the geometric measure of entanglement of general multiqubit W
states.

The derived answer shows that highly entangled W states have two
exceptional points in the parameter space. At the second
exceptional point the reduced density operator of a some qubit is
a constant multiple of the unit operator and then the maximal
product overlap of these states is a constant regardless how many
qubits are involved and what are the values of the remaining
entanglement parameters. These states are known as shared quantum
states and can be used as quantum channels for the perfect
teleportation and dense coding.

Next it is shown that W-states have two different entangled
regions: the symmetric and asymmetric entangled regions. In the
computational basis these regions can be defined as follows. If a
W state is in the symmetric region, then the entanglement is a
fully symmetric function on the state parameters. Conversely, if a
W state is in the asymmetric region, then there is an exceptional
parameter such that the entanglement dependence on the exceptional
parameter differs dramatically from the dependencies of the
remaining parameters. Hence the point of intersection of the
symmetric and asymmetric regions is the first exceptional point.

The first exceptional point is important for large-scale W
states~\cite{univers}. It approaches to a fixed point when number
of qubits n increases and becomes state-independent(up to 1/n
corrections) when $n \gg 1$. As a consequence the entanglement, as
well as the maximal product overlap, becomes state-independent too
and therefore many-qubit W states have two state-independent
exceptional points. The underlying concept is that states whose
entanglement parameters differ widely may nevertheless have the
same maximal product overlap and this phenomenon should occur at
two fixed points. This is an analog of the universality of
dynamical systems at critical points. It is an intriguing fact
that systems with quite different microscopic parameters may
behave equivalently at criticality. Fortunately the
renormalization group provides an explanation for the emergence of
universality in critical systems~\cite{crit-rev}.

To construct generalized Schmidt decomposition(GSD) for arbitrary systems we apply the variational principle~\cite{inequal}. In order to extend uniquely the Schmidt decomposition to multipartite systems we require that its largest coefficient, as in bipartite case, is the maximal product overlap, otherwise it is an irrelevant solution of stationarity equations.
It is clear how do we single out the canonical form whose largest coefficient is the maximal product overlap. We should single out the closest product state of a given quantum state that gives a true maximum for overlap.  Of course, we cannot find closest product states of generic three-qubit states because there is no method to solve generic stationarity equation so far. Hence to distinguish the true maximum from other stationary points we require that the second variation of the maximal product overlap is negative everywhere and this condition yields the desired inequality.

\bigskip

The thesis consists of Introduction, six Chapters, Summary and
Bibliography.

\bigskip

In {\bf Chapter}~\ref{analytic-3q} we use the reduced density
method to compute analytically the geometric measure of
entanglement of GHZ-type and W-type three-qubit pure states~\cite{analyt}.
We derive explicit expressions for the maximal product overlaps and
closest product states of those states and show that W-type states
consist of two different classes. They are: slightly entangled
W-states for which MPO is the absolute value of the largest
amplitude of the quantum state in the computational basis and
highly entangled W-states for which MPO is the circumradius of the
triangle whose sides are absolute values of the amplitudes of the
quantum state in the same basis.

\medskip

In {\bf Chapter}~\ref{groverian} the same method is used to
connect the maximal product overlap with the polynomial invariants
of three-qubit pure states~\cite{jung08-1}. It is well known that
these states have five polynomial invariants~\cite{sud00}, i.e. invariants
under LU-transformations. Since entanglement should be invariant
under LU-transformations polynomial invariants are real variables
of entanglement measures and the relation between MPO and
polynomial invariants is independent from the choice of a
particular computational basis. Hence we use this relation to
classify entangled regions of the Hilbert space as follows: in
each region some of polynomial invariants are important and define
uniquely MPO while the remaining polynomial invariants are
irrelevant. In this way we obtained six different entangled
regions for three-qubit pure states.

\medskip

In {\bf Chapter}~\ref{shared-1} we use the reduced density method
to compute analytically the geometric entanglement measure of
generic three-qubit pure states which are linear superpositions of
GHZ- and W-type states~\cite{shared}. We give an explicit expression for the
geometric measure of entanglement for three-qubit states that are
linear combinations of four orthogonal product states. It turns
out that the geometric measure for these states has three
different expressions depending on the range of definition in
parameter space. Each expression of the measure has its own
geometrically meaningful interpretation. Such an interpretation
allows oneself to take one step toward a complete understanding
for the general properties of the entanglement measure. The states
that lie on joint surfaces separating different ranges of
definition, designated as shared states, are dual quantum channels
for the perfect teleportation and superdense coding. The
properties of the shared states are fully discussed.

\medskip

In {\bf Chapter}~\ref{duality} we use the duality method to
compute analytically the geometric entanglement measure of generic
n-qubit W-type states~\cite{dual}. We have constructed correspondences among W
states, n-dimensional unit vectors, and separable pure states. The
map reveals two critical values for quantum state parameters. The
first critical value separates symmetric and asymmetric entangled
regions of highly entangled states, while the second one separates
highly and slightly entangled states. The method gives an explicit
expressions for the geometric measure when the state allows
analytical solutions; otherwise it expresses the entanglement as
an implicit function of state parameters.

\medskip

In {\bf Chapter}~\ref{universal} we analyze physical features of
entanglement of many-quabit pure states~\cite{univers}. We show that when $n\gg1$
the geometric entanglement measure of general n-qubit W-states,
except maximally entangled W-states, is a one-variable function
and depends only on the Bloch vector with the minimal $z$
component. Hence one can prepare a W state with the required
maximal product overlap by altering the Bloch vector of a single
qubit. Next we compute analytically the geometric measure of
large-scale W states by describing these systems in terms of very
few parameters. The final formula relates two quantities, namely
the maximal product overlap and the Bloch vector, that can be
easily estimated in experiments.

\medskip

In {\bf Chapter}~\ref{schmidt} we derive a non-strict inequality between three-qubit Schmidt coefficients, where the largest coefficient defines the least upper bound for the three nondiagonal coefficients or, equivalently, the three nondiagonal coefficients together define the greatest lower bound for the largest coefficient. The main role of the inequality is to separate out three-qubit Schmidt coefficients from the set of four positive and one complex numbers. Besides, it is shown the existence of another inequality which should establish an upper bound for the remaining Schmidt coefficient.

\medskip

In {\bf Summary} we give the main points of our results and conclusions.

\medskip

In {\bf Bibliography} we list our references in order of appearance.


\chapter[Three Qubit Geometric Measure]{ Analytic Expressions for Geometric Measure of Three Qubit States}\label{analytic-3q}

\renewcommand{\ra}{\rangle}
\renewcommand{\la}{\langle}
\renewcommand{\tr}{\mathrm{tr}}
\newcommand{\p}{{\bm r}}
\newcommand{\s}{{\bm \sigma}}
\newcommand{\sv}{{\bm s}}
\newcommand{\nv}{{\bm n}}
\newcommand{\iv}{{\bm i}}
\newcommand{\mv}{{\bm m}}
\newcommand{\sq}{\varrho}
\newcommand{\lm}{\Lambda_{\max}^2}
\newcommand{\openone}{{\textsl{1\!\!1}}}
\newcommand{\ket}[1]{\vert #1 \ra}
\newcommand{\bra}[1]{\la #1 \vert}
\newcommand{\bk}[2]{\la #1 \vert #2 \ra}
\newcommand{\kb}[2]{\vert #1 \ra \la #2 \vert}
\newcommand{\ov}[2]{\left\la #1 | #2 \right\ra}

In this chapter we compute analytically the geometric measure of
entanglement of three-qubit pure states \cite{analyt}.

The entanglement of bipartite systems is well-understood
\cite{ben-conc,woot-98,ben-error,niels}, while the entanglement of
multipartite systems offers a real challenge to physicists. The
main point which makes difficult to understand the entanglement
for the multi-qubit systems is mainly due to the fact that the
analytic expressions for the various entanglement measures is
extremely hard to derive.

We consider pure three qubit systems \cite{acin,lind,coff,red,swiz},
although the entanglement of mixed states attracts a considerable
attention. Only very few analytical results for tripartite
entanglement have been obtained so far and we need more light on
the subject.

Recently the idea was suggested that nonlinear eigenproblem can be
reduced to the linear eigenproblem for the case of three qubit
pure states \cite{reduced}. The idea is based on theorem stating that
any reduced $(n-1)$-qubit state uniquely determines the geometric
measure of the original $n$-qubit pure state. This means that two
qubit mixed states can be used to calculate the geometric measure
of three qubit pure states and this will be fully addressed in
this work.

The method gives two algebraic equations of degree six defining
the geometric measure of entanglement. Thus the difficult problem
of geometric measure calculation is reduced to the algebraic
equation root finding. Equations contain valuable information, are
good bases for the numerical calculations and may test numerical
calculations based on other numerical techniques \cite{Shim-shor}.

Furthermore, the method allows to find the nearest separable
states for three qubit states of most interest and get analytic
expressions for their geometric measures. It turn out that highly
entangled states have their own feature. Each highly entangled
state has a vicinity with no product state and all nearest product
states are on the boundary of the vicinity and form an
one-parametric set.

This chapter is organized as follows.
In Section 1.1 we define the geometric measure of
entanglement and derive stationarity equations.
In Section 1.2 we derive algebraic equations in the case of pure
three qubit states and give general solutions.
In Section 1.3 we examine W-type
states and deduce analytic expression for their geometric
measures. States symmetric under permutation of two qubits are
considered in Section 1.4, where the overlap of the state
functions with the product states are maximized directly. In last
Section 1.5 we make concluding remarks.

\section{Geometric measure of entanglement}

We start by developing a general formulation,
appropriate for multipartite systems comprising n parts,
in which each part has its distinct Hilbert space.
Let $|\psi\rangle$ be a pure state of an $n$-party system ${\cal H} =
{\cal H}_1\otimes{\cal H}_2\otimes\cdots\otimes{\cal H}_n$ , where
the dimensions of the individual state spaces ${\cal H}_k$ are
finite but otherwise arbitrary. Denote by $\ket{q_1q_2...q_n}$ product states
which are defined as the tensor products

$$\ket{q_1q_2...q_n}\equiv\ket{q_1}\otimes\ket{q_2}\otimes\cdots\otimes\ket{q_n},$$
where $\ket{q_k}\in{\cal H}_k, k=1,2,...,n $.

The geometric measure of entanglement $E_g$ for an $n$-part pure
state $\psi$ is defined as $E_g(\psi)=-\ln \lm(\psi)$, where the
maximal product overlap $\Lambda_{max}(\psi)$ is given
by~\cite{wei-03}
\begin{equation}\label{mpo}
\Lambda_{max}=\max_{q_1,q_2,...,q_n}|\ov{\psi}{q_1q_2...q_n}|,
\end{equation}
where the normalization condition $\ov{q_k}{q_k}=1(k=1,2,...,n)$ is understood
and the maximization is performed over all product states.

The nearest product state is a stationary point for the overlap
with $\ket{\psi}$, so the states $\ket{q_k}$ satisfy the nonlinear
eigenvalue equations
\begin{equation}\label{1.stat-eqq}
\bk{q_1q_2\cdots\widehat{q_k}\cdots
q_n}{\psi}=\Lambda_{k}\ket{q_k};\;k=1,2,\cdots,n,
\end{equation}
where the caret means exclusion and eigenvalues $\Lambda_{k}$
are associated with the Lagrange
multipliers enforcing constraints $\ov{q_k}{q_k}=1(k=1,2,...,n)$.

Since phases of local states $\ket{q_k}$ are irrelevant one can
choose them such that $\Lambda_{k}$'s are all positive. On the other
hand $|\Lambda_{k}|=\Lambda_{max}$ and therefore the stationarity
equations can be rewritten as
\begin{equation}\label{1.stat-eq}
\bk{q_1q_2\cdots\widehat{q_k}\cdots q_n}{\psi}=\Lambda_{max}\ket{q_k};\;k=1,2,\cdots,n.
\end{equation}
This is a system of nonlinear equations and its maximal eigenvalue
and corresponding eigenvector
are the maximal product overlap and the nearest product
state of a given pure states $\ket{\psi}$, respectively.

The extension of the geometric measure of entanglement to mixed states
can be made via the use of the convex roof (or hull ) construction,
as it is done for the entanglement of formation~\cite{woot-98}.
We omit it since  mixed states are not considered in this thesis.

\section{Algebraic equations.}

Consider now three qubits A,B,C with state function $|\psi\ra$.
The entanglement eigenvalue $\Lambda_{max}(\psi)$ is given by

\begin{equation}\label{gen.lm}
\Lambda_{max}=\max_{q^1q^2q^3}|\la q^1q^2q^3|\psi\ra|
\end{equation}

\noindent and the maximization runs over all normalized complete
product states $|q^1\ra\otimes|q^2\ra\otimes|q^3\ra$. Superscripts
label single qubit states and spin indices are omitted for
simplicity. Since in the following we will use density matrices
rather than state functions, our first aim is to rewrite
Eq.(\ref{gen.lm}) in terms of density matrices. Let us denote by
$\rho^{ABC}=|\psi\ra\la\psi|$ the density matrix of the
three-qubit state and by $\sq^k=|q^k\ra\la q^k|$ the density
matrices of the single qubit states. The equation for the square
of the entanglement eigenvalue takes the form

\begin{equation}\label{gen.pmax}
\lm(\psi)=\max_{\sq^1\sq^2\sq^3}
\tr\left(\rho^{ABC}\sq^1\otimes\sq^2\otimes\sq^3\right).
\end{equation}

An important equality

\begin{equation}\label{gen.unit}
\max_{\sq^3} \tr(\rho^{ABC}\sq^1\otimes\sq^2\otimes\sq^3)=
\tr(\rho^{ABC}\sq^1\otimes\sq^2\otimes\openone^3)
\end{equation}

\noindent was derived in \cite{reduced} where $\openone$ is a unit
matrix. It has a clear meaning. The matrix
$\tr(\rho^{ABC}\sq^1\otimes\sq^2)$ is $2\otimes2$ hermitian matrix
and has two eigenvalues. One of eigenvalues is always zero and
another is always positive and therefore the maximization of the
matrix simply takes the nonzero eigenvalue. Note that its
minimization gives zero as the minimization takes the zero
eigenvalue.

We use Eq.(\ref{gen.unit}) to reexpress the entanglement
eigenvalue by reduced density matrix $\rho^{AB}$ of qubits A and B
in a form

\begin{equation}\label{gen.pred.analytic}
\lm(\psi)=\max_{\sq^1\sq^2}
\tr\left(\rho^{AB}\sq^1\otimes\sq^2\right).
\end{equation}

We denote by $\sv_1$ and $\sv_2$ the unit Bloch vectors of the
density matrices $\sq^1$ and $\sq^2$ respectively and adopt the
usual summation convention on repeated indices $i$ and $j$. Then

\begin{equation}\label{gen.s1s2}
\lm=\frac{1}{4}\max_{s_1^2=s_2^2=1}\left(1+\sv_1\cdot
\p_1+\sv_2\cdot \p_2+g_{ij}\,s_{1i}s_{2j}\right),
\end{equation}

where

\begin{equation}\label{gen.vec}
\p_1=\tr(\rho^A\s),\,\,\p_2=\tr(\rho^B\s),\,\,
g_{ij}=\tr(\rho^{AB}\sigma_i\otimes\sigma_j)
\end{equation}

\noindent and $\sigma_i$'s are Pauli matrices. The matrix $g_{ij}$
is not necessarily to be symmetric but must has only real entries.
The maximization gives a pair of equations

\begin{equation}\label{gen.eq}
\p_1+g\sv_2=\lambda_1\sv_1,\quad\p_2+g^T\sv_1=\lambda_2\sv_2,
\end{equation}

\noindent where Lagrange multipliers $\lambda_1$ and $\lambda_2$
are enforcing unit nature of the Bloch vectors. The solution of
Eq.(\ref{gen.eq}) is

\begin{subequations}\label{gen.sol}
\begin{equation}\label{gen.sol1}
\sv_1=\left(\lambda_1\lambda_2\openone-g\,g^T\right)^{-1}
\left(\lambda_2\p_1+g\,\p_2\right),
\end{equation}
\begin{equation}\label{gen.sol2}
\sv_2=\left(\lambda_1\lambda_2\openone-g^Tg\right)^{-1}
\left(\lambda_1\p_2+g^T\p_1\right).
\end{equation}
\end{subequations}

\noindent Now, the only unknowns are Lagrange multipliers, which
should be determined by equations

\begin{equation}\label{gen.alg}
|\sv_1|^2=1,\quad |\sv_2|^2=1.
\end{equation}

In general, Eq.(\ref{gen.alg}) give two algebraic equations of
degree six. However, the solution (\ref{gen.sol}) is valid if
Eq.(\ref{gen.eq}) supports a unique solution and this is by no
means always the case. If the solution of Eq.(\ref{gen.eq})
contains a free parameter, then
$\det(\lambda_1\lambda_2\openone-gg^T)=0$ and, as a result,
Eq.(\ref{gen.sol}) cannot not applicable. The example presented in
Section III will demonstrate this situation.

In order to test Eq.(\ref{gen.sol}) let us consider an arbitrary
superposition of W

\begin{equation}\label{ww.w}
|W\ra=\frac{1}{\sqrt{3}}\left(|100\ra+|010\ra+|001\ra\right)
\end{equation}

\noindent and flipped W

\begin{equation}\label{ww.flip}
|\widetilde{W}\ra=\frac{1}{\sqrt{3}}\left(|011\ra+|101\ra+|110\ra\right)
\end{equation}

\noindent states, i.e. the state

\begin{equation}
|\psi\ra=\cos\theta|W\ra+\sin\theta|\widetilde{W}\ra.
\end{equation}

Straightforward calculation yields

\begin{subequations}\label{ww.matr}
\begin{equation}\label{ww.matrv}
\p_1=\p_2=\frac{1}{3}\left(2\sin2\theta\iv+\cos2\theta\nv\right),
\end{equation}
\begin{equation}\label{ww.matrm}
g=\frac{1}{3}
\begin{pmatrix}
2 & 0 & 0\\
0 & 2 & 0\\
0 & 0 &\!\! -1
\end{pmatrix}
,
\end{equation}
\end{subequations}

\noindent where unit vectors $\iv$ and $\nv$ are aligned with the
axes $x$ and $z$, respectively. Both vectors $\iv$ and $\nv$ are
eigenvectors of matrices $g$ and $g^T$. Therefore $\sv_1$ and
$\sv_2$ are linear combinations of $\iv$ and $\nv$. Also from
$\p_1=\p_2$ and $g=g^T$ it follows that $\sv_1=\sv_2$ and
$\lambda_1=\lambda_2$. Then Eq.(\ref{gen.sol}) for general
solution give

\begin{equation}\label{ww.sv}
\sv_1=\sv_2=\sin2\varphi\,\iv+ \cos2\varphi\,\nv
\end{equation}

\noindent where

\begin{equation}\label{ww.var}
\sin2\varphi=\frac{2\sin2\theta}{3\lambda-2},\quad
\cos2\varphi=\frac{\cos2\theta}{3\lambda+1}.
\end{equation}

The elimination of the Lagrange multiplier $\lambda$ from
Eq.(\ref{ww.var}) gives

\begin{equation}\label{ww.el}
3\sin2\varphi\cos2\varphi=\cos2\theta\sin2\varphi-
2\sin2\theta\cos2\varphi.
\end{equation}

Let us denote by $t=\tan\varphi$. After the separation of the
irrelevant root $t=-\tan\theta$, Eq.(\ref{ww.el}) takes the form

\begin{equation}\label{ww.eq}
\sin\theta\,t^3+2\cos\theta\,t^2-2\sin\theta\,t-\cos\theta=0.
\end{equation}

\noindent This equation exactly coincides with that derived in
\cite{wei-03}. Since a detailed analysis was given in
Ref.\cite{wei-03}, we do not want to repeat the same calculation
here. Instead we would like to consider the three-qubit states
that allow the analytic expressions for the geometric entanglement
measure by making use of Eq.(\ref{gen.eq}).

\section{W-type states.}

Consider W-type state

\begin{equation}\label{w.psi}
|\psi\ra=a|100\ra+b|010\ra+c|001\ra,\quad a^2+b^2+c^2=1.
\end{equation}

\noindent Without loss of generality we  consider only the case of
positive parameters $a,b,c$. Direct calculation yields

\begin{equation}\label{w.matr}
 \p_1=r_1\,\nv,\quad\p_2=r_2\,\nv,\quad g=
\begin{pmatrix}
\omega & 0 & 0\\
0 & \omega & 0\\
0 & 0 & -r_3
\end{pmatrix}
,
\end{equation}
where
\begin{equation}\label{w.eig}
r_1=b^2+c^2-a^2,\, r_2=a^2+c^2-b^2,\,r_3=a^2+b^2-c^2
\end{equation}

\noindent and $\omega=2ab$. The unit vector $\nv$ is  aligned with
the axis $z$. Any vector perpendicular to $\nv$ is an eigenvector
of $g$ with eigenvalue $\omega$. Then from Eq.(\ref{gen.eq}) it
follows that the components of vectors $\sv_1$ and $\sv_2$
perpendicular to $\nv$ are collinear. We denote by $\mv$ the unit
vector along that direction and parameterize vectors $\sv_1$ and
$\sv_2$ as follows

\begin{equation}\label{w.par}
\sv_1=\cos\alpha\,\nv+\sin\alpha\,\mv,\quad
\sv_2=\cos\beta\,\nv+\sin\beta\,\mv.
\end{equation}

Then Eq.(\ref{gen.eq}) reduces to the following four equations

\begin{subequations}\label{w.angle}
\begin{equation}\label{w.cos}
r_1-r_3\cos\beta=\lambda_1\cos\alpha,\quad
r_2-r_3\cos\alpha=\lambda_2\cos\beta,
\end{equation}
\begin{equation}\label{w.sin}
\omega\sin\beta=\lambda_1\sin\alpha,\quad
\omega\sin\alpha=\lambda_2\sin\beta,
\end{equation}
\end{subequations}

\noindent which are used to solve the four unknown constants
$\lambda_1,\lambda_2,\alpha$ and $\beta$. Eq.(\ref{w.sin}) impose
either

\begin{equation}\label{w.l1l2}
\lambda_1\lambda_2-\omega^2=0
\end{equation}

\noindent or

\begin{equation}\label{w.sinsin}
\sin\alpha\sin\beta=0.
\end{equation}

First consider the case $r_1>0,r_2>0,r_3>0$ and coefficients
$a,b,c$ form an acute triangle. Eq.(\ref{w.sinsin}) does not give
a true maximum and this can be understood as follows. If both
vectors $\sv_1$ and $\sv_2$ are aligned with the axis $z$, then
the last term in Eq.(\ref{gen.s1s2}) is negative. If vectors
$\sv_1$ and $\sv_2$ are antiparallel, then one of scalar products
in Eq.(\ref{gen.s1s2}) is negative. In this reason $\lm$ cannot be
maximal. Then Eq.(\ref{w.l1l2}) gives true maximum and we have to
choose positive values for $\lambda_1$ and $\lambda_2$ to get
maximum.

First we use Eq.(\ref{w.cos}) to connect the angles  $\alpha$ and
$\beta$ with the Lagrange multipliers $\lambda_1$ and $\lambda_2$

\begin{equation}\label{w.cosval}
\cos\alpha=\frac{\lambda_2r_1-r_2r_3}{\omega^2-r_3^2},\quad
\cos\beta=\frac{\lambda_1r_2-r_1r_3}{\omega^2-r_3^2}.
\end{equation}

Then Eq.(\ref{w.sin}) and (\ref{w.l1l2}) give the following
expressions for Lagrange multipliers $\lambda_1$ and $\lambda_2$

\begin{subequations}\label{w.lagroot}
\begin{equation}\label{w.lagroot1}
\lambda_1=\omega\left(\frac{\omega^2+r_1^2-r_3^2}
{\omega^2+r_2^2-r_3^2}\right)^{1/2},
\end{equation}
\begin{equation}\label{w.lagroot2}
\lambda_2=\omega\left(\frac{\omega^2+r_2^2-r_3^2}
{\omega^2+r_1^2-r_3^2}\right)^{1/2}.
\end{equation}
\end{subequations}

Eq.(\ref{gen.eq}) allows to write a shorter expression for the
entanglement eigenvalue

\begin{equation}\label{w.lsimple}
\lm=\frac{1}{4}\left(1+\lambda_2+r_1\cos\alpha\right).
\end{equation}

Now we insert the values of $\lambda_2$ and $\cos\alpha$ into
Eq.(\ref{w.lsimple}) and obtain

\begin{equation}\label{w.frac}
4\lm=1+\frac{\omega\sqrt{(\omega^2+r_1^2-r_3^2)
(\omega^2+r_2^2-r_3^2)}-r_1r_2r_3}{\omega^2-r_3^2}.
\end{equation}

The denominator in above expression is  multiple of the area $S$
of the triangle $a,b,c$

\begin{equation}\label{w.area}
\omega^2-r_3^2=16S^2.
\end{equation}

A little algebra yields for the numerator

\begin{eqnarray}\label{w.numer}
 & & \omega\sqrt{(\omega^2+r_1^2-r_3^2)+
(\omega^2+r_2^2-r_3^2)}-r_1r_2r_3    \\ \nonumber & &\hspace{.3cm}
= 16\,a^2b^2c^2-\omega^2+r_3^2.
\end{eqnarray}

Combining together the numerator and denominator, we obtain the
final expression for the entanglement eigenvalue

\begin{equation}\label{w.rad}
\lm=4R^2,
\end{equation}

\noindent where $R$ is the circumradius of the triangle $a,b,c$.
Entanglement value is minimal  when triangle is regular, i.e. for
W-state and $\lm(W)=4/9$ \cite{Shim-grov}.

Now consider the case $r_3<0$. Since $r_3+r_1=2b^2\geq0$, we have
$r_1>0$ and similarly $r_2>0$. Eq.(\ref{w.sinsin}) gives true
maximum in this case and both vectors are aligned with the axis
$z$

\begin{equation}\label{w.parz}
\sv_1=\sv_2=\nv
\end{equation}

\noindent resulting in $\lm=c^2$. In view of symmetry

\begin{equation}\label{w.max}
\lm=\max(a^2,b^2,c^2),\quad \max(a^2,b^2,c^2)>\frac{1}{2}.
\end{equation}

Since the matrix $g$ and vectors $\p_1$ and $\p_2$ are invariant
under rotations around axis $z$ the same properties must have
Bloch vectors $\sv_1$ and $\sv_2$. There are two possibilities:

\medskip

i)Bloch vectors are unique and aligned with the axis $z$. The
solution given by Eq.(\ref{w.parz}) corresponds to this situation
and the resulting entanglement eigenvalue Eq.(\ref{w.max})
satisfies the inequality

\begin{equation}\label{w.slight}
\frac{1}{2}<\lm\leq1.
\end{equation}

ii)Bloch vectors have nonzero components in $xy$ plane and the
solution is not unique.  Eq.(\ref{w.par}) corresponds to this
situation and contains a free parameter. The free parameter is the
angle defining the direction of the vector $\mv$ in the $xy$
plane. Then Eq.(\ref{w.rad}) gives the entanglement eigenvalue in
highly entangled region

\begin{equation}\label{w.high}
\frac{4}{9}\leq\lm<\frac{1}{2}.
\end{equation}

Eq.(\ref{w.rad}) and (\ref{w.max}) have joint curves when
parameters $a,b,c$ form a right triangle and give $\lm=1/2$. The
GHZ states have same entanglement value and it seems to imply
something interesting. GHZ state can be used for teleportation and
superdense coding, but W-state cannot be. However, the W-type
state with right triangle coefficients can be used for
teleportation and superdense coding \cite{agr}. In other words,
both type of states can be applied provided they have the required
entanglement eigenvalue $\lm=1/2$.

\section{Symmetric States.}

Now let us consider the state which is symmetric under permutation
of qubits A and B and contains three real independent parameters

\begin{equation}\label{ghz.psi}
|\psi\ra=a|000\ra+b|111\ra+c|001\ra+d|110\ra,
\end{equation}

\noindent where $a^2+b^2+c^2+d^2=1$. According to Generalized
Schmidt Decomposition \cite{acin} the states with different sets
of parameters are local-unitary(LU) inequivalent. The relevant
quantities are

\begin{equation}\label{ghz.matr}
 \p_1=\p_2=r\,\nv,\quad g=
\begin{pmatrix}
\omega & 0 & 0\\
0 & -\omega & 0\\
0 & 0 & 1
\end{pmatrix}
,
\end{equation}

\noindent where

\begin{equation}\label{ghz.eig}
r=a^2+c^2-b^2-d^2,\quad \omega=2ad+2bc
\end{equation}

\noindent and the unit vector $\nv$ again is  aligned with the
axis $z$.

All three terms in the l.h.s. of Eq.(\ref{gen.s1s2}) are bounded
above:
\begin{itemize}

\item $\sv_1\cdot\p_1\leq|r|$,

\item $\sv_2\cdot\p_2\leq|r|$,

\item and owing to inequality $|\omega|\leq1,
\,g_{ij}\,s_{1i}s_{2j}\leq1$.

\end{itemize}

Quite surprisingly all upper limits are reached simultaneously at

\begin{equation}\label{ghz.sign}
\sv_1=\sv_2={\rm Sign}(r)\nv,
\end{equation}

\noindent which results in

\begin{equation}\label{ghz.mod}
\lm=\frac{1}{2}\left(1+|r|\right).
\end{equation}

This expression has a clear meaning. To understand it we
parameterize the state as

\begin{equation}\label{ghz.repsi}
|\psi\ra=k_1|00q_1\ra+k_2|11q_2\ra,
\end{equation}

\noindent where $q_1$ and $q_2$ are arbitrary single normalized
qubit states and positive parameters $k_1$ and $k_2$ satisfy
$k_1^2+k_2^2=1$. Then

\begin{equation}\label{ghz.l}
\lm=\max(k_1^2,k_2^2),
\end{equation}

\noindent i.e. the maximization takes a larger coefficient in
Eq.(\ref{ghz.repsi}). In bipartite case the maximization takes the
largest coefficient in Schmidt decomposition
\cite{biham-02,vidjon} and in this sense Eq.(\ref{ghz.repsi})
effectively takes the place of Schmidt decomposition. When
$|q_1\ra=|0\ra$  and $|q_2\ra=|1\ra$, Eq.(\ref{ghz.l}) gives the
known answer for generalized GHZ state \cite{wei-03,Shim-grov}.

The entanglement eigenvalue is minimal $\lm=1/2$ on condition that
$k_1=k_2$. These states can be described as follows

\begin{equation}\label{ghz.half}
|\psi\ra=|00q_1\ra+|11q_2\ra
\end{equation}

\noindent where $q_1$ and $q_2$ are arbitrary single qubit
normalized states. The entanglement eigenvalue is constant
$\lm=1/2$ and does not depend on single qubit state parameters.
Hence one may expect that all these states can be applied for
teleportation and superdense coding. It would be interesting to
check whether this assumption is correct or not.

It turns out that GHZ state is not a unique state and is one of
two-parametric LU inequivalent states that have $\lm=1/2$. On the
other hand W-state is unique up to LU transformations and the low
bound $\lm=4/9$ is reached if and only if $a=b=c$. However, one
cannot make such conclusions in general. Five real parameters are
necessary to parameterize the set of inequivalent three qubit pure
states \cite{acin}. And there is no explicit argument that W-state
is not just one of LU inequivalent states that have $\lm=4/9$.

\section{Summary.}

We have derived algebraic equations defining geometric measure of
three qubit pure states. These equations have a degree higher than
four and explicit solutions for general cases cannot be derived
analytically. However, the explicit expressions are not important.
Remember that explicit expressions for the algebraic equations of
degree three and four have a limited practical significance but
the equations itself are more important. This is especially true
for equations of higher degree; main results can be derived from
the equations rather than from the expressions of their roots.

Eq.(\ref{gen.eq}) give the nearest separable state directly and
this separable states have useful applications. In order to
construct an entanglement witness, for example, the crucial point
lies in finding the nearest separable state \cite{bert}. This will
be especially interesting for highly entangled states that have a
whole set of nearest separable states and allow to construct a set
of entanglement witnesses.

The expression in r.h.s. of Eq.(\ref{gen.s1s2}) can be maximized
directly for various three qubit states. Although it is very hard
to solve the higher-degree equation, it turns out that the wide
range of the three-qubit states have a symmetry and this symmetry
reduces the equations of degree six to the quadratic equations. In
this reason Eq.(\ref{gen.s1s2}) can be used to derive the analytic
expressions of the various entanglement measures for the
three-qubit states. Also Eq.(\ref{gen.s1s2}) can be a starting
point to explore the numerical computation of the entanglement
measures for the higher-qubit systems.


\chapter[Three-Qubit Groverian Measure]{Three-Qubit Groverian Measure}\label{groverian}

In this chapter we connect the geometric entanglement measure with
polynomial invariants in the case of three-quibt pure states
\cite{jung08-1}.

About decade ago the axioms which entanglement measures should
satisfy were studied~\cite{vedr-97}. The most important property
for measure is monotonicity under local operation and classical
communication(LOCC)~\cite{vidal98-1}. Following the axioms, many
entanglement measures were constructed such as relative
entropy\cite{plen01-1}, entanglement of
distillation\cite{ben-error} and
formation\cite{ben-conc,woot-98,form2,form4}, geometric
measure\cite{shim-95,barn-01,wei-03,pit}, Schmidt
measure\cite{eisert01} and Groverian measure\cite{biham-02}.
Entanglement measures are used in various branches of quantum
mechanics. Especially, recently, they are used to try to
understand Zamolodchikov's c-theorem\cite{zamo86} more profoundly.
It may be an important application of the quantum information
techniques to understand the effect of renormalization group in
field theories\cite{orus-08}.

The purpose of this paper is to compute the Groverian measure for
various three-qubit quantum states. The Groverian measure
$G(\psi)$ for three-qubit state $|\psi\rangle$ is defined by
$G(\psi) \equiv \sqrt{1 - P_{max}}$ where
\begin{equation}
\label{pmax1} P_{max} = \lm
\end{equation}
Thus $P_{max}$ can be interpreted as a maximal overlap between the
given state $|\psi\rangle$ and product states. Groverian measure
is an operational treatment of a geometric measure. Thus, if one
can compute $G(\psi)$, one can also compute the geometric measure
of pure state by $G^2(\psi)$. Sometimes it is more convenient to
re-express Eq.(\ref{pmax1}) in terms of the density matrix $\rho =
|\psi \rangle \langle \psi |$. This can be easily accomplished by
an expression
\begin{equation}
\label{pmax2} P_{max} = \max_{R^1,R^2,R^3} \mbox{Tr} \left[\rho
R^1 \otimes R^2 \otimes R^3 \right]
\end{equation}
where $R^i \equiv |q_i \rangle \langle q_i |$ density matrix for
the product state. Eq.(\ref{pmax1}) and Eq.(\ref{pmax2})
manifestly show that $P_{max}$ and $G(\psi)$ are local-unitary(LU)
invariant quantities. Since it is well-known that three-qubit
system has five independent LU-invariants\cite{acin,sud00,coff,pi},
{\it say} $J_i (i = 1, \cdots, 5)$, we would like to focus on the
relation of the Groverian measures to LU-invariants $J_i$'s in
this paper.

This chapter is organized as follows.

In Section 2.1 we review simple case, {\it i.e.} two-qubit system.
Using Bloch form of the density matrix it is shown in this section that
two-qubit system has only one independent LU-invariant quantity, {\it say}
$J$. It is also shown that Groverian measure and $P_{max}$ for
arbitrary two-qubit states can be expressed solely in terms of
$J$.

In Section 2.2 we have discussed how to derive
LU-invariants in higher-qubit systems. In fact, we have derived
many LU-invariant quantities using Bloch form of the density
matrix in three-qubit system. It is shown that all LU-invariants
derived can be expressed in terms of $J_i$'s discussed in
Ref.\cite{acin}. Recently, it was shown in Ref.\cite{reduced}
that $P_{max}$ for $n$-qubit state can be computed from
$(n-1)$-qubit reduced mixed state. This theorem was used in
Ref.\cite{analyt} and Ref.\cite{shared} to compute analytically
the geometric measures for various three-qubit states. In this
section we have discussed the physical reason why this theorem is
possible from the aspect of LU-invariance.

In {Section 2.3 we have computed the Groverian measures for
various types of the three-qubit system. The five types we discussed in
this section were originally developed in Ref.\cite{acin} for the
classification of the three-qubit states. It has been shown that
the Groverian measures for type 1, type 2, and type 3 can be
analytically computed. We have expressed all analytical results in
terms of LU-invariants $J_i$'s. For type 4 and type 5 the
analytical computation seems to be highly nontrivial and may need
separate publications. Thus the analytical calculation for these
types is not presented in this paper. The results of this section
are summarized in Table I.

In Section 2.4 we have
discussed the modified W-like state, which has three-independent
real parameters. In fact, this state cannot be categorized in the
five types discussed in Section 2.3. The analytic expressions of
the Groverian measure for this state was computed recently in
Ref.\cite{shared}. It was shown that the measure has three
different expressions depending on the domains of the parameter
space. It turned out that each expression has its own geometrical
meaning. In this section we have re-expressed all expressions of
the Groverian measure in terms of LU-invariants.

In Section 2.5 brief conclusion is given.

\section{Two Qubit: Simple Case}

In this section we consider $P_{max}$ for the two-qubit system.
The Groverian measure for two-qubit system is already
well-known\cite{Shim-grov}. However, we revisit this issue here to
explore how the measure is expressed in terms of the LU-invariant
quantities. The Schmidt decomposition\cite{schmidt1907,ek-sch} makes the
most general expression of the two-qubit state vector to be simple
form
\begin{equation}
\label{s1} |\psi \rangle = \lambda_0 |00\rangle + \lambda_1 |11
\rangle
\end{equation}
with $\lambda_0, \lambda_1 \geq 0$ and $\lambda_0^2 + \lambda_1^2
= 1$. The density matrix for $|\psi \rangle$ can be expressed in
the Bloch form as following:
\begin{equation}
\label{s2} \rho = |\psi \rangle \langle \psi | = \frac{1}{4}
\left[\openone \otimes \openone + v_{1 \alpha} \sigma_{\alpha}
\otimes \openone + v_{2 \alpha} \openone \otimes \sigma_{\alpha} +
g_{\alpha \beta} \sigma_{\alpha} \otimes \sigma_{\beta} \right],
\end{equation}
where
\begin{eqnarray}
\label{s3} \vec{v}_1 = \vec{v}_2 = \left(   \begin{array}{c}
                                      0                \\
                                      0                \\
                                  \lambda_0^2 - \lambda_1^2
                                  \end{array}
                                                         \right),
\hspace{1.0cm} g_{\alpha \beta} = \left(        \begin{array}{ccc}
                          2 \lambda_0 \lambda_1 &    0    &    0    \\
                               0    &   -2 \lambda_0 \lambda_1   &   0    \\
                               0    &    0    &    1
                                  \end{array}
                                                         \right).
\end{eqnarray}

In order to discuss the LU transformation we consider first the
quantity $U \sigma_{\alpha} U^{\dagger}$ where $U$ is $2 \times 2$
unitary matrix. With direct calculation one can prove easily
\begin{equation}
\label{lu1} U \sigma_{\alpha} U^{\dagger} = {\cal O}_{\alpha
\beta} \sigma_{\beta},
\end{equation}
where the explicit expression of ${\cal O}_{\alpha \beta}$ is
given in appendix A. Since ${\cal O}_{\alpha \beta}$ is a real
matrix satisfying ${\cal O} {\cal O}^T = {\cal O}^T {\cal O} =
\openone$, it is an element of the rotation group O(3). Therefore,
Eq.(\ref{lu1}) implies that the LU-invariants in the density
matrix (\ref{s2}) are $|\vec{v}_1|$, $|\vec{v}_2|$, $\mbox{Tr}[g
g^T]$ etc.

All LU-invariant quantities can be written in terms of one
quantity, {\it say} $J \equiv \lambda_0^2 \lambda_1^2$. In fact,
$J$ can be expressed in terms of two-qubit
concurrence\cite{woot-98} ${\cal C}$ by ${\cal C}^2 / 4$. Then it
is easy to show
\begin{eqnarray}
\label{s4} & &|\vec{v}_1|^2 = |\vec{v}_2|^2 = 1 - 4 J,    \\
\nonumber & & g_{\alpha \beta} g_{\alpha \beta} = 1 + 8 J.
\end{eqnarray}

It is well-known that $P_{max}$ is simply square of larger Schmidt
number in two-qubit case

\begin{equation}
\label{s7} P_{max} = \mbox{max} \left( \lambda_0^2, \lambda_1^2
\right).
\end{equation}

It can be re-expressed in terms of reduced density operators

\begin{equation}
\label{s5} P_{max} = \frac{1}{2} \left[1 + \sqrt{1 - 4 \mbox{det}
\rho^A} \right],
\end{equation}

\noindent where $\rho^A = \mbox{Tr}_B \rho = (1 + v_{1 \alpha}
\sigma_{\alpha})/2$. Since $P_{max}$ is invariant under
LU-transfor\-mation, it should be expressed in terms of
LU-invariant quantities. In fact, $P_{max}$ in Eq.(\ref{s5}) can
be re-written as
\begin{equation}
\label{s6} P_{max} = \frac{1}{2} \left[1 + \sqrt{1 - 4 J}\right].
\end{equation}
Eq.(\ref{s6}) implies that $P_{max}$ is manifestly LU-invariant.

\section{Local Unitary Invariants}
The Bloch representation of the $3$-qubit density matrix can be
written in the form
\begin{eqnarray}
\label{density1} \rho &=& \frac{1}{8} \Bigg[\openone \otimes
\openone \otimes \openone + v_{1 \alpha} \sigma_{\alpha}\otimes
\openone \otimes \openone + v_{2 \alpha} \openone \otimes
\sigma_{\alpha} \otimes \openone + v_{3 \alpha} \openone \otimes
\openone \otimes \sigma_{\alpha} \nonumber\\   & &+
h^{(1)}_{\alpha \beta} \openone \otimes \sigma_{\alpha}\otimes
\sigma_{\beta} + h^{(2)}_{\alpha \beta} \sigma_{\alpha}\otimes
\openone \otimes \sigma_{\beta} + h^{(3)}_{\alpha \beta}
\sigma_{\alpha}\otimes \sigma_{\beta} \otimes \openone
\nonumber\\
& &+ g_{\alpha \beta \gamma} \sigma_{\alpha}\otimes \sigma_{\beta}
\otimes \sigma_{\gamma} \Bigg], \vspace{1cm}
\end{eqnarray}
where $\sigma_{\alpha}$ is Pauli matrix. According to
Eq.(\ref{lu1}) and  appendix A it is easy to show that the
LU-invariants in the density matrix (\ref{density1}) are
$|\vec{v}_1|$, $|\vec{v}_2|$, $|\vec{v}_3|$, $\mbox{Tr}[h^{(1)}
h^{(1) T}]$, $\mbox{Tr}[h^{(2)} h^{(2) T}]$, $\mbox{Tr}[h^{(3)}
h^{(3) T}]$, $g_{\alpha \beta \gamma} g_{\alpha \beta \gamma}$
etc.

Few years ago Ac\'in et al\cite{acin} represented the three-qubit
arbitrary states in a simple form using a generalized Schmidt
decomposition\cite{schmidt1907} as following:
\begin{equation}
\label{state1} |\psi\rangle = \lambda_0 |000\rangle + \lambda_1
e^{i \varphi} |100\rangle + \lambda_2 |101\rangle + \lambda_3
|110\rangle + \lambda_4 |111\rangle
\end{equation}
with $\lambda_i \geq 0$, $0 \leq \varphi \leq \pi$, and $\sum_i
\lambda_i^2 = 1$. The five algebraically independent polynomial
LU-invariants were also constructed in Ref.\cite{acin}:
\begin{eqnarray}
\label{lu2} & &J_1 = \lambda_1^2 \lambda_4^2 + \lambda_2^2
\lambda_3^2 - 2 \lambda_1 \lambda_2 \lambda_3 \lambda_4 \cos
\varphi,
                                                      \\  \nonumber
& &J_2 = \lambda_0^2 \lambda_2^2, \hspace{1.0cm}
   J_3 = \lambda_0^2 \lambda_3^2, \hspace{1.0cm}
   J_4 = \lambda_0^2 \lambda_4^2,
                                                      \\  \nonumber
& &J_5 = \lambda_0^2 (J_1 + \lambda_2^2 \lambda_3^2 - \lambda_1^2
\lambda_4^2).
\end{eqnarray}

In order to determine how many states have the same values of the
invariants $J_1, J_2, ...J_5$, and therefore how many further
discrete-valued invariants are needed to specify uniquely a pure
state of three qubits up to local transformations, one would need
to find the number of different sets of parameters $\varphi$ and
$\lambda_i(i=0,1,...4)$, yielding the same invariants. Once
$\lambda_0$ is found, other parameters are determined uniquely and
therefore we derive an equation defining $\lambda_0$ in terms of
polynomial invariants.

\begin{equation}\label{added}
(J_1+J_4)\lambda_0^4-(J_5+J_4)\lambda_0^2+J_2J_3+J_2J_4+J_3J_4+J_4^2=0.
\end{equation}

This equation has at most two positive roots and consequently an
additional discrete-valued invariant is required to specify
uniquely a pure three qubit state. Generally 18 LU-invariants,
nine of which may be taken to have only discrete values, are
needed to determine a mixed 2-qubit state \cite{mixed}.

If one represents the density matrix $|\psi \rangle \langle \psi
|$ as a Bloch form like Eq.(\ref{density1}), it is possible to
construct $v_{1 \alpha}$, $v_{2 \alpha}$, $v_{3 \alpha}$,
$h^{(1)}_{\alpha \beta}$, $h^{(2)}_{\alpha \beta}$,
$h^{(3)}_{\alpha \beta}$, and $g_{\alpha \beta \gamma}$
explicitly, which are summarized in appendix B. Using these
explicit expressions one can show directly that all polynomial
LU-invariant quantities of pure states are expressed in terms of
$J_i$ as following:
\begin{eqnarray}
\label{lu3} & &|\vec{v}_1|^2 = 1 - 4 (J_2 + J_3 + J_4),
\hspace{1.0cm} |\vec{v}_2|^2 = 1 - 4 (J_1 + J_3 + J_4)
                                              \\   \nonumber
& &|\vec{v}_3|^2 = 1 - 4 (J_1 + J_2 + J_4), \hspace{1.0cm}
\mbox{Tr}[h^{(1)} h^{(1) T}] = 1 + 4 (2 J_1 - J_2 - J_3)
                                               \\   \nonumber
& &\mbox{Tr}[h^{(2)} h^{(2) T}] = 1 - 4 ( J_1 - 2 J_2 + J_3),
\hspace{.5cm} \mbox{Tr}[h^{(3)} h^{(3) T}] = 1 - 4 (J_1 + J_2 - 2
J_3)
                                                \\  \nonumber
& &g_{\alpha \beta \gamma} g_{\alpha \beta \gamma} = 1 + 4(2 J_1 +
2 J_2 + 2 J_3 + 3 J_4)
                                               \\   \nonumber
& &h^{(3)}_{\alpha \beta} v^{(1)}_{\alpha} v^{(2)}_{\beta} = 1 - 4
(J_1 + J_2 + J_3 + J_4 - J_5).
\end{eqnarray}

Recently, Ref.\cite{reduced} has shown that $P_{max}$ for
$n$-qubit pure state can be computed from $(n-1)$-qubit reduced
mixed state. This is followed from a fact
\begin{eqnarray}\label{theorem1}
 & &\max_{R^1,R^2,\cdots ,R^n} \mbox{Tr}\left[\rho R^1
\otimes R^2 \otimes \cdots \otimes R^n \right] = \\\nonumber
 & &\max_{R^1, R^2,\cdots ,R^{n-1}} \mbox{Tr}\left[\rho R^1 \otimes R^2 \otimes \cdots
\otimes R^{n-1} \otimes \openone                                                                   \right]
\end{eqnarray}
which is Theorem I of Ref.\cite{reduced}. Here, we would like to
discuss the physical meaning of Eq.(\ref{theorem1}) from the
aspect of LU-invariance. Eq.(\ref{theorem1}) in $3$-qubit system
reduces to
\begin{equation}
\label{pmax21} P_{max} = \max_{R^1,R^2} \mbox{Tr} \left[\rho^{AB}
R^1 \otimes R^2 \right]
\end{equation}
where $\rho^{AB} = \mbox{Tr}_C \rho$. From Eq.(\ref{density1})
$\rho^{AB}$ simply reduces to
\begin{equation}
\label{density21} \rho = \frac{1}{4} \left[\openone \otimes
\openone + v_{1 \alpha} \sigma_{\alpha} \otimes \openone + v_{2
\alpha} \openone \otimes \sigma_{\alpha} + h^{(3)}_{\alpha \beta}
\sigma_{\alpha} \otimes \sigma_{\beta} \right]
\end{equation}
where $ v_{1 \alpha}$, $v_{2 \alpha}$ and $ h^{(3)}_{\alpha
\beta}$ are explicitly given in appendix B. Of course, the
LU-invariant quantities of $\rho^{AB}$ are $|\vec{v}_1|$,
$|\vec{v}_2|$, $\mbox{Tr}[h^{(3)} h^{(3) T}]$, $h^{(3)}_{\alpha
\beta} v_{1 \alpha} v_{2 \beta}$ etc, all of which, of course, can
be re-expressed in terms of $J_1$, $J_2$, $J_3$, $J_4$ and $J_5$.
It is worthwhile noting that we need all $J_i$'s to express the
LU-invariant quantities of $\rho^{AB}$. This means that the
reduced state $\rho^{AB}$ does have full information on the
LU-invariance of the original pure state $\rho$.

Indeed, any reduced state resulting from a partial trace over a
single qubit uniquely determines any entanglement measure of
original system, given that the initial state is pure. Consider an
$(n-1)$-qubit reduced density matrix that can be purified by a
single qubit reference system. Let $|\psi^\prime\rangle$ be any
joint pure state. All other purifications can be obtained from the
state $|\psi^\prime\rangle$ by LU-transformations
$U\otimes\openone^{\otimes(n-1)}$, where $U$ is a local unitary
matrix acting on single qubit. Since any entanglement measure must
be invariant under LU-transformations, it must be same for all
purifications independently of $U$. Hence the reduced density
matrix determines any entanglement measure on the initial pure
state. That is why we can compute $P_{max}$ of $n$-qubit pure
state from the $(n-1)$-qubit reduced mixed state.

Generally, the information on the LU-invariance of the original
$n$-qubit state is partly lost if we take partial trace twice. In
order to show this explicitly let us consider $\rho^A \equiv
\mbox{Tr}_B \rho^{AB}$ and $\rho^B \equiv \mbox{Tr}_A \rho^{AB}$:
\begin{eqnarray}
\label{density22} \rho^A&=&\frac{1}{2} \left[ \openone + v_{1
\alpha} \sigma_{\alpha} \right]
                                                              \\   \nonumber
\rho^B&=&\frac{1}{2} \left[ \openone + v_{2 \alpha}
\sigma_{\alpha} \right].
\end{eqnarray}
Eq.(\ref{lu1}) and appendix A imply that their LU-invariant
quantities are only $|\vec{v}_1|$ and $|\vec{v}_2|$ respectively.
Thus, we do not need $J_5$ to express the LU-invariant quantities
of $\rho^A$ and $\rho^B$. This fact indicates that the mixed
states $\rho^A$ and $\rho^B$ partly loose the information of the
LU-invariance of the original pure state $\rho$. This is why
$(n-2)$-qubit reduced state cannot be used to compute $P_{max}$ of
$n$-qubit pure state.

\section{Calculation of $P_{max}$}

\subsection{General Feature}
If we insert the Bloch representation
\begin{equation}
\label{bloch41} R^1 = \frac{\openone + \vec{s}_1 \cdot
\vec{\sigma}}{2} \hspace{1.0cm} R^2 = \frac{\openone + \vec{s}_2
\cdot \vec{\sigma}}{2}
\end{equation}
with $|\vec{s}_1| = |\vec{s}_2| = 1$ into Eq.(\ref{pmax21}),
$P_{max}$ for $3$-qubit state becomes
\begin{equation}
\label{pmax41} P_{max} = \frac{1}{4} \max_{|\vec{s}_1| =
|\vec{s}_2| = 1} \left[1 + \vec{r}_1 \cdot \vec{s}_1 + \vec{r}_2
\cdot \vec{s}_2 +
     g_{i j} s_{1 i} s_{2 j} \right]
\end{equation}
where
\begin{eqnarray}
\label{def41} & &\vec{r}_1 = \mbox{Tr} \left[\rho^A \vec{\sigma}
\right]
                                                         \\   \nonumber
& &\vec{r}_2 = \mbox{Tr} \left[\rho^B \vec{\sigma} \right]
                                                         \\   \nonumber
& &g_{ij} = \mbox{Tr} \left[\rho^{AB} \sigma_i \otimes \sigma_j
\right].
\end{eqnarray}
Since in Eq.(\ref{pmax41}) $P_{max}$ is maximization with
constraint
 $|\vec{s}_1| = |\vec{s}_2| = 1$, we should use the Lagrange multiplier method,
which yields a pair of equations
\begin{eqnarray}
\label{algebraic1} & &\vec{r}_1 + g \vec{s}_2 = \Lambda_1
\vec{s}_1          \\   \nonumber & &\vec{r}_2 + g^T \vec{s}_1 =
\Lambda_2 \vec{s}_2,
\end{eqnarray}
where the symbol $g$ represents the matrix $g_{ij}$ in
Eq.(\ref{def41}). Thus we should solve $\vec{s}_1$, $\vec{s}_2$,
$\Lambda_1$ and $\Lambda_2$ by eq.(\ref{algebraic1}) and the
constraint $|\vec{s}_1| = |\vec{s}_2| = 1$. Although it is highly
nontrivial to solve Eq.(\ref{algebraic1}), sometimes it is not
difficult if the given $3$-qubit state $|\psi\rangle$ has rich
symmetries. Now, we would like to compute $P_{max}$ for various
types of $3$-qubit system.

\subsection{Type 1 (Product States): $J_1 = J_2 = J_3 = J_4 = J_5 = 0$}
In order for all $J_i$'s to be zero we have two cases $\lambda_0 =
J_1 = 0$ or $\lambda_2 = \lambda_3 = \lambda_4 = 0$.
\subsubsection{$\lambda_0 = J_1 = 0$}
If $\lambda_0 = 0$, $|\psi \rangle$ in Eq.(\ref{state1}) becomes
$|\psi \rangle = |1\rangle \otimes |BC \rangle$ where
\begin{equation}
\label{bc41} |BC\rangle = \lambda_1 e^{i \varphi} |00\rangle +
\lambda_2 |01\rangle + \lambda_3 |10\rangle + \lambda_4
|11\rangle.
\end{equation}
Thus $P_{max}$ for $|\psi\rangle$ equals to that for $|BC\rangle$.
Since $|BC\rangle$ is two-qubit state, one can easily compute
$P_{max}$ using Eq.(\ref{s5}), which is
\begin{equation}
\label{pmax42} P_{max} = \frac{1}{2} \left[1 + \sqrt{1 - 4
\mbox{det}
           \left(\mbox{Tr}_B |BC\rangle \langle BC| \right)} \right]
        = \frac{1}{2} \left[1 + \sqrt{1 - 4 J_1}\right].
\end{equation}
If, therefore, $\lambda_0 = J_1 = 0$, we have $P_{max} = 1$, which
gives a vanishing Groverian measure.
\subsubsection{$\lambda_2 = \lambda_3 = \lambda_4 = 0$}
In this case $|\psi \rangle$ in Eq.(\ref{state1}) becomes
\begin{equation}
\label{psi41} |\psi \rangle = \left(\lambda_0 |0\rangle +
\lambda_1 e^{i \varphi} |1\rangle\right) \otimes |0\rangle \otimes
|0\rangle.
\end{equation}
Since $|\psi \rangle$ is completely product state, $P_{max}$
becomes one.
\subsection{Type2a (biseparable states)}
In this type we have following three cases.
\subsubsection{$J_1 \neq 0$ and $J_2=J_3=J_4=J_5=0$}
In this case we have $\lambda_0 = 0$. Thus $P_{max}$ for this case
is exactly same with Eq.(\ref{pmax42}).
\subsubsection{$J_2 \neq 0$ and $J_1=J_3=J_4=J_5=0$}
In this case we have $\lambda_2 = \lambda_4 = 0$. Thus $P_{max}$
for $|\psi\rangle$ equals to that for $|AC\rangle$, where
\begin{equation}
\label{ac41} |AC\rangle = \lambda_0 |00\rangle + \lambda_1 e^{i
\varphi} |10\rangle + \lambda_2 |11\rangle.
\end{equation}
Using Eq.(\ref{s5}), therefore, one can easily compute $P_{max}$,
which is
\begin{equation}
\label{pmax43} P_{max} = \frac{1}{2} \left[1 + \sqrt{1 - 4
J_2}\right].
\end{equation}
\subsubsection{$J_3 \neq 0$ and $J_1=J_2=J_4=J_5=0$}
In this case $P_{max}$ for $|\psi\rangle$ equals to that for
$|AB\rangle$, where
\begin{equation}
\label{ab41} |AB\rangle = \lambda_0 |00\rangle + \lambda_1 e^{i
\varphi} |10\rangle + \lambda_3 |11\rangle.
\end{equation}
Thus $P_{max}$ for $|\psi\rangle$ is
\begin{equation}
\label{pmax44} P_{max} = \frac{1}{2} \left[1 + \sqrt{1 - 4
J_3}\right].
\end{equation}

\subsection{Type2b (generalized GHZ states): $J_4 \neq 0$, $J_1=J_2=J_3=J_5=0$}
In this case we have $\lambda_1 = \lambda_2 = \lambda_3 = 0$ and
$|\psi\rangle$ becomes
\begin{equation}
\label{ghz41} |\psi \rangle = \lambda_0 |000\rangle + \lambda_4
|111\rangle
\end{equation}
with $\lambda_0^2 + \lambda_4^2 = 1$. Then it is easy to show
\begin{eqnarray}
\label{def42} & &\vec{r}_1 = \mbox{Tr} \left[\rho^A \vec{\sigma}
\right] = (0, 0, \lambda_0^2 - \lambda_4^2)
                                                        \\   \nonumber
& &\vec{r}_2 = \mbox{Tr} \left[\rho^B \vec{\sigma} \right] = (0,
0, \lambda_0^2 - \lambda_4^2)
                                                        \\    \nonumber
& &g_{ij} = \mbox{Tr} \left[\rho^{AB} \sigma_i \otimes \sigma_j
\right] = \left(            \begin{array}{ccc}
                    0  &  0  &  0    \\
                    0  &  0  &  0    \\
                    0  &  0  &  1
                    \end{array}            \right).
\end{eqnarray}
Thus $P_{max}$ reduces to
\begin{equation}
\label{pmax45} P_{max} = \frac{1}{4} \max_{|\vec{s}_1| =
|\vec{s}_2| = 1} \left[1 + (\lambda_0^2 - \lambda_4^2) (s_{1 z} +
s_{2 z}) + s_{1 z} s_{2 z} \right].
\end{equation}
Since Eq.(\ref{pmax45}) is simple, we do not need to solve
Eq.(\ref{algebraic1}) for the maximization. If $\lambda_0 >
\lambda_4$, the maximization can be achieved by simply choosing
$\vec{s}_1 = \vec{s}_2 = (0, 0, 1)$. If $\lambda_0 < \lambda_4$,
we choose $\vec{s}_1 = \vec{s}_2 = (0, 0, -1)$. Thus we have
\begin{equation}
\label{pmax46} P_{max} = \mbox{max} (\lambda_0^2, \lambda_4^2).
\end{equation}

In order to express $P_{max}$ in Eq.(\ref{pmax46}) in terms of
LU-invariants we follow the following procedure. First we note
\begin{equation}
\label{pmax47} P_{max} = \frac{1}{2} \left[ (\lambda_0^2 +
\lambda_4^2) + |\lambda_0^2 - \lambda_4^2|
                                                                        \right].
\end{equation}
Since $|\lambda_0^2 - \lambda_4^2| = \sqrt{(\lambda_0^2 +
\lambda_4^2)^2 -
         4 \lambda_0^2 \lambda_4^2} = \sqrt{1 - 4 J_4}$, we get finally
\begin{equation}
\label{pmax48} P_{max} = \frac{1}{2} \left[1 + \sqrt{1 - 4
J_4}\right].
\end{equation}

\subsection{Type3a (tri-Bell states)}
In this case we have $\lambda_1 = \lambda_4 = 0$ and
$|\psi\rangle$ becomes
\begin{equation}
\label{w41} |\psi \rangle = \lambda_0 |000\rangle + \lambda_2
|101\rangle + \lambda_3 |110\rangle
\end{equation}
with $\lambda_0^2 + \lambda_2^2 + \lambda_3^2 = 1$. If we take
LU-transformation $\sigma_x$ in the first-qubit, $|\psi\rangle$ is
changed into $|\psi'\rangle$ which is usual W-type
state\cite{Chir} as follows:
\begin{equation}
\label{w42} |\psi'\rangle = \lambda_0 |100\rangle + \lambda_3
|010\rangle + \lambda_2 |001\rangle.
\end{equation}
The LU-invariants in this type are
\begin{eqnarray}
\label{lu41} & & J_1 = \lambda_2^2 \lambda_3^2 \hspace{1.0cm}
    J_2 = \lambda_0^2 \lambda_2^2
                                        \\   \nonumber
& & J_3 =  \lambda_0^2 \lambda_3^2 \hspace{1.0cm}
    J_5 = 2 \lambda_0^2 \lambda_2^2 \lambda_3^2.
\end{eqnarray}
Then it is easy to derive a relation
\begin{equation}
\label{rela41} J_1 J_2 + J_1 J_3 + J_2 J_3 = \sqrt{J_1 J_2 J_3} =
\frac{1}{2} J_5.
\end{equation}

Recently, $P_{max}$ for $|\psi'\rangle$ is computed analytically
in Ref.\cite{analyt} by solving the Lagrange multiplier equations
(\ref{algebraic1}) explicitly. In order to express $P_{max}$
explicitly we first define
\begin{eqnarray}
\label{def43} r_1&=& \lambda_3^2 + \lambda_2^2 - \lambda_0^2   \\
\nonumber r_2&=& \lambda_0^2 + \lambda_2^2 - \lambda_3^2   \\
\nonumber r_3&=& \lambda_0^2 + \lambda_3^2 - \lambda_2^2   \\
\nonumber \omega&=&2 \lambda_0 \lambda_3.
\end{eqnarray}
Also we define
\begin{eqnarray}
\label{def44} a&=&\mbox{max} (\lambda_0, \lambda_2, \lambda_3)
\\    \nonumber b&=&\mbox{mid} (\lambda_0, \lambda_2, \lambda_3)
\\    \nonumber c&=&\mbox{min} (\lambda_0, \lambda_2, \lambda_3).
\end{eqnarray}
Then $P_{max}$ is expressed differently in two different regions
as follows. If $a^2 \geq b^2 + c^2$, $P_{max}$ becomes
\begin{equation}
\label{pmax49} P_{max}^{>} = a^2 = \mbox{max}(\lambda_0^2,
\lambda_2^2, \lambda_3^2).
\end{equation}
In order to express $P_{max}$ in terms of LU-invariants we express
Eq.(\ref{pmax49}) differently as follows
\begin{eqnarray}\label{pmax50}
  & &P_{max}^{>} = \\ \nonumber
  & &\frac{1}{4} \left[(\lambda_0^2 +
\lambda_3^2 + \lambda_2^2)
       + |\lambda_0^2 + \lambda_3^2 - \lambda_2^2|
       + |\lambda_0^2 - \lambda_3^2 + \lambda_2^2|
       + |\lambda_0^2 - \lambda_3^2 - \lambda_2^2|  \right].
\end{eqnarray}
Using equalities
\begin{eqnarray}
\label{equal41} & & |\lambda_0^2 + \lambda_3^2 - \lambda_2^2| =
\sqrt{1 - 4 \lambda_0^2 \lambda_2^2 - 4 \lambda_2^2 \lambda_3^2} =
\sqrt{1 - 4(J_1 + J_2)}
                                                      \\   \nonumber
& & |\lambda_0^2 - \lambda_3^2 + \lambda_2^2| = \sqrt{1 - 4
\lambda_0^2 \lambda_3^2 - 4 \lambda_2^2 \lambda_3^2} = \sqrt{1 -
4(J_1 + J_3)}
                                                      \\   \nonumber
& & |\lambda_0^2 - \lambda_3^2 - \lambda_2^2| = \sqrt{1 - 4
\lambda_0^2 \lambda_2^2 - 4 \lambda_0^2 \lambda_3^2} = \sqrt{1 -
4(J_2 + J_3)},
\end{eqnarray}
we can express $P_{max}$ in Eq.(\ref{pmax49}) as follows:
\begin{eqnarray}\label{pmax51}
 & &P_{max}^{>} = \\ \nonumber
  & &\frac{1}{4} \left[1 + \sqrt{1 - 4(J_1
+ J_2)} + \sqrt{1 - 4(J_1 + J_3)} +  \sqrt{1 - 4(J_2 + J_3)}
                                                   \right].
\end{eqnarray}

If $a^2 \leq b^2 + c^2$, $P_{max}$ becomes
\begin{equation}
\label{pmax52} P_{max}^{<} = \frac{1}{4} \left[1 + \frac{\omega
\sqrt{(\omega^2 + r_1^2 - r_3^2)(\omega^2 + r_2^2 - r_3^2)}
                - r_1 r_2 r_3}
               {\omega^2 - r_3^2}          \right].
\end{equation}
It was shown in Ref.\cite{analyt} that $P_{max} = 4 R^2$, where
$R$ is a circumradius of the triangle $\lambda_0$, $\lambda_2$ and
$\lambda_3$. When $a^2 \leq b^2 + c^2$, one can show easily $r_1 =
\sqrt{1 - 4(J_2 + J_3)}$, $r_2 = \sqrt{1 - 4(J_1 + J_3)}$, $r_3 =
\sqrt{1 - 4(J_1 + J_2)}$, and $\omega = 2 \sqrt{J_3}$. Using
$\omega^2 - r_3^2 - r_1 r_2 r_3 = 8 \lambda_0^2
\lambda_2^2\lambda_3^2$, One can show easily that $P_{max}$ in
Eq.(\ref{pmax52}) in terms of LU-invariants becomes
\begin{equation}
\label{pmax53} P_{max}^{<} = \frac{4 \sqrt{J_1 J_2 J_3}}{4 (J_1 +
J_2 + J_3) - 1}.
\end{equation}

Let us consider $\lambda_0 = 0$ limit in this type. Then we have
$J_2 = J_3 = 0$. Thus $P_{max}^{>}$ reduces to $(1/2) (1 + \sqrt{1
- 4 J_1})$ which exactly coincides with Eq.(\ref{pmax42}). By same
way one can prove that Eq.(\ref{pmax51}) has correct limits to
various other types.

\subsection{Type3b (extended GHZ states)}
This type consists of $3$ types, {\it i.e.} $\lambda_1 =
\lambda_2=0$, $\lambda_1=\lambda_3=0$ and $\lambda_2=\lambda_3=0$.

\subsubsection{$\lambda_1 = \lambda_2 = 0$}
In this case the state (\ref{state1}) becomes
\begin{equation}
\label{eghz41} |\psi \rangle = \lambda_0 |000\rangle + \lambda_3
|110\rangle + \lambda_4 |111\rangle
\end{equation}
with $\lambda_0^2 + \lambda_3^2 + \lambda_4^2 = 1$. The
non-vanishing LU-invariants are
\begin{equation}
\label{eghzlu1} J_3 = \lambda_0^2 \lambda_3^2, \hspace{1.0cm} J_4
= \lambda_0^2 \lambda_4^2.
\end{equation}
Note that $J_3 + J_4$ is expressed in terms of solely $\lambda_0$
as
\begin{equation}
\label{eghzlu2} J_3 + J_4 = \lambda_0^2 (1 - \lambda_0^2).
\end{equation}

Eq.(\ref{eghz41}) can be re-written as
\begin{equation}
\label{eghz42} |\psi \rangle = \lambda_0 |00q_1 \rangle + \sqrt{1
- \lambda_0^2} |11q_2 \rangle
\end{equation}
where $|q_1\rangle = |0\rangle$ and $|q_2\rangle = (1 / \sqrt{1 -
\lambda_0^2}) (\lambda_3 |0\rangle + \lambda_4|1\rangle)$ are
normalized one qubit states. Thus, from Ref.\cite{analyt},
$P_{max}$ for $|\psi \rangle$ is
\begin{equation}
\label{eghzpmax1} P_{max} = \mbox{max} \left(\lambda_0^2, 1 -
\lambda_0^2 \right) = \frac{1}{2} \left[1 + \sqrt{(1 - 2
\lambda_0^2)^2} \right].
\end{equation}
With an aid of Eq.(\ref{eghzlu2}) $P_{max}$ in
Eq.(\ref{eghzpmax1}) can be easily expressed in terms of
LU-invariants as following:
\begin{equation}
\label{eghzpmax2} P_{max} = \frac{1}{2} \left[1 + \sqrt{1 - 4 (J_3
+ J_4)} \right].
\end{equation}
If we take $\lambda_3 = 0$ limit in this type, we have $J_3 = 0$,
which makes Eq.(\ref{eghzpmax2}) to be $(1/2) (1 + \sqrt{1 - 4
J_4})$. This exactly coincides with Eq.(\ref{pmax48}).

\subsubsection{$\lambda_1 = \lambda_3 = 0$}
In this case $|\psi\rangle$ and LU-invariants are
\begin{equation}
\label{eghz43} |\psi\rangle = \lambda_0 |0 q_1 0\rangle + \sqrt{1
- \lambda_0^2} |1 q_2 1\rangle
\end{equation}
and
\begin{equation}
\label{eghzlu3} J_2 = \lambda_0^2 \lambda_2^2, \hspace{1.0cm} J_4
= \lambda_0^2 \lambda_4^2
\end{equation}
where $|q_1\rangle = |0\rangle$, $|q_2\rangle = (1 /\sqrt{1 -
\lambda_0^2})(\lambda_2 |0\rangle + \lambda_4 |1\rangle)$, and
$\lambda_0^2 + \lambda_2^2 + \lambda_4^2=1$. The same method used
in the previous subsection easily yields
\begin{equation}
\label{eghzpmax3} P_{max} = \frac{1}{2} \left[1 + \sqrt{1 - 4 (J_2
+ J_4)} \right].
\end{equation}
One can show that Eq.(\ref{eghzpmax3}) has correct limits to other
types.

\subsubsection{$\lambda_2 = \lambda_3 = 0$}
In this case $|\psi\rangle$ and LU-invariants are
\begin{equation}
\label{eghz44} |\psi\rangle = \sqrt{1 - \lambda_4^2} |q_1 0
0\rangle +  \lambda_4 |q_2 1 1\rangle
\end{equation}
and
\begin{equation}
\label{eghzlu4} J_1 = \lambda_1^2 \lambda_4^2, \hspace{1.0cm} J_4
= \lambda_0^2 \lambda_4^2
\end{equation}
where $|q_1\rangle = (1 / \sqrt{1 - \lambda_4^2}) (\lambda_0
|0\rangle + \lambda_1 e^{i \varphi} |1\rangle)$, $|q_2\rangle =
|1\rangle$, and $\lambda_0^2 + \lambda_1^2 + \lambda_4^2=1$. It is
easy to show
\begin{equation}
\label{eghzpmax4} P_{max} = \frac{1}{2} \left[1 + \sqrt{1 - 4 (J_1
+ J_4)} \right].
\end{equation}
One can show that Eq.(\ref{eghzpmax4}) has correct limits to other
types.

\subsection{Type4a ($\lambda_4 = 0$)}
In this case the state vector $|\psi \rangle$ in Eq.(\ref{state1})
reduces to
\begin{equation}
\label{4a1} |\psi\rangle = \lambda_0 |000\rangle + \lambda_1 e^{i
\varphi} |100\rangle + \lambda_2 |101 \rangle + \lambda_3
|110\rangle
\end{equation}
with $\lambda_0^2 + \lambda_1^2 + \lambda_2^2 + \lambda_3^2 = 1$.
The non-vanishing LU-invariants are
\begin{eqnarray}
\label{4alu1} & &J_1 = \lambda_2^2 \lambda_3^2 \hspace{1.0cm} J_2
= \lambda_0^2 \lambda_2^2
                                                              \\   \nonumber
& &J_3 = \lambda_0^2 \lambda_3^2 \hspace{1.0cm} J_5 = 2
\lambda_0^2 \lambda_2^2 \lambda_3^2.
\end{eqnarray}
From Eq.(\ref{4alu1}) it is easy to show
\begin{equation}
\label{4alu2} \sqrt{J_1 J_2 J_3} = \frac{1}{2} J_5.
\end{equation}

The remarkable fact deduced from Eq.(\ref{4alu1}) is that the
non--vanishing LU--invariants are independent of the phase factor
$\varphi$. This indicates that the Groverian measure for
Eq.(\ref{4a1}) is also independent of $\varphi$

In order to compute $P_{max}$ analytically in this type, we should
solve the Lagrange multiplier equations (\ref{algebraic1}) with
\begin{eqnarray}
\label{rrg1} & &\vec{r}_1 = \mbox{Tr} [\rho^A \vec{\sigma}] = (2
\lambda_0 \lambda_1 \cos \varphi, 2 \lambda_0 \lambda_1 \sin
\varphi, 2 \lambda_0^2 - 1)
                                                              \\   \nonumber
& &\vec{r}_2 = \mbox{Tr} [\rho^B \vec{\sigma}] = (2 \lambda_1
\lambda_3 \cos \varphi,
 -2\lambda_1 \lambda_3 \sin \varphi, 1 - 2 \lambda_3^2)
                                                              \\   \nonumber
& & g_{ij} = \mbox{Tr}[\rho^{AB} \sigma_i \otimes \sigma_j] =
\left(                    \begin{array}{ccc}
      2 \lambda_0 \lambda_3       &     0     &    2 \lambda_0 \lambda_1 \cos \varphi  \\
         0     &    -2 \lambda_0 \lambda_3   &   2 \lambda_0 \lambda_1 \sin \varphi   \\
       -2 \lambda_1 \lambda_3 \cos \varphi  &  2 \lambda_1 \lambda_3 \sin \varphi  &
        \lambda_0^2 - \lambda_1^2 - \lambda_2^2 + \lambda_3^2
                          \end{array}
                                                                  \right).
\end{eqnarray}
Although we have freedom to choose the phase factor $\varphi$, it
is impossible to find singular values of the matrix $g$, which
makes it formidable task to solve Eq.(\ref{algebraic1}). Based on
Ref.\cite{analyt} and Ref.\cite{shared}, furthermore, we can
conjecture that $P_{max}$ for this type may have several different
expressions depending on the domains in parameter space.
Therefore, it may need long calculation to compute $P_{max}$
analytically. We would like to leave this issue for our future
research work and the explicit expressions of $P_{max}$ are not
presented in this paper.

\subsection{Type4b}
This type consists of the $2$ cases, {\it i.e.} $\lambda_2=0$ and
$\lambda_3=0$.

\subsubsection{$\lambda_2=0$}
In this case the state vector $|\psi \rangle$ in Eq.(\ref{state1})
reduces to
\begin{equation}
\label{4b1} |\psi\rangle = \lambda_0 |000\rangle + \lambda_1 e^{i
\varphi} |100\rangle + \lambda_3 |110 \rangle + \lambda_4
|111\rangle
\end{equation}
with $\lambda_0^2 + \lambda_1^2 + \lambda_3^2 + \lambda_4^2 = 1$.
The LU-invariants are
\begin{equation}
\label{4blu1} J_1 = \lambda_1^2 \lambda_4^2 \hspace{.5cm} J_3 =
\lambda_0^2 \lambda_3^2 \hspace{.5cm} J_4 = \lambda_0^2
\lambda_4^2.
\end{equation}
Eq.(\ref{4blu1}) implies that the Groverian measure for
Eq.(\ref{4b1}) is independent of the phase factor $\varphi$ like
type 4a. This fact may drastically reduce the calculation
procedure for solving the Lagrange multiplier equation
(\ref{algebraic1}). In spite of this fact, however, solving
Eq.(\ref{algebraic1}) is highly non-trivial as we commented in the
previous type. The explicit expressions of the Groverian measure
are not presented in this paper and we hope to present them
elsewhere in the near future.

\subsubsection{$\lambda_3=0$}
In this case the state vector $|\psi \rangle$ in Eq.(\ref{state1})
reduces to
\begin{equation}
\label{4b2} |\psi\rangle = \lambda_0 |000\rangle + \lambda_1 e^{i
\varphi} |100\rangle + \lambda_2 |101 \rangle + \lambda_4
|111\rangle
\end{equation}
with $\lambda_0^2 + \lambda_1^2 + \lambda_2^2 + \lambda_4^2 = 1$.
The LU-invariants are
\begin{equation}
\label{4blu2} J_1 = \lambda_1^2 \lambda_4^2 \hspace{.5cm} J_2 =
\lambda_0^2 \lambda_2^2 \hspace{.5cm} J_4 = \lambda_0^2
\lambda_4^2.
\end{equation}
Eq.(\ref{4blu2}) implies that the Groverian measure for
Eq.(\ref{4b2}) is independent of the phase factor $\varphi$ like
type 4a.

\subsection{Type4c ($\lambda_1=0$)}
In this case the state vector $|\psi \rangle$ in Eq.(\ref{state1})
reduces to
\begin{equation}
\label{4c1} |\psi\rangle = \lambda_0 |000\rangle + \lambda_2
|101\rangle + \lambda_3 |110 \rangle + \lambda_4 |111\rangle
\end{equation}
with $\lambda_0^2 + \lambda_2^2 + \lambda_3^2 + \lambda_4^2 = 1$.
The LU-invariants in this type are
\begin{eqnarray}
\label{4clu1} & &J_1 = \lambda_2^2 \lambda_3^2 \hspace{.5cm}
   J_2 = \lambda_0^2 \lambda_2^2 \hspace{.5cm}
   J_3 = \lambda_0^2 \lambda_3^2
                                             \\   \nonumber
& &J_4 = \lambda_0^2 \lambda_4^2 \hspace{.5cm}
   J_5 = 2 \lambda_0^2 \lambda_2^2 \lambda_3^2.
\end{eqnarray}
From Eq.(\ref{4clu1}) it is easy to show
\begin{equation}
\label{4clu2} J_1 (J_2 + J_3 + J_4) + J_2 J_3 = \sqrt{J_1 J_2 J_3}
= \frac{1}{2} J_5.
\end{equation}

In this type $\vec{r}_1$, $\vec{r}_2$ and $g_{ij}$ defined in
Eq.(\ref{def41}) are
\begin{eqnarray}
\label{4ccom1} & &\vec{r}_1 = (0, 0, 2 \lambda_0^2 - 1)        \\
\nonumber & &\vec{r}_2 = (2 \lambda_2 \lambda_4, 0, \lambda_0^2 +
\lambda_2^2 - \lambda_3^3 - \lambda_4^2)
                                                               \\   \nonumber
& &g_{ij} = \left(             \begin{array}{ccc}
                   2 \lambda_0 \lambda_3  &  0  &  0      \\
                   0  &  -2 \lambda_0 \lambda_3  &  0     \\
                   -2 \lambda_2 \lambda_4  &  0  &  1 - 2 \lambda_2^2
                               \end{array}                            \right).
\end{eqnarray}
Like type 4a and type 4b solving Eq.(\ref{algebraic1}) is highly
non-trivial mainly due to non-diagonalization of $g_{ij}$. Of
course, the fact that the first component of $\vec{r}_2$ is
non-zero makes hard to solve Eq.(\ref{algebraic1}) too. The
explicit expressions of the Groverian measure in this type are not
given in this paper.

\subsection{Type5 (real states): $\varphi = 0$, $\pi$}

\subsubsection{$\varphi = 0$}
In this case the state vector $|\psi \rangle$ in Eq.(\ref{state1})
reduces to
\begin{equation}
\label{501} |\psi\rangle = \lambda_0 |000\rangle + \lambda_1
|100\rangle + \lambda_2 |101 \rangle + \lambda_3 |110 \rangle +
\lambda_4 |111\rangle
\end{equation}
with $\lambda_0^2 + \lambda_1^2 + \lambda_2^2 + \lambda_3^2 +
\lambda_4^2 = 1$. The LU-invariants in this case are
\begin{eqnarray}
\label{500lu1} & &J_1 = (\lambda_2 \lambda_3 - \lambda_1
\lambda_4)^2 \hspace{.5cm}
   J_2 = \lambda_0^2 \lambda_2^2 \hspace{.5cm}
   J_3 = \lambda_0^2 \lambda_3^2
                                             \\   \nonumber
& &J_4 = \lambda_0^2 \lambda_4^2 \hspace{.5cm}
   J_5 = 2 \lambda_0^2 \lambda_2 \lambda_3 (\lambda_2 \lambda_3 - \lambda_1 \lambda_4).
\end{eqnarray}
It is easy to show $\sqrt{J_1 J_2 J_3} = J_5 / 2$.

\subsubsection{$\varphi = \pi$}
In this case the state vector $|\psi \rangle$ in Eq.(\ref{state1})
reduces to
\begin{equation}
\label{502} |\psi\rangle = \lambda_0 |000\rangle - \lambda_1
|100\rangle + \lambda_2 |101 \rangle + \lambda_3 |110 \rangle +
\lambda_4 |111\rangle
\end{equation}
with $\lambda_0^2 + \lambda_1^2 + \lambda_2^2 + \lambda_3^2 +
\lambda_4^2 = 1$. The LU-invariants in this case are
\begin{eqnarray}
\label{500lu2} & &J_1 = (\lambda_2 \lambda_3 + \lambda_1
\lambda_4)^2 \hspace{.5cm}
   J_2 = \lambda_0^2 \lambda_2^2 \hspace{.5cm}
   J_3 = \lambda_0^2 \lambda_3^2
                                             \\   \nonumber
& &J_4 = \lambda_0^2 \lambda_4^2 \hspace{.5cm}
   J_5 = 2 \lambda_0^2 \lambda_2 \lambda_3 (\lambda_2 \lambda_3 + \lambda_1 \lambda_4).
\end{eqnarray}
It is easy to show $\sqrt{J_1 J_2 J_3} = J_5 / 2$ in this type.

The analytic calculation of $P_{max}$ in type 5 is most difficult
problem. In addition, we don't know whether it is mathematically
possible or not. However, the geometric interpretation of
$P_{max}$ presented in Ref.\cite{analyt} and Ref.\cite{shared} may
provide us valuable insight. We hope to leave this issue for our
future research work too. The results in this section is
summarized in Table I.

\begin{center}
\begin{tabular}{c|c|c|c}  \hline
\multicolumn{2}{c|} {Type} & conditions & $P_{max}$  \\ \hline
\hline \multicolumn{2}{c|} {Type I}    &  $J_i = 0$  & $1$
\\ \hline
{}  &   {}   & $J_i = 0$ except $J_1$  &  $\frac{1}{2} \left(1 + \sqrt{1 - 4 J_1} \right)$  \\
                                                                               \cline{3-4}
Type II  &  a  & $J_i = 0$ except $J_2$  & $\frac{1}{2} \left(1 + \sqrt{1 - 4 J_2} \right)$ \\
                                                                               \cline{3-4}
{}   &  {}  & $J_i = 0$ except $J_3$   & $\frac{1}{2} \left(1 + \sqrt{1 - 4 J_3} \right)$ \\
                                                                               \cline{2-4}
{}   &  b  &  $J_i = 0$ except $J_4$  &  $\frac{1}{2} \left(1 + \sqrt{1 - 4 J_4} \right)$ \\
                                                                               \hline
{}  &  {}  &  {}  &  $\scriptstyle \left(1 + \sqrt{1 - 4(J_1 +
J_2)} + \sqrt{1 - 4(J_1 + J_3)} +  \sqrt{1 - 4(J_2 + J_3)}
\right)/4$
                                                             \\

{}  &  a  &  $\lambda_1 = \lambda_4 = 0$  &
                                                           if $a^2 \geq b^2 + c^2$  \\
                                                                                \cline{4-4}
{}  & {}  &  {}  &  $4 \sqrt{J_1 J_2 J_3} / \left(4(J_1 + J_2 +
J_3) - 1 \right)$
                                                          \\

{}  & {}  &  {}  &  {}
                                                        if $a^2 \leq b^2 + c^2$  \\
                                                                                \cline{2-4}
Type III  & {} & $\lambda_1 = \lambda_2 = 0$ & $\frac{1}{2}
\left(1 +
                                                 \sqrt{1 - 4 (J_3 + J_4)} \right)$  \\
                                                                                \cline{3-4}
{} & b  &  $\lambda_1 = \lambda_3 = 0$ & $\frac{1}{2} \left(1 +
                                                 \sqrt{1 - 4 (J_2 + J_4)} \right)$  \\
                                                                                \cline{3-4}
{} & {} & $\lambda_2 = \lambda_3 = 0$ & $\frac{1}{2} \left(1 +
                                                 \sqrt{1 - 4 (J_1 + J_4)} \right)$  \\
                                                                                \hline
{} & a & $\lambda_4 = 0$ & independent of $\varphi$: not presented  \\
                                                                                \cline{2-4}
Type IV & b & $\lambda_2 = 0$ & independent of $\varphi$: not presented  \\
                                                                                \cline{3-4}
{} & {} &  $\lambda_3 = 0$  &  independent of $\varphi$: not presented  \\
                                                                                \cline{2-4}
{} & c  & $\lambda_1 = 0$  &  not presented  \\
                                                                                \hline
\multicolumn{2}{c|} {Type V} & $\varphi = 0$ & not presented  \\
                                                                                \cline{3-4}
\multicolumn{2}{c|} {} & $\varphi = \pi$ & not presented         \\
                                                                                \hline

\end{tabular}

\vspace{0.1cm} Table I: Summary of $P_{max}$ in various types.
\end{center}
\vspace{0.5cm}

\section{New Type}

\subsection{standard form}
In this section we consider new type in $3$-qubit states. The type
we consider is
\begin{equation}
\label{newstate1}
|\Phi\rangle=a|100\rangle+b|010\rangle+c|001\rangle+q|111\rangle,\quad
a^2+b^2+c^2+q^2=1.
\end{equation}
First, we would like to derive the standard form like
Eq.(\ref{state1}) from $|\Phi\rangle$. This can be achieved as
following. First, we consider LU-transformation of $|\Phi\rangle$,
{\it i.e.} $(U \otimes \openone \otimes \openone) |\Phi\rangle$,
where
\begin{eqnarray}
\label{newunitary1} U = \frac{1}{\sqrt{a q + bc}} \left(
\begin{array}{cc}
               \sqrt{a q} e^{i \theta}   &  \sqrt{b c} e^{i \theta}  \\
               -\sqrt{b c}               &   \sqrt{a q}
                  \end{array}                             \right).
\end{eqnarray}
After LU-transformation, we perform Schmidt decomposition
following Ref.\cite{acin}. Finally we choose $\theta$ to make all
$\lambda_i$ to be positive. Then we can derive the standard form
(\ref{state1}) from $|\Phi \rangle$ with $\varphi = 0$ or $\pi$,
and
\begin{eqnarray}
\label{relacoe1} & &\lambda_0 = \sqrt{\frac{(ac + bq) (ab +
cq)}{aq + bc}}  \\  \nonumber & &\lambda_1 =
\frac{\sqrt{abcq}}{\sqrt{(ab + cq) (ac + bq) (aq + bc)}}
               |a^2 + q^2 - b^2 - c^2|                    \\   \nonumber
& &\lambda_2 = \frac{1}{\lambda_0} |ac - bq|              \\
\nonumber & &\lambda_3 = \frac{1}{\lambda_0} |ab - cq|
\\   \nonumber & &\lambda_4 = \frac{2 \sqrt{abcq}}{\lambda_0}.
\end{eqnarray}
It is easy to prove that the normalization condition $a^2 + b^2 +
c^2 + q^2 = 1$ guarantees the normalization
\begin{equation}
\label{newnormal1} \lambda_0^2 + \lambda_1^2 + \lambda_2^2 +
\lambda_3^2 + \lambda_4^2 = 1.
\end{equation}
Since $|\Phi\rangle$ has three free parameters, we need one more
constraint between $\lambda_i$'s. This additional constraint can
be derived by trial and error. The explicit expression for this
additional relation is
\begin{equation}
\label{additional1} \lambda_0^2 (\lambda_2^2 + \lambda_3^2 +
\lambda_4^2 ) = \frac{1}{4} - \frac{\lambda_1^2}{\lambda_4^2}
(\lambda_2^2 + \lambda_4^2) (\lambda_3^2 + \lambda_4^2).
\end{equation}
Since all $\lambda_i$'s are not vanishing but there are only three
free parameters, $|\Phi \rangle$ is not involved in the types
discussed in the previous section.

\subsection{LU-invariants}
Using Eq.(\ref{relacoe1}) it is easy to derive LU-invariants which
are
\begin{eqnarray}
\label{newlu-1} & &J_1 = (\lambda_1 \lambda_4 - \lambda_2
\lambda_3)^2
                           = \frac{1}{(ab + cq)^2 (ac + bq)^2}    \\   \nonumber
& & \hspace{1.5cm} \times \left[ 2 abcq |a^2 + q^2 - b^2 - c^2| -
(aq + bc) |(ab - cq) (ac - bq)| \right]^2
                                                                  \\  \nonumber
& & J_2 = \lambda_0^2 \lambda_2^2 = (ac - bq)^2
\\   \nonumber & & J_3 = \lambda_0^2 \lambda_3^2 = (ab - cq)^2
\\   \nonumber & & J_4 = \lambda_0^2 \lambda_4^2 = 4 abcq
\\   \nonumber & & J_5 = \lambda_0^2 \left( J_1 + \lambda_2^2
\lambda_3^2 - \lambda_1^2 \lambda_4^2
                                                   \right).
\end{eqnarray}
One can show directly that $J_5 = 2 \sqrt{J_1 J_2 J_3}$. Since
$|\Phi \rangle$ has three free parameters, there should exist
additional relation between $J_i$'s. However, the explicit
expression may be hardly derived. In principle, this constraint
can be derived as following. First, we express the coefficients
$a$, $b$, $c$, and $q$ in terms of $J_1$, $J_2$, $J_3$ and $J_4$
using first four equations of Eq.(\ref{newlu-1}). Then the
normalization condition $a^2 + b^2 + c^2 + q^2 = 1$ gives explicit
expression of this additional constraint. Since, however, this
procedure requires the solutions of quartic equation, it seems to
be hard to derive it explicitly.

Since $J_1$ contains absolute value, it is dependent on the
regions in the parameter space. Direct calculation shows that
$J_1$ is

\begin{equation}\label{newj1a}
J_1 = (aq - bc)^2,
\end{equation}
when $(a^2 + q^2 - b^2 - c^2) (ab - cq) (ac - bq) \geq 0$ and
\begin{equation}\label{newj1b}
J_1 = (aq - bc)^2  \left[1 + 2 \frac{(ab - cq) (ac - bq) (aq + bc)
}{ (ab + cq) (ac + bq) (aq - bc) } \right]^2,
\end{equation}
when $(a^2 + q^2 - b^2 - c^2) (ab - cq) (ac - bq) < 0$.

Since $P_{max}$ is manifestly LU-invariant quantity, it is obvious
that it also depends on the regions on the parameter space.

\subsection{calculation of $ P_{max} $}

$P_{max}$ for state $|\Phi \rangle$ in Eq.(\ref{newstate1}) has
been analytically computed recently in Ref.\cite{shared}. It turns
out that $P_{max}$ is differently expressed in three distinct
ranges of definition in parameter space. The final expressions can
be interpreted geometrically as discussed in Ref.\cite{shared}. To
express $P_{max}$ explicitly we define
\begin{eqnarray}
\label{share1} & &r_1 \equiv b^2 + c^2 - a^2 - q^2  \hspace{1.0cm}
   r_2 \equiv a^2 + c^2 - b^2 - q^2
                                                \\   \nonumber
& &r_3 \equiv a^2 + b^2 - c^2 - q^2  \hspace{1.0cm} \omega \equiv
ab + qc  \hspace{1.0cm} \mu \equiv ab - qc.
\end{eqnarray}

The first expression of $P_{max}$, which can be expressed in terms
of circumradius of convex quadrangle is
\begin{equation}
\label{share2} P_{max}^{(Q)} = \frac{4(ab + qc) (ac + qb) (aq +
bc)}{4 \omega^2 - r_3^2}.
\end{equation}
The second expression of $P_{max}$, which can be expressed in
terms of circumradius of crossed-quadrangle is
\begin{equation}
\label{shhare3} P_{max}^{(CQ)} = \frac{(ab - cq) (ac - bq) (bc -
aq)}{4 S_x^2}
\end{equation}
where
\begin{equation}
\label{share4} S_x^2 = \frac{1}{16} (a + b + c + q) (a + b - c -
q) (a - b + c - q) (-a + b + c - q).
\end{equation}
The final expression of $P_{max}$ corresponds to the largest
coefficient:
\begin{equation}
\label{share5} P_{max}^{(L)} = \mbox{max} (a^2, b^2, c^2, q^2)
              = \frac{1}{4} \left(1 + |r_1| + |r_2| + |r_3| \right).
\end{equation}
The applicable domain for each $P_{max}$ is fully discussed in
Ref.\cite{shared}.

Now we would like to express all expressions of $P_{max}$ in terms
of LU-invariants. For the simplicity we choose a simplified case,
that is $(a^2 + q^2 - b^2 - c^2) (ab - cq) (ac - bq) \geq 0$. Then
it is easy to derive
\begin{eqnarray}
\label{share6} & &r_1^2 = 1 - 4(J_2 + J_3 + J_4)  \hspace{1.0cm}
r_2^2 = 1 - 4(J_1 + J_3 + J_4)
                                                                 \\   \nonumber
& &r_3^2 = 1 - 4(J_1 + J_2 + J_4)  \hspace{1.0cm} \omega^2 = J_3 +
J_4.
\end{eqnarray}
Then it is simple to express $P_{max}^{(Q)}$ and $P_{max}^{(CQ)}$
as following:
\begin{eqnarray}
\label{share7} & &P_{max}^{(Q)} = \frac{4 \sqrt{(J_1 + J_4) (J_2 +
J_4) (J_3 + J_4)}}
                        {4(J_1 + J_2 + J_3 + 2 J_4) - 1}
                                                          \\  \nonumber
& &P_{max}^{(CQ)} = \frac{4 \sqrt{J_1 J_2 J_3}}{4 (J_1 + J_2 + J_3
+ J_4) - 1}.
\end{eqnarray}
If we take $q=0$ limit, we have $\lambda_4 = J_4 = 0$. Thus
$P_{max}^{(Q)}$ and $P_{max}^{(CQ)}$ reduce to $4 \sqrt{J_1 J_2
J_3} / (4 (J_1 + J_2 + J_3) - 1)$, which exactly coincides with
$P_{max}^<$ in Eq.(\ref{pmax53}). Finally Eq.(\ref{share6}) makes
$P_{max}^{(L)}$ to be
\begin{eqnarray}
\label{share8} P_{max}^{(L)} &=& \frac{1}{4} \bigg( 1 + \sqrt{1 -
4(J_2 + J_3 + J_4)} \nonumber \\ &+&
                \sqrt{1 - 4(J_1 + J_3 + J_4)} +
                              \sqrt{1 - 4(J_1 + J_2 + J_4)} \bigg).
\end{eqnarray}
One can show that $P_{max}^{(L)}$ equals to $P_{max}^{>}$ in
Eq.(\ref{pmax51}) when $q=0$. This indicates that our results
(\ref{share7}) and (\ref{share8}) have correct limits to other
types of three-qubit system.

\section{Conclusion}

We tried to compute the Groverian measure analytically in the
various types of three-qubit system. The types we considered in
this paper are given in Ref.\cite{acin} for the classification of
the three-qubit system.

For type 1, type 2 and type 3 the Groverian measures are
analytically computed. All results, furthermore, can be
represented in terms of LU-invariant quantities. This reflects the
manifest LU-invariance of the Groverian measure.

For type 4 and type 5 we could not derive the analytical
expressions of the measures because the Lagrange multiplier
equations (\ref{algebraic1}) is highly difficult to solve.
However, the consideration of LU-invariants indicates that the
Groverian measure in type 4 should be independent of the phase
factor $\varphi$. We expect that this fact may drastically
simplify the calculational procedure for obtaining the analytical
results of the measure in type 4. The derivation in type 5 is most
difficult problem. However, it might be possible to get valuable
insight from the geometric interpretation of $P_{max}$, presented
in Ref.\cite{analyt} and Ref.\cite{shared}. We would like to
revisit type 4 and type 5 in the near future.

We think that the most important problem in the research of
entanglement is to understand the general properties of
entanglement measures in arbitrary qubit systems. In order to
explore this issue we would like to extend, as a next step, our
calculation to four-qubit states. In addition, the Groverian
measure for four-qubit pure state is related to that for two-qubit
mixed state via purification\cite{shapira06}. Although general
theory for entanglement is far from complete understanding at
present stage, we would like to go toward this direction in the
future.


\chapter[Shared Quantum States]{Geometric measure of entanglement and shared quantum states}\label{shared-1}

\renewcommand{\p}{{\bm r}}
\renewcommand{\s}{{\bm \sigma}}
\renewcommand{\tr}{\mathrm{tr}}
\renewcommand{\ra}{\rangle}
\renewcommand{\la}{\langle}
\renewcommand{\sv}{{\bm v}}
\renewcommand{\iv}{{\bm i}}
\renewcommand{\lm}{\Lambda_{\max}^2}
\newcommand{\su}{{\bm u}}
\newcommand{\jv}{{\bm j}}
\newcommand{\kv}{{\bm k}}

In this chapter we present the first calculation of the geometric
measure of entanglement for generic three qubit states which are
expressed as linear combinations of four orthogonal product
states~\cite{shared}.

Any pure three qubit state can be written in terms of five
preassigned orthogonal product states \cite{acin} via Schmidt
decomposition. Thus the states discussed here are more general
states compared to the well-known GHZ \cite{ghz} and W \cite{Chir}
states.

The progress made to date allows oneself to calculate the
geometric measure of entanglement for pure three qubit systems
\cite{reduced}. The basic idea is to use $(n-1)$-qubit mixed
states to calculate the geometric measure of $n$-qubit pure
states. In the case of three qubits this idea converts the task
effectively into the maximization of the two-qubit mixed state
over product states and yields linear eigenvalue equations
\cite{analyt}. The solution of these linear eigenvalue equations
reduces to the root finding for algebraic equations of degree six.
However, three-qubit states containing symmetries allow complete
analytical solutions and explicit expressions as the symmetry
reduces the equations of degree six to the quadratic equations.
Analytic expressions derived in this way are unique and the
presented effective method can be applied for extended quantum
systems. Our aim is to derive analytic expressions for a wider
class of three qubit systems and in this sense this work is the
continuation of Ref.\cite{analyt}.

We consider most general three qubit states that allow to derive
analytic expressions for entanglement eigenvalue. These states can
be expressed as linear combinations of four given orthogonal
product states. If any of coefficients in this expansion vanishes,
then one obtains the states analyzed in \cite{analyt}. Notice that
arbitrary linear combinations of five product states \cite{acin}
give a couple of algebraic equations of degree six. Hence
\'Evariste Galois's theorem does not allow to get analytic
expressions for these states except some particular cases.

We derive analytic expressions for an entanglement eigenvalue.
Each expression has its own applicable domain depending on state
parameters and these applicable domains are split up by separating
surfaces. Thus the geometric measure distinguishes different types
of states depending on the corresponding applicable domain. States
that lie on separating surfaces are shared by two types of states
and acquire new features.

This chapter is organized as follows.
In Section 3.1 we derive stationarity equations and their
solutions. In Section 3.2 we specify three qubit states
under consideration and find relevant quantities. In
Section 3.3 we calculate entanglement eigenvalues and
present explicit expressions. In Section 3.4 we separate
the validity domains of the derived expressions. In
Section 3.5 we discuss shared states. In Section
3.6 we make concluding remarks.

\section{Stationarity equations}

In this section we briefly review the derivation of the
stationarity equations  and their general solutions \cite{analyt}.
Denote by $\rho^{ABC}$ the density matrix of the three-qubit pure
state and define the entanglement eigenvalue $\lm$ \cite{wei-03}

\begin{equation}\label{gen.pmax.shared}
\lm=\max_{\varrho^1\varrho^2\varrho^3}
\tr\left(\rho^{ABC}\varrho^1\otimes\varrho^2\otimes\varrho^3\right),
\end{equation}

\noindent where the maximization runs over all normalized complete
product states. Theorem 1 of Ref.\cite{reduced} states that the
maximization of a pure state over a single qubit state can be
completely derived by using  a particle traced over density
matrix. Hence the theorem allows us to re-express the entanglement
eigenvalue by reduced density matrix $\rho^{AB}$ of qubits A and B

\begin{equation}\label{gen.pred.shared}
\lm=\max_{\varrho^1\varrho^2}
\tr\left(\rho^{AB}\varrho^1\otimes\varrho^2\right).
\end{equation}

Now we introduce four Bloch vectors:

\smallskip

1)\,$\p_A$ for the reduced density matrix $\rho^A$ of the qubit A,

2)\,$\p_B$ for the reduced density matrix $\rho^B$ of the qubit B,

3)\,$\su$ for the single qubit state $\varrho^1$,

4)\,$\sv$ for the single qubit state $\varrho^2$.

\smallskip

Then the expression for entanglement eigenvalue
(\ref{gen.pred.shared}) takes the form

\begin{equation}\label{gen.s1s2.shared}
\lm=\frac{1}{4}\max_{u^2=v^2=1}\left(1+\su\cdot \p_A+\sv\cdot
\p_B+g_{ij}\,u_iv_j\right),
\end{equation}

\noindent where(summation on repeated indices $i$ and $j$ is
understood)

\begin{equation}\label{gen.vec.shared}
g_{ij}=\tr(\rho^{AB}\sigma_i\otimes\sigma_j)
\end{equation}

\noindent and $\sigma_i$'s are Pauli matrices. The closest product
state satisfies the stationarity conditions

\begin{equation}\label{gen.eq.shared}
\p_A+g\sv=\lambda_1\su,\quad\p_B+g^T\su=\lambda_2\sv,
\end{equation}

\noindent where Lagrange multipliers $\lambda_1$ and $\lambda_2$
enforce the unit Bloch vectors $\su$ and $\sv$. The solutions of
Eq.(\ref{gen.eq.shared}) are

\begin{equation}\label{gen.sol.shared}
\su=\left(\lambda_1\lambda_2\openone-g\,g^T\right)^{-1}
\left(\lambda_2\p_A+g\,\p_B\right),\quad
\sv=\left(\lambda_1\lambda_2\openone-g^Tg\right)^{-1}
\left(\lambda_1\p_B+g^T\p_A\right).
\end{equation}

\noindent Unknown Lagrange multipliers are defined by equations

\begin{equation}\label{gen.alg.shared}
u^2=1,\quad v^2=1.
\end{equation}

In general, Eq.(\ref{gen.alg.shared}) gives algebraic equations of
degree six. The reason for this is that stationarity equations
define all extremes of the reduced density matrix $\rho^{AB}$ over
product states, regardless of them being global or local. And the
degree of the algebraic equations is the number of possible
extremes.

Eq.(\ref{gen.sol.shared}) contains valuable information. It
provides solid bases for a new numerical approach. This can be
compared with the numerical calculations based on other technique
\cite{Shim-grov}.

\section{Three Qubit State}

We consider a four-parameter state

\begin{equation}\label{w.psi.shared}
|\psi\ra=a|100\ra+b|010\ra+c|001\ra+d|111\ra,
\end{equation}

\noindent where free parameters $a,b,c,d$ satisfy the
normalization condition $a^2+b^2+c^2+d^2=1$. Without loss of
generality we consider only the case of positive parameters
$a,b,c,d$. At first sight, it is not obvious whether the state
allows analytic solutions or not. However, it does and our first
task is to confirm the existence of the analytic solutions.

In fact, entanglement of the state Eq.(\ref{w.psi.shared}) is
invariant under the permutations of four parameters $a,b,c,d$. The
invariance under the permutations of three parameters $a,b,c$ is
the consequence of the invariance under the permutations of qubits
A,B,C. Now we make a local unitary(LU) transformation that
relabels the bases of qubits B and C, i.e.
$0_B\leftrightarrow1_B,\; 0_C\leftrightarrow1_C$, and does not
change the basis of qubit A. This LU-transformation interchanges
the coefficients as follows: $a\leftrightarrow d,\;
b\leftrightarrow c$. Since any entanglement measure must be
invariant under LU-transformations and the permutation
$b\leftrightarrow c$, it must be also invariant under the
permutation $a\leftrightarrow d$. In view of this symmetry, any
entanglement measure must be invariant under the permutations of
all the state parameters $a,b,c,d$. Owing to this symmetry, the
state allows to derive analytic expressions for the entanglement
eigenvalues. The necessary condition is \cite{analyt}

\begin{equation}\label{w.det}
\det\left(\lambda_1\lambda_2\openone-g\,g^T\right)=0.
\end{equation}

Indeed, if the condition (\ref{w.det}) is fulfilled, then the
expressions (\ref{gen.sol.shared}) for the general solutions are
not applicable and Eq.(\ref{gen.eq.shared}) admits further
simplification.

Denote by $\iv,\jv,\kv$ unit vectors along axes $x,y,z$
respectively. Straightforward calculation yields

\begin{equation}\label{w.matr.shared}
 \p_A=r_1\,\kv,\quad\p_B=r_2\,\kv,\quad g=
\begin{pmatrix}
2\omega & 0 & 0\\
0 & 2\mu & 0\\
0 & 0 & -r_3
\end{pmatrix}
,
\end{equation}
where
\begin{eqnarray}\label{w.eig.shared}
 & & r_1=b^2+c^2-a^2-d^2,\quad r_2=a^2+c^2-b^2-d^2,
\\\nonumber
 & & r_3=a^2+b^2-c^2-d^2,\quad \omega=ab+dc,\quad\mu=ab-dc.
\end{eqnarray}

Vectors $\su$ and $\sv$ can be written as linear combinations

\begin{equation}\label{w.ijk}
\su=u_i\iv+u_j\jv+u_k\kv,\quad \sv=v_i\iv+v_j\jv+v_k\kv
\end{equation}

\noindent of vectors $\iv,\jv,\kv$. The substitution of the
Eq.(\ref{w.ijk}) into Eq.(\ref{gen.eq.shared}) gives a couple of
equations in each direction. The result is a system of six linear
equations

\begin{subequations}\label{w.sub}
\begin{equation}\label{w.sub1}
2\omega\,v_i=\lambda_1u_i,\quad2\omega\,u_i=\lambda_2v_i,
\end{equation}
\begin{equation}\label{w.sub2}
2\mu\,v_j=\lambda_1u_j,\quad2\mu\,u_j=\lambda_2v_j,
\end{equation}
\begin{equation}\label{w.sub3}
r_1-r_3 v_k=\lambda_1u_k,\quad r_2-r_3 u_k=\lambda_2v_k.
\end{equation}
\end{subequations}

Above equations impose two conditions

\begin{subequations}\label{w.imp}
\begin{equation}\label{w.impom}
(\lambda_1\lambda_2-4\omega^2)u_iv_i=0,
\end{equation}
\begin{equation}\label{w.impmu}
(\lambda_1\lambda_2-4\mu^2)u_jv_j=0.
\end{equation}
\end{subequations}

From these equations it can be deduced that the condition
(\ref{w.det}) is valid and the system of equations
(\ref{gen.eq.shared}) and (\ref{gen.alg.shared}) is solvable. Note
that as a consequences of Eq.(\ref{w.sub}) $x$ and/or $y$
components of vectors $\su$ and $\sv$ vanish simultaneously.
Hence, conditions (\ref{w.imp}) are satisfied in following three
cases:

\begin{itemize}

\item vectors $\su$ and $\sv$ lie in $xz$ plane
\begin{equation}\label{w.vers1}
\lambda_1\lambda_2-4\omega^2=0,\quad u_jv_j=0,
\end{equation}

\item vectors $\su$ and $\sv$ lie in $yz$ plane
\begin{equation}\label{w.vers2}
\lambda_1\lambda_2-4\mu^2=0,\quad u_iv_i=0,
\end{equation}

\item vectors $\su$ and $\sv$ are aligned with axis $z$
\begin{equation}\label{w.vers3}
u_iv_i=u_jv_j=0.
\end{equation}

\end{itemize}

These cases are examined individually in next section.

\section{Explicit expressions}

In this section we analyze all three cases and derive explicit
expressions for entanglement eigenvalue. Each expression has its
own range of definition in which they are deemed applicable. Three
ranges of definition cover the four dimensional sphere given by
normalization condition. It is necessary to separate the validity
domains and to make clear which of expressions should be applied
for a given state. It turns out that the separation of domains
requires solving inequalities that contain polynomials of degree
six. This is a nontrivial task and we investigate it in the next
section.

\subsection{Circumradius of Convex Quadrangle}

Let us consider the first case. Our main task is to find Lagrange
multipliers $\lambda_1$ and $\lambda_2$. From equations
(\ref{w.sub3}) and (\ref{w.vers1}) we have

\begin{equation}\label{cq.sz}
u_k=\frac{\lambda_2r_1-r_2r_3}{4\omega^2-r_3^2},\quad
v_k=\frac{\lambda_1r_2-r_1r_3}{4\omega^2-r_3^2}.
\end{equation}

In its turn Eq.(\ref{w.sub1}) gives

\begin{equation}\label{cq.sx}
\lambda_1u_i^2=\lambda_2v_i^2.
\end{equation}

Eq.(\ref{gen.alg.shared}) allows the substitution of expressions
(\ref{cq.sz}) into Eq.(\ref{cq.sx}). Then we can obtain the second
equation for Lagrange multipliers

\begin{equation}\label{cq.con}
\lambda_1\left(4\omega^2+r_2^2-r_3^2 \right)=
\lambda_2\left(4\omega^2+r_1^2-r_3^2 \right).
\end{equation}

This equation has a simple form owing to condition (\ref{w.det}).
Thus we can factorize the equation of degree six into the
quadratic equations. Equations (\ref{cq.con}) and (\ref{w.vers1})
together yield

\begin{equation}\label{cq.lam-om}
\lambda_1=2\omega\,\frac{bc+ad}{ac+bd},\quad
\lambda_2=2\omega\,\frac{ac+bd}{bc+ad}.
\end{equation}

Note that we kept only positive values of Lagrange multipliers and
omitted negative values to get the maximal value of $\lm$. Now
Eq.(\ref{gen.s1s2.shared}) takes the form

\begin{equation}\label{cq.eigprim1}
4\lm=1+\frac{8(ab+cd)(ac+bd)(ad+bc)-r_1r_2r_3}{4\omega^2-r_3^2}.
\end{equation}

In fact, entanglement eigenvalue is the sum of two equal terms and
this statement follows from the identity

\begin{equation}\label{cq.iden}
1-\frac{r_1r_2r_3}{4\omega^2-r_3^2}=
8\frac{(ab+cd)(ac+bd)(ad+bc)}{4\omega^2-r_3^2}.
\end{equation}

\noindent To derive this identity one has to use the normalization
condition $a^2+b^2+c^2+d^2 = 1$. The identity allows to rewrite
Eq.(\ref{cq.eigprim1}) as follows

\begin{equation}\label{cq.eigfin}
\lm=4R_q^2,
\end{equation}

\noindent where

\begin{equation}\label{cq.circ}
R_q^2=\frac{(ab+cd)(ac+bd)(ad+bc)}{4\omega^2-r_3^2}.
\end{equation}

Above formula has a geometric interpretation and now we
demonstrate it. Let us define a quantity $p \equiv (a+b+c+d)/2$.
Then the denominator can be rewritten as

\begin{equation}\label{cq.her}
4\omega^2-r_3^2=16(p-a)(p-b)(p-c)(p-d).
\end{equation}

Five independent parameters are necessary to construct a convex
quadrangle. However, four independent parameters are necessary to
construct a convex quadrangle that has circumradius. For such
quadrangles the area $S_q$ is given exactly by Eq.(\ref{cq.her})
up to numerical factor, that is $S^2_q=(p-a)(p-b)(p-c)(p-d)$.
Hence Eq.(\ref{cq.circ}) can be rewritten as

\begin{equation}\label{cq.radius}
R_q^2=\frac{(ab+cd)(ac+bd)(ad+bc)}{16S^2_q}.
\end{equation}

\noindent Thus $R_q$ can be interpreted as a  circumradius of the
convex quadrangle. Eq.(\ref{cq.radius}) is the generalization of
the corresponding formula of Ref.\cite{analyt} and reduces to the
circumradius of the triangle if one of parameters is zero.

Eq.(\ref{cq.eigfin}) is valid if vectors $\su$ and $\sv$ are unit
and have non-vanishing $x$ components. These conditions have short
formulations

\begin{equation}\label{cq.ineq}
|u_k|\leq1,\quad|v_k|\leq1.
\end{equation}

Above inequalities are polynomials of degree six and algebraic
solutions are unlikely. However, it is still possible do define
the domain of validity of Eq.(\ref{cq.radius}).

\subsection{Circumradius of Crossed-Quadrangle}

Here, we consider the second case given by Eq.(\ref{w.vers2}).
Derivations repeat steps of the previous subsection and the only
difference is the interchange $\omega\leftrightarrow\mu$.
Therefore we skip some obvious steps and present only main
results. Components of vectors $\su$ and $\sv$ along axis $z$ are

\begin{equation}\label{cr.sz}
u_k=\frac{\lambda_2r_1-r_2r_3}{4\mu^2-r_3^2},\quad
v_k=\frac{\lambda_1r_2-r_1r_3}{4\mu^2-r_3^2}.
\end{equation}

The second equation for Lagrange multipliers

\begin{equation}\label{cr.con}
\lambda_1\left(4\mu^2+r_2^2-r_3^2 \right)=
\lambda_2\left(4\mu^2+r_1^2-r_3^2 \right)
\end{equation}

together with Eq.(\ref{w.vers2}) yields

\begin{equation}\label{cr.lam-mu}
\lambda_1=\pm2\mu\,\frac{bc-ad}{ac-bd},\quad
\lambda_2=\pm2\mu\,\frac{ac-bd}{bc-ad}.
\end{equation}

\noindent Using these expressions, one can derive the following
expression for entanglement eigenvalue

\begin{equation}\label{cr.eigprim2}
4\lm=
1+\frac{\lambda_2(4\mu^2+r_1^2-r_3^2)-r_1r_2r_3}{4\mu^2-r_3^2}.
\end{equation}

Now the restrictions $1/4<\lm\leq1$ derived in Ref.\cite{reduced}
uniquely define the signs in Eq.(\ref{cr.lam-mu}). Right signs
enforce strictly positive fraction in right hand side of
Eq.(\ref{cr.eigprim2}). To make a right choice, we replace $d$ by
$-d$ in the identity (\ref{cq.iden}) and rewrite
Eq.(\ref{cr.eigprim2}) as follows

\begin{eqnarray}\label{cr.eigprim3}
4\lm=\frac{1}{2}\,\frac{(ac-bd)(bc-ad)(ab-cd)}
{p(p-c-d)(p-b-d)(p-a-d)}
\nonumber\\\pm\frac{1}{2}\,\frac{(ac-bd)(bc-ad)(ab-cd)}
{p(p-c-d)(p-b-d)(p-a-d)}.
\end{eqnarray}


Lower sign yields zero and is wrong. It shows that reduced density
matrix $\rho^{AB}$ still has zero eigenvalue.

Upper sign may yield a true answer. Entanglement eigenvalue is

\begin{equation}\label{cr.eigfin}
\lm=4R_\times^2,
\end{equation}

\noindent where

\begin{equation}\label{cr.radius}
R_\times^2=\frac{(ac-bd)(bc-ad)(ab-cd)}{16S^2_\times},
\end{equation}

\noindent and $S^2_\times=p(p-c-d)(p-b-d)(p-a-d)$. The formula
(\ref{cr.radius}) may seem suspicious because it is not clear
whether right hand side is positive and lies in required region.
To clarify the situation we present a geometrical treatment of
Eq.(\ref{cr.radius}).

The geometrical figure $ABCD$ in Fig.1A is not a quadrangle and is
not a polygon at all. The reason is that it has crossed sides $AD$
and $BC$. We call figure $ABCD$ crossed-quadrangle in a figurative
sense as it has four sides and a cross point. Another
justification of this term is that we will compare figure $ABCD$
in Fig.1A with a convex quadrangle $ABCD$ containing the same
sides.

\begin{figure}
\includegraphics[width=10cm]{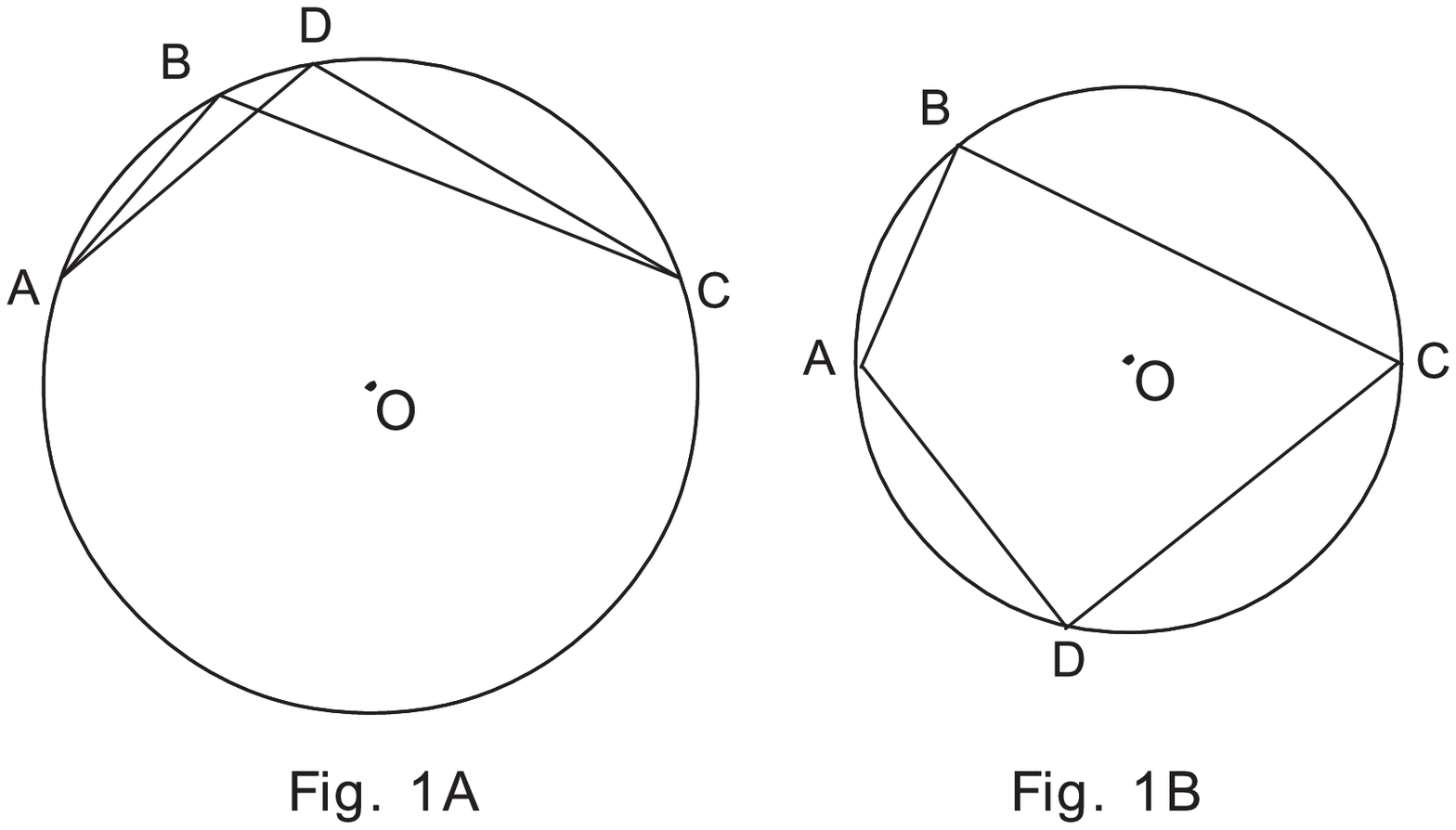}
\caption{\label{cr.eps}This figure shows the example for the case
when crossed quadrangle(Fig.1A) has larger circumradius than that
of convex quadrangle(Fig.1B) with same sides.}
\end{figure}

Consider a crossed-quadrangle $ABCD$ with sides
$AB=a,BC=b,CD=c,DA=d$ that has circumcircle. It is easy to find
the length of the interval $AC$

\begin{equation}\label{cr.ab}
AC^2=\frac{(ac-bd)(bc-ad)}{ab-cd}.
\end{equation}

This relation is true unless triangles $ABC$ and $ADC$ have the
same height and as a consequence equal areas. Note that $S_\times$
is not an area of the crossed-quadrangle. It is the difference
between the areas of the noted triangles.

Using Eq.(\ref{cr.ab}), one can derive exactly
Eq.(\ref{cr.radius}) for the circumradius of the
crossed-quadrangle.

Eq.(\ref{cr.eigfin}) is meaningful if vectors $\su$ and $\sv$ are
unit and have nonzero components along the axis $y$.

\subsection{Largest Coefficient}

In this subsection we consider the last case described by
Eq.(\ref{w.vers3}). Entanglement eigenvalue takes maximal value if
all terms in r.h.s. of Eq.(\ref{gen.s1s2.shared}) are positive.
Then equations (\ref{w.vers3}) and (\ref{w.matr.shared}) together
impose

\begin{equation}\label{lc.vec}
\su={\rm Sign}(r_1)\kv,\quad\sv={\rm Sign}(r_2)\kv,\quad
r_1r_2r_3<0,
\end{equation}

\noindent where Sign(x) gives -1, 0 or 1 depending on whether x is
negative, zero, or positive. Substituting these values into
Eq.(\ref{gen.s1s2.shared}), we obtain

\begin{equation}\label{lc.3}
\lm=\frac{1}{4}\left(1+|r_1|+|r_2|+|r_3|\right).
\end{equation}

Owing to inequality, $r_1r_2r_3<0$, above expression always gives
a square of the largest coefficient $l$

\begin{equation}\label{lc.l}
l=\max(a,b,c,d)
\end{equation}

\noindent in Eq.(\ref{w.psi.shared}). Indeed, let us consider the
case $r_1>0,r_2>0,r_3<0$. From inequalities $r_1>0,r_2>0$ it
follows that $c^2>d^2+|a^2-b^2|$ and therefore $c^2>d^2$. Note,
$c^2>d^2$ is necessary but not sufficient condition. Now if $d>b$,
then $r_1>0$ yields $c>a$ and if $d<b$, then $r_3<0$ yields $c>a$.
Thus inequality $c>a$ is true in all cases. Similarly $c>b$ and
$c$ is the largest coefficient. On the other hand $\lm=c^2$ and
Eq.(\ref{lc.3}) really gives the largest coefficient in this case.

Similarly, cases $r_1>0,r_2<0,r_3>0$ and $r_1<0,r_2>0,r_3>0$ yield
$\lm=b^2$ and $\lm=a^2$, respectively. And again entanglement
eigenvalue takes the value of the largest coefficient.

The last possibility $r_1<0,r_2<0,r_3<0$ can be analyzed using
analogous speculations. One obtains $\lm=d^2$ and $d$ is the
largest coefficient.

Combining all cases mentioned earlier, we rewrite Eq.(\ref{lc.3})
as follows

\begin{equation}\label{lc.larg}
\lm=l^2.
\end{equation}

This expression is valid if both vectors $\su$ and $\sv$ are
collinear with the axes $z$.

\bigskip

We have derived three expressions for
(\ref{cq.eigfin}),(\ref{cr.eigfin}) and (\ref{lc.larg}) for
entanglement eigenvalue. They are valid when vectors $\su$ and
$\sv$ lie in $xz$ plane, lie in $yz$ plane and are collinear with
axis $z$, respectively. The following section goes on to specify
these domains by parameters $a,b,c,d$.

\section{Applicable Domains}

Mainly, two points are being analyzed. First, we probe into the
meaningful geometrical interpretations of quantities $R_q$ and
$R_\times$. Second, we separate validity domains of equations
(\ref{cq.eigfin}),(\ref{cr.eigfin}) and (\ref{lc.larg}). It is
mentioned earlier that algebraic methods for solving the
inequalities of degree six are ineffective. Hence, we use
geometric tools that are elegant and concise in this case.

We consider four parameters $a,b,c,d$ as free parameters as the
normalization condition is irrelevant here. Indeed, one can use
the state $|\psi\ra/\sqrt{a^2+b^2+c^2+d^2}$ where all parameters
are free. If one repeats the same steps, the only difference is
that the entanglement eigenvalue $\lm$ is replaced by
$\lm/(a^2+b^2+c^2+d^2)$. In other words, normalization condition
re-scales the quadrangle, convex or crossed, so that the
circumradius always lies in the required region. Consequently, in
constructing quadrangles we can neglect the normalization
condition and consider four free parameters $a,b,c,d$.

\subsection{Existence of circumcircle.}

It is known that four sides $a,b,c,d$ of the convex quadrangle
must obey the inequality $p-l>0$. Any set of such parameters forms
a cyclic quadrilateral. Note that the quadrangle is not unique as
the sides can be arranged in different orders. But all these
quadrangles have the same circumcircle and the circumradius is
unique.

The sides of a crossed-quadrangle must obey the same condition.
Indeed, from Fig.1A it follows that $BC-AB<AC<AD+DC$ and
$DC-AD<AC<AB+BC$. Therefore $AB+AD+DC>BC$ and $AB+BC+AD>DC$. The
sides $BC$ and $DC$ are two largest sides and consequently
$p-l>0$. However, the existence of the circumcircle requires an
additional condition and it is explained here. The relation
$r_3=2\mu\cos ABC$ forces $4\mu^2\geq r_3^2$ and, therefore

\begin{equation}\label{ad.posden}
S_\times^2\geq0.
\end{equation}

\noindent Thus the denominator in Eq.(\ref{cr.radius}) must be
positive. On the other hand the inequality $AC^2\geq0$ forces a
positive numerator of the same fraction

\begin{equation}\label{ad.posnum}
(ac-bd)(bc-ad)(ab-cd)\geq0.
\end{equation}

These two inequalities impose conditions on parameters $a,b,c,d$.
For the future considerations, we need to write explicitly the
condition imposed by inequality (\ref{ad.posnum}). The numerator
is a symmetric function on parameters $a,b,c,d$ and it suffices to
analyze only the case $a\geq b\geq c\geq d$. Obviously
$(ac-bd)\geq0,\,(ab-cd)\geq0$ and it remains  the constraint
$bc\geq ad$. The last inequality states that the product of the
largest and smallest coefficients must not exceed the product of
remaining coefficients. Denote by $s$ the smallest coefficient

\begin{equation}\label{ad.small}
s=\min(a,b,c,d).
\end{equation}

We can summarize all cases as follows

\begin{equation}\label{ad.largesmal}
l^2s^2\leq abcd.
\end{equation}

This is necessary but not sufficient condition for the existence
of $R_\times$. The next condition $S_\times^2>0$ we do not analyze
because the first condition (\ref{ad.largesmal}) suffices to
separate the validity domains.

\subsection{Separation of validity domains.}

In this section we define applicable domains of expressions
(\ref{cq.eigfin}),(\ref{cr.eigfin}) and (\ref{lc.larg}) step by
step.

\smallskip

\paragraph{Circumradius of convex quadrangle.} First we separate
the validity domains between the convex quadrangle and the largest
coefficient. In a highly entangled region, where the center of
circumcircle lies inside the quadrangle, the circumradius is
greater than any of sides and yield a correct answer. This
situation is changed when the center lies on the largest side of
the quadrangle and both equations (\ref{cq.eigfin}) and
(\ref{lc.larg}) give equal answers. Suppose that the side $a$ is
the largest one and the center lies on the side $a$. A little
geometrical speculation yields

\begin{equation}\label{ad.con-lar-a}
a^2=b^2+c^2+d^2+2\frac{bcd}{a}.
\end{equation}

From this equation we deduce that if $a^2$ is smaller than r.h.s.,
{\it i.e.}

\begin{equation}\label{ad.dom-lar-a}
a^2\leq b^2+c^2+d^2+2\frac{bcd}{a},
\end{equation}

\noindent then the circumradius-formula is valid. If $a^2$ is
greater than r.h.s in Eq.(\ref{ad.con-lar-a}), then the largest
coefficient formula is valid. The inequality (\ref{ad.dom-lar-a})
also guarantees the existence of the cyclic quadrilateral. Indeed,
using the inequality

\begin{equation}\label{ad.inter}
bc+cd+bd\geq3\frac{bcd}{a},
\end{equation}

\noindent one derives

\begin{equation}\label{ad.cycl}
(b+c+d)^2\geq b^2+c^2+d^2+\frac{6bcd}{a}\geq a^2.
\end{equation}

\noindent  Above inequality ensures the existence of a convex
quadrangle with the given sides.

To get a confidence, we can solve equation $u_k=\pm1$ using the
relation (\ref{ad.con-lar-a}). However, it is more transparent to
factorize it as following:

\begin{subequations}\label{ad.board1}
\begin{eqnarray}\label{ad.board1+}
(4\omega^2-r_3^2)(1+u_k) &=&
\frac{2ad}{bc+ad}\left(b^2+c^2+d^2+\frac{2bcd}{a}-a^2\right)
\nonumber\\&\times&\left(a^2+b^2+c^2+\frac{2abc}{d}-d^2\right)
\end{eqnarray}
\begin{eqnarray}\label{ad.board1-}
(4\omega^2-r_3^2)(1-u_k) &=& \frac{2bc}{bc+ad}
\left(a^2+c^2+d^2+\frac{2acd}{b}-b^2\right)
\nonumber\\&\times&\left(a^2+b^2+d^2+\frac{2abd}{c}-c^2\right).
\end{eqnarray}
\end{subequations}

Similarly, we have

\begin{subequations}\label{ad.board2}
\begin{eqnarray}\label{ad.board2+}
(4\omega^2-r_3^2)(1+v_k) &=&
\frac{2bd}{ac+bd}\left(a^2+c^2+d^2+\frac{2acd}{b}-b^2\right)
\nonumber\\&\times&\left(a^2+b^2+c^2+\frac{2abc}{d}-d^2\right)
\end{eqnarray}
\begin{eqnarray}\label{ad.board2-}
(4\omega^2-r_3^2)(1-v_k) &=&
\frac{2ac}{ac+bd}\left(b^2+c^2+d^2+\frac{2bcd}{a}-a^2\right)
\nonumber\\&\times&\left(a^2+b^2+d^2+\frac{2abd}{c}-c^2\right).
\end{eqnarray}
\end{subequations}

Thus, the circumradius of the convex quadrangle gives a correct
answer if all brackets in the above equations are positive.  In
general, Eq.(\ref{cq.eigfin}) is valid if

\begin{equation}\label{ad.con-lar}
l^2\leq\frac{1}{2}+\frac{abcd}{l^2}.
\end{equation}

When one of parameters vanishes, i.e. $abcd=0$, inequality
(\ref{ad.con-lar}) coincides with the corresponding condition in
Ref.\cite{analyt}.

\smallskip

\paragraph{Circumradius of crossed quadrangle.} Next we
separate the validity domains between the convex and the crossed
quadrangles. If $S_\times^2<0$, then crossed one has no
circumcircle and the only choice is the circumradius of the convex
quadrangle. If $S_\times^2>0$, then we use the equality

\begin{equation}\label{ad.dif}
4R^2_q-4R^2_\times=\frac{r}{2}\frac{abcd}{S^2_qS^2_\times}
\end{equation}

\noindent where $r=r_1r_2r_3$. It shows that $r>0$ yields
$R_q>R_\times$ and vice-versa. Entanglement eigenvalue always
takes the maximal value. Therefore, $\lm=4R_q^2$ if $r>0$ and
$\lm=4R_\times^2$ if $r<0$. Thus $r=0$ is the separating surface
and it is necessary to analyze the condition $r<0$.

Suppose $a\geq b\geq c\geq d$. Then $r_2$ and $r_3$ are positive.
Therefore $r$ is negative if and only if $r_1$ is negative, which
implies

\begin{equation}\label{ad.aplusd}
a^2+d^2>b^2+c^2.
\end{equation}

Now suppose $a\geq d\geq b\geq c$. Then $r_1$ is negative    and
$r_3$ is positive. Therefore $r_2$ must be positive, which implies

\begin{equation}\label{ad.aplusc}
a^2+c^2>b^2+d^2.
\end{equation}

It is easy to see that in both cases left hand sides contain the
largest and smallest coefficients. This result can be generalized
as follows: $r\leq0$ if and only if

\begin{equation}\label{ad.geq}
l^2\geq \frac{1}{2}-s^2.
\end{equation}

It remains to separate the validity domains between the
crossed-quadrangle and the largest coefficient. We can use three
equivalent ways to make this separation:

\smallskip

1)to use the geometric picture and to see when $4R_\times^2$ and
$l^2$ coincide,

2)directly factorize equation $u_k=\pm1$,

3)change the sign of the parameter $d$.

\smallskip

All of these give the same result stating that
Eq.(\ref{cr.eigfin}) is valid if

\begin{equation}\label{ad.leq}
l^2\leq\frac{1}{2}-\frac{abcd}{l^2}.
\end{equation}

Inequalities (\ref{ad.geq}) and (\ref{ad.leq}) together yield

\begin{equation}\label{ad.togeth}
l^2s^2\geq abcd.
\end{equation}

This inequality is contradicted by (\ref{ad.largesmal}) unless
$l^2s^2=abcd$. Special cases like $l^2s^2=abcd$ are considered in
the next section. Now we would like to comment the fact that
crossed quadrangle survives only in exceptional cases. Actually
crossed case can be obtained from the convex cases by changing the
sign of any parameter. It crucially depends on signs of parameters
or, in general, on phases of parameters. On the other hand all
phases in Eq.(\ref{w.psi.shared}) can be eliminated by
LU-transformations. For example, the phase of $d$ can be
eliminated by redefinition of the phase of the state function
$|\psi\ra$ and the phases of remaining parameters can be absorbed
in the definitions of basis vectors $|1\ra$ of the qubits A, B and
C. Owing to this entanglement eigenvalue being LU invariant
quantity does not depend on phases. However, crossed case is
relevant if one considers states given by Generalized Schmidt
Decomposition(GSD) \cite{acin}. In this case phases can not be
gauged away and crossed case has its own range of definition. This
range has shrunk to the separating surface $r=0$ in our case.

\medskip

Now we are ready to present a distinct separation of the validity
domains:

\begin{equation}\label{ad.full}
\lm=
\begin{cases}
 \enskip 4R_q^2, &{\rm if}\quad l^2\leq1/2+abcd/l^2\cr
 \enskip l^2 &{\rm if}\quad l^2\geq1/2+abcd/l^2
\end{cases}
\end{equation}

As an illustration we present the plot of $d$-dependence of $\lm$
in Fig.2 when $a=b=c$.

\begin{figure}
\includegraphics[width=10cm]{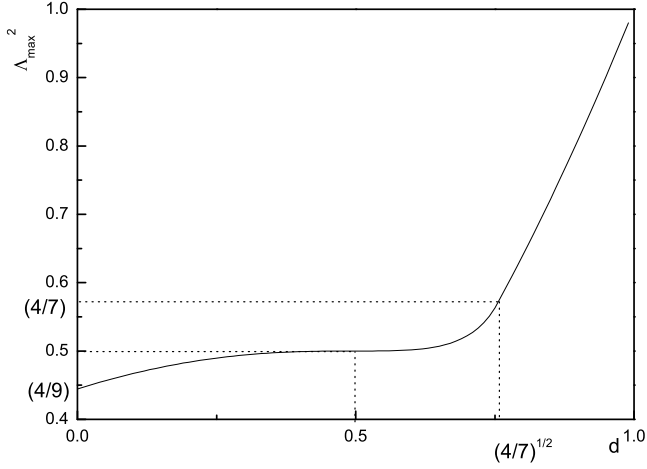}
\caption{\label{sim.eps}Plot of $d$-dependence of
$\Lambda_{max}^2$ when $a=b=c$. When $d\to1$, $\lm $ goes to $1$
as expected. When $d=0$, $\lm$ becomes $4/9$, which coincides with
the result of Ref.\cite{wei-03}. When $r=0$ which implies $a = d =
1/2$, $\lm$ becomes $1/2$ (it is shown as dotted line). When $d =
2 a$, which implies $d = \sqrt{4/7}$, $\lm$ goes to $4/7$, which
is one of shared states (it is also shown as another dotted
line).}
\end{figure}

We have distinguished three types of quantum states depending on
which expression takes entanglement eigenvalue. Also there are
states that lie on surfaces separating different applicable
domains. They are shared by two types of quantum states and may
have interesting features. We will call those shared states. Such
shared states are considered in the next section.

\section{Shared States.}

Consider quantum states for which both convex and crossed
quadrangles yield the same entanglement eigenvalue.
Eq.(\ref{cr.ab}) is not applicable and we rewrite equations
(\ref{cq.radius}) and (\ref{cr.radius}) as follows

\begin{equation}\label{ss.con-cr}
4R_q^2=\frac{1}{2}\left(1-\frac{r}{16S^2_q}\right),\quad
4R_\times^2=\frac{1}{2}\left(1-\frac{r}{16S^2_\times}\right).
\end{equation}

These equations show that if the state lies on the separating
surface $r=0$, then entanglement eigenvalue is a constant

\begin{equation}\label{ss.shar1}
\lm=\frac{1}{2}
\end{equation}

\noindent and does not depend on the state parameters. This fact
has a simple interpretation. Consider the case $r_1=0$. Then
$b^2+c^2=a^2+d^2=1/2$ and the quadrangle consists of two right
triangles. These two triangles have a common hypotenuse and legs
$b,c$ and $a,d$, respectively, regardless of the triangles being
in the same semicircle or in opposite semicircles. In both cases
they yield same circumradius. Decisive factor is that the center
of the circumcircle lies on the diagonal. Thus the perimeter and
diagonals of the quadrangle divide ranges of definition of the
convex quadrangle. When the center of circumcircle passes the
perimeter, entanglement eigenvalue changes-over from convex
circumradius to the largest coefficient. And if the center lies on
the diagonal, convex and crossed circumradiuses become equal.

We would like to bring plausible arguments that this picture is
incomplete and there is a region that has been shrunk to the
point. Consider three-qubit state given by GSD

\begin{equation}\label{ss.gsd}
|\psi\ra=a|100\ra+b|010\ra+b|001\ra+d|111\ra+e|000\ra.
\end{equation}

One of parameters must have non-vanishing phase\cite{acin} and we
can treat this phase as an angle. Then, we have five sides and an
angle. This set defines a  sexangle that has circumcircle. One can
guess that in a highly entangled region entanglement eigenvalue is
the circumradius of the sexangle. However, there is a crucial
difference. Any convex sexangle contains a star type area and the
sides of this area are the diagonals of the sexangle. The
perimeter of the star separates the convex and the crossed cases.
Unfortunately, we can not see this picture in our case because the
diagonals of a quadrangle confine a single point. It is left for
future to calculate the entanglement eigenvalues for arbitrary
three qubit states and justify this general picture.

Shared states given by $r=0$ acquire new properties. They can be
used for perfect teleportation and superdense coding
\cite{analyt,agr}. This statement is not proven clearly, but also
no exceptions are known.

\medskip

Now consider a case where the largest coefficient and circumradius
of the convex quadrangle coincide with each other. The separating
surface is given by

\begin{equation}\label{ss.con-lar}
l^2=\frac{1}{2}+\frac{abcd}{l^2}.
\end{equation}

Entanglement eigenvalue ranges within the narrow interval

\begin{equation}\label{ss.shar2}
\frac{1}{2}\leq\lm\leq\frac{4}{7}.
\end{equation}

It separates slightly and highly entangled states. When one of
coefficients is large enough and satisfies the relation
$l^2>1/2+abcd/l^2$, entanglement eigenvalue takes a larger
coefficient. And the expression (\ref{w.psi.shared}) for the state
function effectively takes the place of Schmidt decomposition. In
highly entangled region no similar picture exists and all
coefficients participate in equal parts and yield the
circumradius. Thus, shared states given by Eq.(\ref{ss.con-lar})
separate slightly entangled states from highly entangled ones, and
can be ascribed to both types.

What is the meaning of these states? Shared states given by $r=0$
acquire new and important features. One can expect that shared
states dividing highly and slightly entangled states also must
acquire some new features. However, these features are yet to be
discovered.

\section{Conclusions}

We have considered four-parametric families of three qubit states
and derived explicit expressions for entanglement eigenvalue. The
final expressions have their own geometrical interpretation. The
result in this paper with the results of Ref.\cite{analyt} show
that the geometric measure has two visiting cards: the
circumradius and the largest coefficient. The geometric
interpretation may enable us to predict the answer for the states
given by GSD. If the center of circumcircle lies in star type area
confined by diagonals of the sexangle, then entanglement
eigenvalue is the circumradius of the crossed sexangle(s). If the
center lies in the remaining part of sexangle, the entanglement
eigenvalue is the circumradius of the convex sexangle. And when
the center passes the perimeter, then entanglement eigenvalue is
the largest coefficient. Although we cannot justify our prediction
due to lack of computational technique, this picture surely
enables us to take a step toward a deeper understanding of the
entanglement measure.

Shared states given by $r=0$ play an important role in quantum
information theory. The application of shared states given by
Eq.(\ref{ss.con-lar}) is somewhat questionable, and should be
analyzed further. It should be pointed out that one has to
understand the properties of these states and find the possible
applications. We would like to investigate this issue elsewhere.

Finally following our procedure, one can obtain the nearest
product state of a given three-parametric W-type state. These two
states will always be separated by a line of densities composed of
the convex combination of W-type states and the nearest product
states \cite{pit}. There is a separable density matrix $\varrho_0$
which splits the line into two parts as follows. One part consists
of separable densities and another part consists of  non-separable
densities. It was shown in Ref.\cite{pit} that an operator
$W=\varrho_0-\rho^{ABC}-\tr[\varrho_0(\varrho_0-\rho^{ABC})]I$ has
the properties $\tr(W\rho^{ABC})<0$, and $\tr(W\varrho)\geq0$ for
the arbitrary separable state $\varrho$. The operator $W$ is
clearly Hermitian and thus is an entanglement witness for the
state. Thus our results allow oneself to construct the
entanglement witnesses for W-type three qubit states. However, the
explicit derivation of $\varrho_0$ seems to be highly non-trivial.


\chapter[Entanglement of n-qubit W states]
{Duality and the geometric measure of entanglement of general multiqubit W states}\label{duality}

\renewcommand{\ra}{\rangle}
\renewcommand{\la}{\langle}
\renewcommand{\tr}{\mathrm{tr}}
\renewcommand{\s}[1]{\sqrt{#1}}
\renewcommand{\uv}{{\bm u}}
\renewcommand{\rv}{{\bm r}}
\newcommand{\e}{\mathrm{e}}
\newcommand{\w}{\mathrm{W}}
\renewcommand{\ket}[1]{\vert #1 \ra}
\renewcommand{\bra}[1]{\la #1 \vert}
\renewcommand{\bk}[2]{\la #1 \vert #2 \ra}
\renewcommand{\kb}[2]{\vert #1 \ra \la #2 \vert}
\renewcommand{\ov}[2]{\left\la #1 | #2 \right\ra}
\renewcommand{\(}{\left(}
\renewcommand{\)}{\right)}
\newcommand{\half}{\tfrac{1}{2}}
\renewcommand{\t}{\theta}
\renewcommand{\v}{\varphi}
\renewcommand{\o}{\otimes}
\newcommand{\<}{\langle}
\renewcommand{\>}{\rangle}
\newcommand{\hlf}[1]{\tfrac{#1}{2}}
\newcommand{\third}{\tfrac{1}{3}}
\newcommand{\thrd}[1]{\tfrac{#1}{3}}
\newcommand{\quarter}{\tfrac{1}{4}}
\newcommand{\eighth}{\tfrac{1}{8}}
\newcommand{\qqquad}{\quad\quad\quad}
\newcommand{\qqqquad}{\quad\quad\quad\quad}
\newcommand{\imp}{\Rightarrow}
\newcommand{\impl}{\Longrightarrow}
\newcommand{\R}{\mathbb{R}}
\newcommand{\x}{\mathbf{x}}
\newcommand{\thmref}[1]{Theorem~\ref{#1}}
\newcommand{\secref}[1]{Section~\ref{#1}}
\newcommand{\lemref}[1]{Lemma~\ref{#1}}
\newcommand{\propref}[1]{Proposition~\ref{#1}}
\newcommand{\be}{\begin{equation}}
\newcommand{\ee}{\end{equation}}
\theoremstyle{plain} 
\newtheorem{lemma}{Lemma}
\newtheorem{theorem}{Theorem}

In this chapter we find the nearest product states and geometric
measure of entanglement for arbitrary generalized W states of $n$
qubits~\cite{dual,toward}.

Quantifying entanglement of multipartite pure states presents a
real challenge to physicists. Intensive studies are under way and
different entanglement measures have been proposed over the
years~\cite{shim-95,vedr-97,ben-schum,vp-02,ben-rain,negat,robust}.
However, it is generally impossible to calculate their value
because the definition of any multipartite entanglement measure
usually includes a massive optimization over certain quantum
protocols or states~\cite{woot-98,analyt,iso}.

Inextricable difficulties of the optimization are rooted in a
tangle of different obstacles. First, the number of entanglement
parameters grows exponentially with the number of particles
involved~\cite{lindenprl}. Second, in the multipartite setting
several inequivalent classes of entanglement
exist~\cite{Chir,four}. Third, the geometry of entangled regions
of robust states is complicated~\cite{shared}. All of these make
the usual optimization methods
ineffective~\cite{sud-geom,shared,wei-guh}. Concise and elegant
tools are required to overcome this problem.

A widely used measure for multipartite systems is the maximal
product overlap $\Lambda_{max}$. In what follows states with
$\lm>1/2$ are referred to as slightly entangled, states with
$\lm<1/2$ are referred to as highly entangled and states with
$\lm=1/2$ are referred to as shared quantum states. In this
chapter we show how to calculate the maximal product overlap of an
arbitrary W state~\cite{Chir}. The method is to establish a
one-to-one correspondence between highly entangled W states and
their nearest product states.

Consider first generalized Greenberger-Horne-Zeilinger
states~\cite{ghz}, i.e.\ states that can be written
$|\mathrm{GHZ}\> = a|0\ldots0\> + b|1\ldots1\>$ in some product
basis. Such states are fragile under local decoherence, i.e.\ they
become disentangled by the loss of any one party, and they are not
highly entangled in the sense defined above. The geometric measure
of these states is computed easily since the maximal overlap
simply takes the value of the modulus of the larger coefficient,
$|a|$ or $|b|$ \cite{Shim-shor}. Accordingly, the nearest separable
state is the product state with the larger coefficient. Thus many
generalized GHZ states with different maximal overlaps can have
the same nearest product state.

Consider now generalized W-states~\cite{par-08}, which can be
written
\begin{equation}\label{0.w}
\ket{\w_n}=c_1\ket{100...0} + c_2\ket{010...0} + \cdots +
c_n\ket{00...01}.
\end{equation}
Without loss of generality we consider only the case of positive
parameters $c_k$ since the phases of the coefficients $c_k$ can be
eliminated by redefinitions of local states
$\ket{1_k},\,k=1,2,...,n$. The states \eqref{0.w} are robust
against decoherence~\cite{raz-02}, i.e.\ loss of any $n-2$ parties
still leaves them in a bipartite entangled state. Surprisingly, if
the state is slightly entangled, then we have the same situation
as for generalized GHZ states: the maximal overlap is the largest
coefficient and, as before, many states can have the same nearest
product state~\cite{toward}. However, the situation is changed
drastically when the state is highly entangled. The calculation of
the maximal overlap in this case is a very difficult problem and
the maximization has been performed only for relatively simple
systems~\cite{wei-03,sud-geom,analyt,Shim-shor,toward,hay-sym,mgbbb,zch}.

On the other hand, different highly entangled W-states have
different nearest product states. This makes it possible to map
the W-state to its nearest product state and quickly obtain its
geometric measure of entanglement. More precisely, we construct
two bijections. The first one creates a map between highly
entangled $n$-qubit W states and $n$-dimensional unit vectors
$\x$. The second one does the same between  $n$-dimensional unit
vectors and $n$-part product states. Thus we obtain a double map,
or {\it duality}, as follows
\begin{equation}\label{bijec}
 \ket{\w_n}\; \leftrightarrow\; \x\; \leftrightarrow\;
\ket{u_1}\o\ket{u_2}\o\cdots\o\ket{u_n}.
\end{equation}

The main advantage of the map is that if one knows any of the
three vectors, then one instantly finds the other two.

This chapter is organized as follows. In Section 4.1 we construct a classifying map.
In Section 4.2 we consider highly entangled multi-qubit W states. In Section 4.3 we derive a
closed-form expression for the maximal overlap of n-qubit W states. In Section 4.4 we summarize
our results.

\medskip

\section{Classifying map.}

Now we prove a theorem that provides a basis for the map.

\begin{theorem}  \label{Wtou} Let $\ket{\w_n}$ be an arbitrary W state
\eqref{0.w} with non-negative real coefficients $c_i$, and let
$\ket{u_1}\o\ket{u_2}\o\cdots\o\ket{u_n}$ be its nearest product
state. Then the phase of $\ket{u_k}$ can be chosen so that
 $$\ket{u_k}=\sin\t_k\ket{0}+\cos\t_k\ket{1},\, 0\leq\t_k\leq\frac{\pi}{2}, \,
k=1,2,...,n.$$ where
\begin{equation}\label{0.dircos}
\cos^2\t_1+\cos^2\t_2+\cdots+\cos^2\t_n=1.
\end{equation}
\end{theorem}

\begin{proof} The nearest product state is a stationary point for the overlap
with $\ket{\w_n}$, so the states $\ket{u_k}$ satisfy the nonlinear
eigenvalue equations \cite{wei-03,hig,analyt}
\begin{equation}\label{0.stat-eq}
\bk{u_1u_2\cdots\widehat{u_k}\cdots
u_n}{\w_n}=\Lambda_{max}\ket{u_k};\;k=1,2,\cdots,n
\end{equation}
where the caret means exclusion. We can choose the phase of
$\ket{u_k}$ so that $ \ket{u_k} = \sin\t_k\ket{0} +
\e^{i\phi_k}\cos\t_k\ket{1}, $ and then \eqref{0.stat-eq} gives
the pair of equations
\begin{subequations}\label{0.eqn}
\begin{equation}\label{0.eqcos}
c_k\prod_{j\neq k}\sin\t_j = \Lambda_{max}\e^{i\phi_k}\cos\t_k,
\end{equation}
\begin{equation}\label{0.eqsin}
\sum_{l\neq k}\e^{-i\phi_l}c_l\cos\t_l\prod_{j\neq k,l}\sin\t_j =
\Lambda_{max}\sin\t_k.
\end{equation}
\end{subequations}
Eq. \eqref{0.eqcos} shows that $\Lambda_{max}\e^{i\phi_k}$ is
real, so $\phi_k = -\arg(\Lambda_{max})$ is independent of $k$.
Then the modulus of the overlap $|\<u_1\cdots u_n|\w_n\>|$ is
independent of $\phi$, so we can assume that $\phi = 0$. Now
multiplying eq.(\ref{0.eqsin}) by $\sin\t_k$ and using
eq.(\ref{0.eqcos}) gives Eq.(\ref{0.dircos}).
\end{proof}
Thus the angles $\cos\t_k$ define a unit $n$-dimensional Euclidean
vector $\x$. We can also define a length $r$ as follows. From
Eq.(\ref{0.eqcos}) it follows that the ratio $\sin2\t_k/c_k$ does
not depend on $k$. If this ratio is non-zero we can define
\begin{equation}\label{0.rmod}
\frac{1}{r} \equiv \frac{\sin2\t_1}{c_1} = \frac{\sin2\t_2}{c_2} =
\cdots = \frac{\sin2\t_n}{c_n}.
\end{equation}

\smallskip

\section{Highly entangled W states.}

Equations (5) admit a trivial solution
$\sin2\t_k=0,\,k=1,2,\cdots,n$ and a special solution with nonzero
values of all sines. The trivial solution gives the largest
coefficient of $|\w_n\>$ for the maximal overlap and is valid for
slightly entangled states. We consider them later and now focus on
the special solutions. From Eq.(\ref{0.rmod}) it follows that
\begin{equation}\label{1.cos}
\cos^2\t_k=\frac{1}{2}\(1\pm\sqrt{1-\frac{c_k^2}{r^2}}\),\;k=1,2,\cdots,n.
\end{equation}
The plus sign means that $\cos2\t_k > 0$. Then from
Eq.(\ref{0.dircos}) it follows that this is possible for at most
one angle; specifically, we prove that if $\cos2\t_k>0$ for some
$k$, then $c_k$ is the largest coefficient in Eq.(\ref{0.w}).
Suppose $\cos2\t_k>0$ but $c_k$ is not the largest coefficient and
there exists a greater coefficient, say $c_l$. Then from
Eq.(\ref{0.rmod}) it follows that $\sin2\t_l>\sin2\t_k>0$ and
consequently $|\cos2\t_l|<|\cos2\t_k|$. Now we rewrite
Eq.(\ref{0.dircos}) as follows:
\begin{equation}\label{1.2theta}
-\cos2\t_1-\cos2\t_2-\cdots-\cos2\t_n=n-2.
\end{equation}
From $|\cos2\t_l|<|\cos2\t_k|$ and $\cos2\t_k>0$ it follows that
$-\cos2\t_k-\cos2\t_l<0$  which is in contradiction with
Eq.(\ref{1.2theta}). Thus $c_k$ must be the largest coefficient.

Without loss of generality we assume that $0 \le c_1 \le \cdots
\le c_n$. Then in \eqref{1.cos} we must take the $-$ sign for $k =
1,\ldots,n-1$ and \eqref{0.dircos} becomes
\begin{equation}\label{1.r}
\s{1 - \frac{c_1^2}{r^2}}  + \cdots + \s{1 -
\frac{c_{n-1}^2}{r^2}} \pm \s{1 - \frac{c_n^2}{r^2}} = n-2
\end{equation}
We will denote the left-hand sides of these equations as
$f_\pm(r)$. We also use $f_0(r)$ to denote this expression without
the last term. The function $r(c_1,c_2,...,c_n)$ defined by
$f_+(r)=n-2$ is a completely symmetric function of the state
parameters  $c_k$. In contrast, the function defined by
$f_-(r)=n-2$  is an asymmetric function since its dependence on
the maximal coefficient $c_n$ is different. Thus in equation
\eqref{1.r} the upper and lower signs describe symmetric and
asymmetric entangled regions of highly entangled states,
respectively.

For highly entangled states, eqs. \eqref{1.r}$_\pm$ uniquely
define $r$ as a function of the state parameters $c_k$. More
precisely,
\begin{theorem}\label{rsolutions}
There are two critical values $r_1$ and $r_2$ of the largest
coefficient $c_n$, i.e.\ functions of $c_1,\ldots,c_{n-1}$ such
that
\begin{enumerate}
\item If $c_n \le r_1$, there is a unique solution of
(\ref{1.r}$_+$) and no solution of (\ref{1.r}$_-$); \item If $c_n
= r_1$, both (\ref{1.r}$_+$) and (\ref{1.r}$_-$) have a unique
solution, the same for both; \item If $r_1 < c_n \le r_2$, there
is no solution of (\ref{1.r})$_+$ and a unique solution of
(\ref{1.r}$_-$); \item If $c_n > r_2$, neither (\ref{1.r}$_+$) nor
(\ref{1.r}$_-$) has a solution. In this case the state $|\w_n\>$
is slightly entangled.
\end{enumerate}
\end{theorem}

The value $r_1$ is the solution of $f_0(r_1) = n-2$, which exists
and is unique since $f_0(c_{n-1}) < n-2$ and $f_0(r) \rightarrow n
- 1$ monotonically as $r \rightarrow \infty$; and $r_2$ is defined
by \be r_2^2 = c_1^2 + \cdots + c_{n-1}^2. \ee Then $r_2 \ge r_1$,
for $f_0(r_2)\ge n - 2 = f_0(r_1)$ using $\sqrt{x} \ge x$ for $0
\le x \le 1$. Since $f_0$ is an increasing function of $r,$ it
follows that $r_2 \ge r_1$. Now the theorem follows from the
following properties of the functions $f_\pm(r)$($f'_-$ is the
derivative of $f_-$):

\smallskip
1. $f_0$ and $f_+$ are monotonically increasing functions of $r$.

2. $f_+(r) \rightarrow n$ as $r \rightarrow \infty$.

3. If $c_n \le r_1$, $f_+(c_n) = f_0(c_n) \le f_0(r_1) = n-2$.

4. If $c_n \ge r_1$, then $f_+(r) \ge n-2$ for all $r > r_1$.

5. If $c_n < r_1$, then $f_-(c_n) < n - 2$.

6. If $c_n > r_1$, then $f_-(c_n) > n - 2$.

7. If $c_n < r_2$, then $f_-(r) < n - 2$ for large $r$.

8. If $c_n > r_2$ then $f_-(r) > n - 2$ for large $r$.

9. $f'_-(c_n + \epsilon) < 0$ for small $\epsilon$.

${}$\!\!\!10. If $c_n > r_2$, then $f'_-(r) < 0$ for all $r \ge
c_n$.

\smallskip

These properties are illustrated in Figure \ref{curves}.

\begin{figure}[ht!]
\begin{center}
\includegraphics[width=7cm]{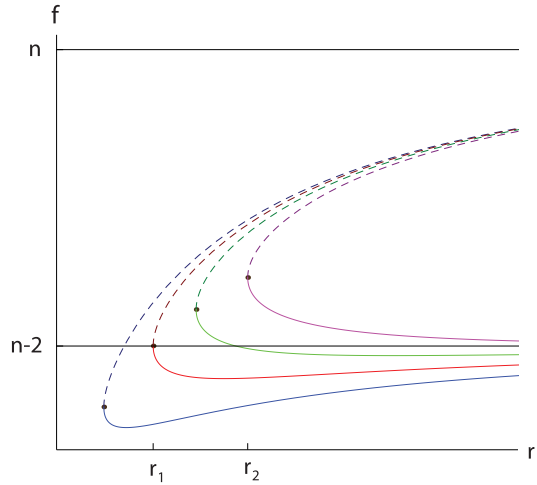}
\caption[fig1]{\label{curves}(Color online) The behaviour of the
functions $f_\pm$ for five-qubit W states. The function $f_+(r)$
(dotted line) and $f_-(r)$ (solid line) are  plotted against $r$
in the four cases $c_n < r_1$, $c_n = r_1$, $r_1 < c_n < r_2$ and
$c_n = r_2$.}
\end{center}
\end{figure}

\section{Geometric measure.}

 We can now identify the nearest product state,
and the largest product state overlap $\Lambda_{max}(|\w_n\>)$,
for any W-state $|\w_n\>$, as follows.

\begin{theorem} \label{nearestproduct} If $c_n \ge 1/2$, the state $|\w_n\>$
defined by \eqref{0.w} is slightly entangled. Its nearest product
state is $|0\ldots 01\>$, with overlap $\Lambda_{max}(|\w_n\>) =
c_n$.

If $c_n \le 1/2$, the state $|\w_n\>$ is highly entangled and has
nearest product state \be\label{product} | u_1\>\ldots|u_n\>
\;\;\text{ where }\;\; |u_k\> = \sin\t_k|0\> + |
\e^{i\phi}\cos\t_k|1\>, \ee with which its overlap is \be
\label{gdef} \Lambda_{max} = 2r\sin\t_1\sin\t_2\ldots\sin\t_n. \ee
Here $r$ is the solution of (\ref{1.r})$_\pm$, whose existence and
uniqueness are guaranteed by Theorem 2; the phase $\phi$ is
arbitrary; and $\t_k$ is given by \eqref{1.cos} with the $-$ sign
for $k = 1,\ldots,n-1$, the $-$ sign for $k = n$ if $r$ satisfies
(\ref{1.r}$_+$), the $+$ sign if $r$ satisfies (\ref{1.r}$_-$).
\end{theorem}
\begin{proof} The nonlinear eigenvalue equations \eqref{0.stat-eq} always have
$n$ solutions
\[
\Lambda_{max} = c_k,\qquad |u_i\> = \begin{cases} |0\> \text{ if }
i\neq k,\\|1\> \text{ if } i = k\end{cases}, \quad k = 1\ldots n
\]
If $c_n \ge \/2$, i.e.\ in case (4) of Theorem 2, there are no
other stationary values, so the largest overlap
$\Lambda_{max}(|\w_n\>)$ equals the largest coefficient $c_n$, the
corresponding product state being $|0\ldots 01\>$.

If $c_n < 1/2$ there is another stationary value given by
\eqref{gdef}. We will now show that this is larger than any of the
trivial stationary values $c_k$. We use the following inequality:
If $y_1,\ldots , y_n$ are real numbers lying between 0 and 1, and
satisfying $y_1 + \cdots + y_n \le 1$, then \be \label{ineq1}
(1-y_1)(1-y_2)\cdots(1-y_n) \ge 1 - y_1 -y_2 - \ldots -y_n. \ee
This is readily proved by induction.  We can apply \eqref{ineq1}
to $n-1$ terms of Eq.(\ref{0.dircos}) to get
$$(1-\cos^2\t_1)\cdots(1-\cos^2\t_{n-1})\ge 1-\cos^2\t_1 - \cdots -
\cos^2\t_{n-1}$$ or \be \sin^2\t_1\sin^2\t_2\sin^2\t_{n-1} \ge
\cos^2\t_n. \ee Now from Eq.(\ref{0.eqcos}) it follows that $\lm
\ge c_n^2$. Thus $\Lambda_{max}$ is the maximal product overlap,
and the nearest product state is $|u_1\>\ldots u_n\>$.

Next we prove that if $|\w_n\>$ is normalised, then $\lm < 1/2$.
For this we need another inequality: If $y_1,\ldots , y_n$ are
real numbers lying between 0 and 1, and satisfying $y_1 + \cdots +
y_n = n-1$, then \be \label{ineq2} y_1 + \cdots + y_n \ge y_1^2 +
\cdots + y_n^2 + 2y_1y_2\ldots y_n. \ee This can also be proved by
induction.

From \eqref{0.rmod}, and using $c_1^2 + \cdots + c_n^2 = 1$, we
find \be\label{rfromtheta} r^2 = \frac{1}{\sin^2 2\t_1 + \cdots +
\sin^2 2\t_n}. \ee Hence \eqref{gdef} gives \be\label{gfromtheta}
\lm = \frac{y_1y_2\ldots y_n}{y_1(1 - y_1) + \cdots y_n(1 - y_n)}
\ee where $y_k = \sin^2\t_k$. But $y_1 + \cdots + y_n = n -1$, so
the inequality \eqref{ineq2} applies, and gives $\lm \le 1/2$.
\end{proof}

Finally, we summarise the correspondence between highly entangled
W-states, their nearest product states, and unit vectors in
$\R^n$.
\begin{theorem}\label{onetoone}
There is a 1:1 correspondence between highly entangled states
$|\w_n\>$ defined by \eqref{0.w}, their nearest product states
with real non-negative coefficients, and unit vectors $\x\in \R^n$
with $0<x_k<1/\sqrt{2}$ ($k = 1,\ldots, n-1$), $0 < x_n < 1$.
\end{theorem}
\begin{proof} By Theorem 3, $|\w_n\>$ is highly entangled if and only if $c_n <
1/2$. If this is the case, \thmref{Wtou} and \eqref{1.cos} show
that its nearest product state is of the form \eqref{product}
where $\x = (\cos\t_1, \ldots, \cos\t_n)$ is a unit vector in
$\R^n$ in the region stated. The angles $\t_k$ are given in terms
of the coefficients $c_k$ by \eqref{0.rmod}, in which $r$ is a
function of the coefficients which, by \thmref{rsolutions}, is
uniquely defined. The nearest product states
$|u_1\>|u_2\>\ldots|u_n\>$ are determined by these angles, up to a
phase $\phi$, by $|u_k\> = \sin\t_k|0\> + \e^{i\phi}\cos\t_k|1\>$,
so there is only one nearest product state with real non-negative
coefficients, and only one unit vector $\x$, for each highly
entangled state $|\w_n\>$. Conversely, given a unit vector $\x =
(\cos\t_1,\ldots,\cos\t_n)$, the quantity $r$ is determined by
\eqref{rfromtheta}, and then the coefficients $c_1,\ldots,c_n$ are
determined by \eqref{0.rmod}. Thus the correspondences
\eqref{bijec} are bijections.
\end{proof}

The equations (\ref{1.r}$_\pm$) cannot always be explicitly solved
to give analytic expressions for $r$ in terms of the coefficients
$c_k$. However, in some cases, including all states for $n=3$,
explicit solutions can be obtained. Then the angles $\t_k$ can be
calculated from \eqref{0.rmod} and eq.\eqref{gdef} gives a formula
for the maximal product overlap $\Lambda_{max}(|\w_n\>)$. This
formula is valid unless any of the angles $\t_k$ vanishes, and
restores all known results for the maximal overlap of highly
entangled W states. When $n=3$ it coincides with the formula (31)
in Ref.\cite{analyt}. When $c_1=c_2=\cdots=c_n$  it coincides
with the formula (52) in Ref.\cite{Shim-grov}. And when $n=4$ and
$c_3=c_4$ it coincides with the formula (37) derived in
Ref.\cite{toward}.

When $ \max(c_1^2,c_2^2,\cdots c_n^2)=r_2^2=1/2$ the two
expressions for $\Lambda_{max}(|\w_n\>)$ given in
\thmref{nearestproduct} coincide; these states are shared quantum
states. The nearest product states and maximal overlaps of shared
states are given by the first case of \thmref{nearestproduct}, but
also they appear as asymptotic limits of the second case. Indeed,
at the limit $\t_n\to0$ we have
\begin{equation}\label{2.app}
\lim_{\t_n\to0}2r\sin\t_n\to c_n,\; \lim_{\t_n\to0}2r\cos\t_k\to
c_k,\,k\neq n.
\end{equation}
Thus the angle $\t_n$ vanishes and the length of the vector $\rv$
goes to infinity, but their product has a finite limit.
Substituting these limits into Eq.(\ref{0.dircos}) one obtains
$c_n^2\to r_2^2$. Therefore entangled regions of highly and
slightly entangled states are separated by the surface
$c_n^2=1/2$; for states on the surface, $r\to\infty$. All of these
states can be used as a quantum channel for the perfect
teleportation and superdense coding~\cite{shared}.

\smallskip

\section{Summary.}

We have constructed correspondences between W states,
$n$-dimensional unit vectors and separable pure states. The map
reveals two critical values for quantum state parameters. The
first critical value separates symmetric and asymmetric entangled
regions of highly entangled states, whiles the second one
separates highly and slightly entangled states. The method gives
an explicit expressions for the geometric measure when the state
allows analytical solutions, otherwise it expresses the
entanglement as an implicit function of state parameters.

It should be noted that the bijection between W states and
$n$-dimensional unit vectors is not related directly to the
geometric measure of entanglement. Therefore it is possible to
extend the method to other entanglement measures. To this end one
creates an appropriate bijection between unit vectors and
optimization points of an entanglement measure one wants to
compute.


\chapter[Universality of many-qubit entanglement]{Universal behavior of the geometric entanglement measure of many-qubit W states}\label{universal}

\renewcommand{\ra}{\rangle}
\renewcommand{\la}{\langle}
\renewcommand{\w}{\mathrm{W}}
\renewcommand{\s}[1]{\sqrt{#1}}
\renewcommand{\ket}[1]{\vert #1 \ra}
\renewcommand{\bra}[1]{\la #1 \vert}
\renewcommand{\bk}[2]{\la #1 \vert #2 \ra}
\renewcommand{\kb}[2]{\vert #1 \ra \la #2 \vert}
\renewcommand{\ov}[2]{\left\la #1 | #2 \right\ra}
\renewcommand{\t}{\theta}
\renewcommand{\v}{\varphi}
\renewcommand{\o}{\otimes}

In this chapter we analyze geometric entanglement measure of
many-qubit W states and derive an interpolating
formula~\cite{univers}.

The physics of many-particle systems differs fundamentally from
the one of a few particles and gives rise to new interesting
phenomena, such as phase transitions~\cite{orus-08,crit-rev} or
quantum computing~\cite{shor-94,ek-91,tele-93,niels}. Entanglement
theory, in particular, appears to have a much more complex and
richer structure in the N-partite case than it has in the
bipartite setting. This is reflected by the fact that multipartite
entanglement is a very active field of research that has led to
important insights into our understanding of many-particle
physics~\cite{vedr-97,wei-03,hig,robust-universal,kryu,lip-una,shar,hub}.
In view of this, it seems worthy to investigate also the behavior
of entanglement measures for large-scale systems. Despite the fact
that the number of entanglement parameters scales exponentially in
the number of particles~\cite{lind}, it is sometimes possible to
capture the most relevant physical properties by describing these
systems in terms of very few parameters.

Recently a duality between highly entangled W states and product
states has been established~\cite{dual}. The important class of W
states~\cite{Chir} represents a particular interesting set of
quantum states associated with high robustness against particle
loss and nonlocal properties of genuine entangled multipartite
states~\cite{par-08,pop-w,sen,usha}. And different experimentally
accessible schemes to generate multipartite W states have been
proposed and put into practice over the
years~\cite{raz-02,wang,jap,w-gen1}

The duality specifies a single-valued function $r$ of entanglement
parameters. We shall refer to $r$ as the entanglement diameter, as
it will play a crucial role throughout this article. Another
reason for the term entanglement diameter is that $r$ can be
interpreted geometrically as a diameter of a circumscribing
sphere. The geometrical interpretation and its illustration will
be presented in the appendix and now we focus on the physical
significance of $r$.

The entanglement diameter uniquely defines the maximal product
overlap and nearest product
state~\cite{shim-95,barn-01,wei-03,biham-02} of a given highly
entangled W state. It has two exceptional points in the parameter
space of W states. At the second exceptional point the reduced
density operator of a some qubit is a constant multiple of the
unit operator and then the entanglement diameter becomes infinite.
The maximal product overlap $\Lambda_{max}$ of these states is a
constant regardless how many qubits are involved and what are the
values of the remaining entanglement parameters. These states are
known as shared quantum states and can be used as quantum channels
for the perfect teleportation and dense coding. Thus the shared
quantum states are uniquely defined as the states whose
entanglement diameter is infinite.

Furthermore, highly entangled W states have two different
entangled regions: the symmetric and asymmetric entangled regions.
In the computational basis these regions can be defined as
follows. If a W state is in the symmetric region, then the
entanglement diameter is a fully symmetric function on the state
parameters. Conversely, if a W state is in the asymmetric region,
then there is a coefficient $c$ such that the $c$ dependence of
the entanglement diameter differs dramatically from the
dependencies of the remaining coefficients. Hence the point of
intersection of the symmetric and asymmetric regions is the first
exceptional point. It depends on state parameters and its role has
not been revealed so far. One thing was clear that the first
exceptional point does not play an important role for three- and
four-qubit W states~\cite{analyt,toward}.

In this chapter we show that the first exceptional point is
important for large-scale W states. It approaches to a fixed point
when number of qubits $N$ increases and becomes
state-independent(up to $1/N$ corrections) when $N\gg1$. As a
consequence the entanglement diameter, as well as the maximal
product overlap, becomes state-independent too and therefore
many-qubit W states have two state-independent exceptional points.
The underlying concept is that states whose entanglement
parameters differ widely, may nevertheless have the same maximal
product overlap and this phenomenon should occur at two fixed
points. This is an analog of the universality of dynamical systems
at critical points. It is an intriguing fact that systems with
quite different microscopic parameters may behave equivalently at
criticality. Fortunately, the renormalization group  provides an
explanation for the emergence of universality in critical
systems~\cite{orus-08,orusand,crit-rev}.

The developed concept distinguishes three classes of W states. The
first class consists of highly entangled W states which are below
both exceptional points and then $r$ varies from $r_{\min}=1/2$ to
$r_0\approx1/\s{3}+O(1/N)$. We will show that these states are in
the symmetric region and their entanglement diameter is a slowly
oscillating function on entanglement parameters. Accordingly, the
maximal product overlap is an almost everywhere constant close to
its greatest lower bound. Similar results have been obtained in
Ref.\cite{zch}, where it is shown that almost all multipartite
pure states with sufficiently large number of parties are nearly
maximally entangled with respect to the geometric
measure~\cite{wei-03} and relative entropy of
entanglement~\cite{vedr-97}. We will not analyze rigorously these
states since they are too entangled to be useful in quantum
information theory~\cite{agr}.

The second and most interesting class consists of highly entangled
W states which are between two exceptional points and then $r$
varies from $r_0$ to infinity. These states are in the asymmetric
region and the behavior of the entanglement diameter is curious.
We will show that $r$ is a one-variable function in this case and
depends only on the Bloch vector ${\bm b}$ of a single qubit. As a
consequence $\Lambda_{max}$ depends only on the same Bloch vector
too and its behavior is universal. That is, regardless how many
many qubits are involved and what are the remaining $N-1$
entanglement parameters the function $\Lambda_{max}({\bm b})$ is
common. We will compute analytically $\Lambda_{max}({\bm b})$ and
thereby find the Groverian and geometric entanglement
measures~\cite{wei-03,biham-02} for the large-scale W states even
if neither the number of particles nor the most of state
parameters are known.

The third class consists of slightly entangled W states which are
above both exceptional points. In this case the maximal product
overlap takes the value of the largest coefficient and these
states do not posses an entanglement diameter. We will not analyze
this trivial case, but will combine the functions
$\Lambda_{max}({\bm b})$ for slightly entangled and highly
entangled asymmetric W states and obtain an interpolating function
$\Lambda_{max}({\bm b})$ valid for both cases. It is in a perfect
agreement with numerical solutions and quantifies the many-qubit
entanglement in high accuracy($\Delta
\Lambda_{max}/\Lambda_{max}\sim10^{-3}$ at $N\sim10)$.

The importance of the interpolating formula in quantum information
is threefold. First, it connects two quantities, namely the Bloch
vector and maximal product overlap, that can be easily estimated
in experiments~\cite{guh-07,bloch}. Second, it is an example of
how do we compute entanglement of a quantum state with many
unknowns. Third, if the Bloch vector varies within the allowable
domain then maximal product overlap ranges from its lower to its
upper bounds. Then one can prepare the W state with the given
maximal product overlap, say ${\Lambda_{max}}_0$, bringing into
the position the Bloch vector, say $\Lambda_{max}({\bm
b_0})={\Lambda_{max}}_0$.

This chapter is organized as follows. In Section 5.1, we
review the main results of Ref.\cite{dual}. In Section
6.2, we consider two- and three-parameter W states in the
symmetric region and show that all of these states are almost
maximally entangled. In Section 5.3, we consider three-
and four-parameter W states in the asymmetric region and compute
explicitly their maximal product overlap. In Section 5.4,
we generalize the results of Sec.III and Sec.IV to arbitrary
many-qubit W states. In Section 5.5, we discuss our
results. In the Appendix C, we provide a geometrical
interpretation for the entanglement diameter.

\section{Maximal product overlap of W states}

In the computational basis N-qubit W states can be written as
\begin{equation}\label{2.w}
\ket{\w_n}=c_1\ket{100...0} + c_2\ket{010...0} + \cdots +
c_N\ket{00...01},
\end{equation}
where the labels within kets refer to qubits 1,2,...,N in that
order. The phases of the coefficients $c_k$ can be absorbed in the
definitions of the local states $\ket{1_i}(i=1,2,...,N)$ and
without loss of generality we consider only the case of positive
parameters. For the simplicity we assume that $c_N$ is the maximal
coefficient, that is, $c_N=\max(c_1, c_2,\cdots,c_N)$.

The maximal product overlap $\Lambda_{max}(\psi)$ of a pure state
$\ket\psi$ is given by
\begin{equation}\label{2.g}
\Lambda_{max}(\psi)=\max_{u_1,u_2,...,u_N}|\ov{\psi}{u_1u_2...u_N}|,
\end{equation}
where the maximization runs over all product states. The larger
$\Lambda_{max}$ is, the less entangled is $\ket{\psi}$. Hence for
a quantum multipartite system the geometric entanglement measure
$E_{\Lambda_{max}}$ is defined as
$$E_{\Lambda_{max}}=-\log \Lambda_{max}(\psi).$$

The maximal product overlap demarcates three different entangled
regions in the parameter space of W states:
\begin{enumerate}
  \item The symmetric region of highly entangled W states. Here
$\Lambda_{max}(c_1,c_2,...,c_N)$ is a symmetric function on all
coefficients $c_i$.
  \item The asymmetric region of highly entangled
W states. Here the invariance of $\Lambda_{max}(c_1,c_2,...,c_N)$
under the permutations of coefficients $c_i$ ceases to be true.
  \item The region of slightly entangled W states. Here the
inequity
$$\lm(c_1,c_2,...,c_N)>1/2$$
 holds.
\end{enumerate}

The appearance of the three entangled regions is the consequence
of the existence of the two critical values for the largest
coefficient $c_N$.  The first critical value
$r_1(c_1,c_2,...,c_{N-1})$ is the solution of
\begin{equation}\label{2.r1}
\s{r_1^2-c_1^2}+\s{r_1^2-c_2^2}+\cdots+\s{r_1^2-c_{N-1}^2}=(N-2)\,r_1,
\end{equation}
which always exists and is unique. Note that the first critical
value $r_1$ for the coefficient $c_N$ depends on the remaining
coefficients $c_i, i=1,2,...,N-1$ but does not depend on $c_N$.
Nonetheless we will use the abbreviation $r_1(c_N)\equiv
r_1(c_1,c_2,...,c_{N-1})$ whenever no confusion occurs.

The second critical value $r_2(c_1,c_2,...,c_{N-1)}$ is given by
\begin{equation}\label{2.r2}
r_2^2=c_1^2+c_2^2+\cdots+c_{N-1}^2.
\end{equation}
In what follows we will use the abbreviation $r_2(c_N)\equiv
r_2(c_1,c_2,...,c_{N-1)}$ for the simplicity.

The second critical value is always greater than the first one and
thus there are three cases. The first case  is $c_N < r_1$ and the
maximal product overlap is expressed via the fully symmetric
entanglement diameter $r(c_1,c_2,...,c_N)$, which is the unique
solution of
\begin{equation}\label{2.symr}
\s{r^2-c_1^2}+\s{r^2-c_2^2}+\cdots+\s{r^2-c_N^2}=(N-2)\,r.
\end{equation}
Then $\Lambda_{max}$ is given by
\begin{equation}\label{2.symg}
\lm=\frac{r^2}{2^{N-2}}\left(1+\s{1-\frac{c_1^2}{r^2}}\right)
\left(1+\s{1-\frac{c_2^2}{r^2}}\right) \cdots
\left(1+\s{1-\frac{c_N^2}{r^2}}\right)
\end{equation}
and is a bounded function satisfying the inequalities  $c_N^2 <
\lm(c_1,c_2,...,c_N) < 1/2$.

The second case is $r_1<c_N<r_2$. In this case the entanglement
diameter $r(c_1,c_2,...,c_N)$ is the unique solution of
\begin{equation}\label{2.asymr}
\s{r^2-c_1^2}+\s{r^2-c_2^2}+\cdots-\s{r^2-c_N^2}=(N-2)\,r
\end{equation}
where only the last radical has the $-$ sign. Then $\Lambda_{max}$
takes the form
\begin{equation}\label{2.asymg}
\lm=\frac{r^2}{2^{N-2}}\left(1+\s{1-\frac{c_1^2}{r^2}}\right)
\left(1+\s{1-\frac{c_2^2}{r^2}}\right) \cdots
\left(1-\s{1-\frac{c_N^2}{r^2}}\right),
\end{equation}
where again the negative root is taken from the last radical. The
expression \eqref{2.asymg} also has an upper and lower bounds and
the inequalities
$$c_N^2 < \lm(c_1,c_2,...,c_N) < 1/2$$
hold everywhere in the asymmetric region.

The third case is $c_N\geq r_2$ and $\Lambda_{max}$ takes the
value of the largest coefficient in this case
\begin{equation}\label{2.slight}
\lm=c_N^2.
\end{equation}
Now $\Lambda_{max}$ is bounded below and satisfies the inequality
$\lm>1/2$.

Despite the fact that there exist three different expressions for
the maximal product overlap it is a  continuous function on state
parameters. Indeed, at $c_N=r_1$ both Eqs. \eqref{2.symr} and
\eqref{2.asymr} have the same solution $r=r_1=c_N$ and expressions
\eqref{2.symg} and \eqref{2.asymg} for $\Lambda_{max}$ coincide.
At $c_N\to r_2$ the solution of \eqref{2.asymr} goes to infinity,
$r\to\infty$, and \eqref{2.asymg} asymptotically comes to
\eqref{2.slight}. At this limit $\lm=c_N^2=r_2^2=1/2$ and thus the
surface $\lm(c_1,c_2,...,c_N)=1/2$ separates out slightly and
highly entangled W states.

\section{Symmetric entanglement region}
In this section we analyze the maximal product overlap of two--
and three--parameter W states that belong to the symmetric region
of entanglement and show that if all coefficients are small, then
$r$ is a  slowly oscillating function close to $1/2$.

\subsection{Two parameter W states}

Equations \eqref{2.symr} and \eqref{2.asymr} are solvable for
$N=3$ and the answer is~\cite{analyt}
\begin{equation}\label{3.w3}
\Lambda_{max}=
\begin{cases}
 \enskip 2R, &{\rm if}\quad c_3^2\leq c_1^2+c_2^2\cr
 \enskip c_3, &{\rm if}\quad c_3^2\geq c_1^2+c_2^2
\end{cases}
\end{equation}
where $R$ is the circumradius of the triangle $c_1,c_2,c_3$.

When $N\geq4$ Eqs. \eqref{2.symr} and \eqref{2.asymr} cannot be
explicitly solved to give analytic expressions for $r$ in terms of
the coefficients $c_k$ unless the state posses a symmetry. For
example, for $N=4$ the equations are solvable if any two
coefficients coincide and unsolvable  if all coefficients are
arbitrary~\cite{toward}.

However, when $N\gg 1$ the situation is different. In many cases
one can derive approximate solutions that quantify the
entanglement of W states in high accuracy. We will find such
approximate solutions and compare them with the exact or numerical
solutions.

Consider first a W states with $N=m+k$ qubits and coefficients
\begin{equation}\label{3.mk-c}
c_1=c_2=\cdots=c_m=a,\quad c_{m+1}=c_{m+2}=\cdots=c_{m+k}=b.
\end{equation}
When  $m>1$ and $n>1$ the state is in the symmetric region and
Eq.\eqref{2.symr} is reduced to
\begin{equation}\label{3.mk}
m\s{r^2-a^2}+k\s{r^2-b^2}=(N-2)\,r.
\end{equation}
This equation is solvable by radicals. Setting $a=\cos\t/\s{m},\;
b=\sin\t/\s{k}$ one obtains
\begin{equation}\label{3.rmk}
r^2=\frac{2Nmk-4(N-1)(m\cos^2\t+k\sin^2\t)+2mk(N-2)\s{D}}{16(N-1)(m-1)(k-1)},
\end{equation}
where
\begin{equation}\label{3.D}
D=1-\frac{N-1}{mk}\sin^22\t.
\end{equation}

At  $m=1$ or $k=1$ the denominator and numerator vanish in
Eq.\eqref{3.rmk}, but their ratio gives the correct answer. We
will not consider this case since it is analyzed in detail  in
Ref.\cite{toward}.

If $m,k\gg1$, then $r$ is almost constant since
\begin{equation}\label{3.rapp}
r^2=\frac{1}{4}+O\left(\frac{1}{m}\right)+O\left(\frac{1}{k}\right).
\end{equation}
The question is when \eqref{3.rapp} achieves  the required
accuracy. It can be understood by reference to Fig.\ref{r-mn},
where the $\t$ dependence of the exact solution \eqref{3.rmk} is
plotted. The graphics show that $\Delta r/r\sim10^{-2}$ at
$N\sim10$.
\begin{figure}[ht!]
\begin{center}
\includegraphics[width=7.5cm]{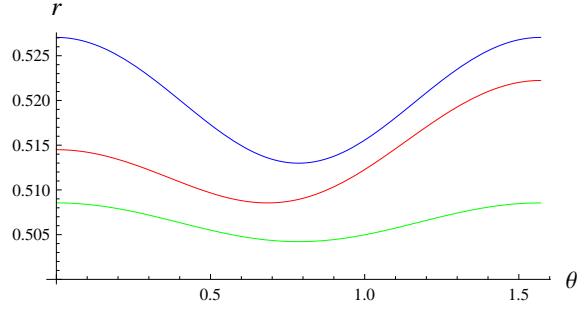}
\caption[fig1]{\label{r-mn}(Color online) The plots of the $\t$
dependence of the exact solution $r(\t)$ for the state
\eqref{3.mk-c}. The top, middle and bottom lines represent the
cases $(m=10,k=10),\;(m=12,k=18)$ and $(m=30,k=30)$,
respectively.}
\end{center}
\end{figure}
\begin{figure}[ht!]
\begin{center}
\includegraphics[width=7.5cm]{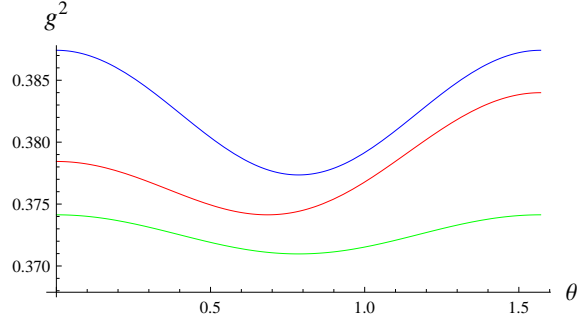}
\caption[fig2]{\label{g-mn}(Color online) The maximal product
overlap function $\lm(\t)$ at different values of $m$ and $k$. The
axes origin is put at the point $(0,1/e)$ to make it easer the
comparison of the exact ant approximate solutions. The top, middle
and bottom lines correspond to the values
$(m=10,k=10),\;(m=12,k=18)$ and $(m=30,k=30)$, respectively.}
\end{center}
\end{figure}

As a consequence of Eq.\eqref{3.rapp} $\lm$ is also almost
constant and close to its lower bound $1/e$~\cite{shim}. Indeed,
using approximations
\begin{equation}\label{3.app-c}
  \frac{1}{2^m}\left(1+\s{1-\frac{a^2}{r^2}}\right)^m\approx e^{-ma^2/4r^2},\quad \frac{1}{2^k}\left(1+\s{1-\frac{b^2}{r^2}}\right)^k\approx e^{-kb^2/4r^2}
\end{equation}
one obtains
\begin{equation}\label{3.eapp}
\lm=\frac{1}{e}+O\left(\frac{1}{m}\right)+O\left(\frac{1}{k}\right).
\end{equation}

The behavior of the maximal product overlap $\Lambda_{max}(\t)$
given by Eqs. \eqref{2.symg} and \eqref{3.rmk} is plotted in
Fig.\ref{g-mn}, which shows that $\Delta
\Lambda_{max}/\Lambda_{max}\sim10^{-2}$ at $m,k\sim10$. It is
difficult if not impossible to observe such small deviations of
the maximal product overlap in experiments and therefore
approximate formulas \eqref{3.rapp} and \eqref{3.eapp} have a good
accuracy when $N\ge20$.

\subsection{Three parameter W states}

Consider now a three-parameter W state with $N=m+k+l$ qubits and
coefficients
\begin{equation}\label{3.cmkl}
  c_1=\cdots=c_m = a,\;c_{m+1}=\cdots=c_{m+k} = b,\;c_{m+k+1}=\cdots=c_{m+k+l} = c.
\end{equation}

We will analyze the case $m,k,l\gg1$. Then Eq. \eqref{2.symr} can
be rewritten as
\begin{equation}\label{3.mkl}
m\s{r^2-a^2}+k\s{r^2-b^2}+l\s{r^2-c^2}=(N-2)\,r.
\end{equation}
From the normalization condition $ma^2+kb^2+lc^2=1$ it follows
that $a^2\leq1/m\ll1$ and similarly $b^2,c^2\ll1$. On the other
hand \eqref{3.mkl} shows that $r\sim1$, and therefore we can
expand the radicals in powers of $a^2/r^2,\;b^2/r^2$ and
$c^2/r^2$. Then
\begin{equation}\label{3.rappr}
r^2=\frac{1}{4}+O\left(\frac{1}{m},\frac{1}{k},\frac{1}{l}\right).
\end{equation}
Again we got the same answer for $r$, which means that for
partitions with large number of qubits $r$ depends neither on
$m,k,l$ nor on $a,b,c$. More precisely, $r$ depends only on the
expression $ma^2+kb^2+lc^2=|\psi|^2$, which drops out owing to the
normalization condition.

The equation \eqref{3.mkl} can be solved explicitly, but the
resulting half-page answer is impractical and we will compare
\eqref{3.rappr} with the numerical solution instead. For this
purpose we use the parametrization
$$a=\sin\t\cos\v/\s{m},\;b=\sin\t\sin\v/\s{k},\;c=\cos\t.$$
The behavior of the numerical solution $r(\t)$ of Eq.\eqref{3.mkl}
for various values $m,k,l$ and $\v$ is plotted in Fig.\ref{r-mnl}.
The graphics show that the approximate solution is in a perfect
agreement with the numerical solution for $N\gg1$.

\begin{figure}[ht!]
\begin{center}
\includegraphics[width=7.5cm]{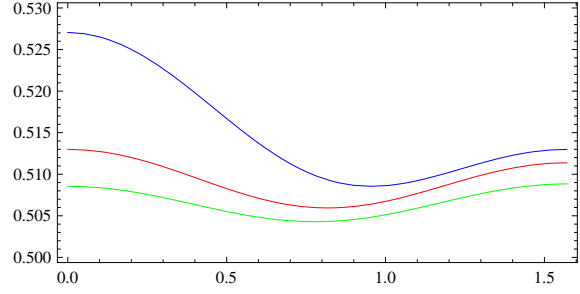}
\caption[fig3]{\label{r-mnl}(Color online)  The curves show the
$\t$ dependence of the function $r(\t)$. The upper, middle and
bottom curves represent the cases
$(m=k=l=10,\v=\pi/4),\;(m=k=l=20,\v=5\pi/12)$ and
$(m=10,k=20,l=30,\v=\pi/6)$, respectively.}
\end{center}
\end{figure}

In summary, in the symmetric region of highly entangled W states
the maximal product overlap does not depend on state parameters
when many qubits are involved. Consider a W state, where
$n_1,n_2,...,n_k$ product vectors in the computational basis have
coefficients $c_1, c_2,...,c_k$, respectively. Then
$\Lambda_{max}$ does not depend on partition numbers $n_i$ or
amplitudes $c_i$  and the approximate solution \eqref{3.rapp} with
the maximal product overlap \eqref{3.eapp} quantifies the
entanglement in high accuracy. For example, at $N\sim 10$ the
accuracy is $\Delta \Lambda_{max}/\Lambda_{max}\sim10^{-2}$. This
approximation is true unless the condition $n_i\gg1(i=1,2,...,k)$
is violated. What is happening if this condition is violated, is
analyzed in the next section.

\section{Asymmetric region of entanglement}
In this section we consider three- and four-parameter W states in
the asymmetric region and show that if one of coefficients exceeds
the first critical value $r_1$, then $r$  is a rapidly increasing
function and ranges from one-third to infinity when the maximal
coefficient ranges from the first critical value to the second
critical value.

\subsection{Three-parameter W states}

Consider now the case when $l=1$ in \eqref{3.cmkl}
\begin{equation}\label{3.cmk1}
  c_1=\cdots=c_m = a,\;c_{m+1}=\cdots=c_{m+k} = b,\;c_{m+k+1} = c.
\end{equation}


If $c\ll1$, then $c/r$ is small and $r$ is almost constant. This
case is analyzed in the previous section and now we focus on the
case when $c/r$ cannot be neglected. Then either $c\lesssim r_1$
or  $r_1<c<r_2$.

When $c\lesssim r_1$  Eq.\eqref{2.symr} takes the form
\begin{equation}\label{4.mk+1}
m\s{r^2-a^2}+k\s{r^2-b^2}+\s{r^2-c^2}=(N-2)\,r.
\end{equation}
The ratios $a/r$ and $b/r$ are small since $m,k\gg1$. Hence we
expand the radicals in powers of these ratios up to  quadratic
terms and solve the resulting equation. The answer is
\begin{equation}\label{4.r+1}
r=\frac{1}{2}\frac{1-c^2}{\s{1-2c^2}},\quad
\s{1-\frac{c^2}{r^2}}=\frac{1-3c^2}{1-c^2},\quad
\max(a^2,b^2)<c^2\leq\frac{1}{3}.
\end{equation}
It is reasonable that $r\to 1/2$ at $c\to0$.

When $c\geq r_1$  Eq.\eqref{2.asymr} takes the form
\begin{equation}\label{4.mk-1}
m\s{r^2-a^2}+k\s{r^2-b^2}-\s{r^2-c^2}=(N-2)\,r.
\end{equation}
Its approximate solution is
\begin{equation}\label{4.r-1}
r=\frac{1}{2}\frac{1-c^2}{\s{1-2c^2}},\quad
\s{1-\frac{c^2}{r^2}}=\frac{3c^2-1}{1-c^2},\quad
\frac{1}{3}<c^2<\frac{1}{2}.
\end{equation}
As one would expect, $r\to\infty$ at $c^2\to1/2$.

Surprisingly,  both solutions \eqref{4.r+1} and \eqref{4.r-1} can
be unified to a single solution as follows
\begin{equation}\label{4.r}
r=\frac{1}{2}\frac{1-c^2}{\s{1-2c^2}},\quad
\max(a^2,b^2)<c^2<\frac{1}{2}.
\end{equation}
The question at issue is when \eqref{4.r} gives a required
accuracy in the asymmetric region $r_1<c<r_2$. We compare it with
the numerical solutions of \eqref{4.mk+1} and \eqref{4.mk-1} for
the values $(m=8,k=10,a/b=0.8,r_1^2\approx0.34)$ in
Fig.\ref{r-mn1}, where the solid line is the plot of \eqref{4.r}
and the dashed line is the numerical solution. Remarkably, the
approximate solution is in a perfect agreement with the numerical
one in the asymmetric region.

\begin{figure}[ht!]
\begin{center}
\includegraphics[width=7.5cm]{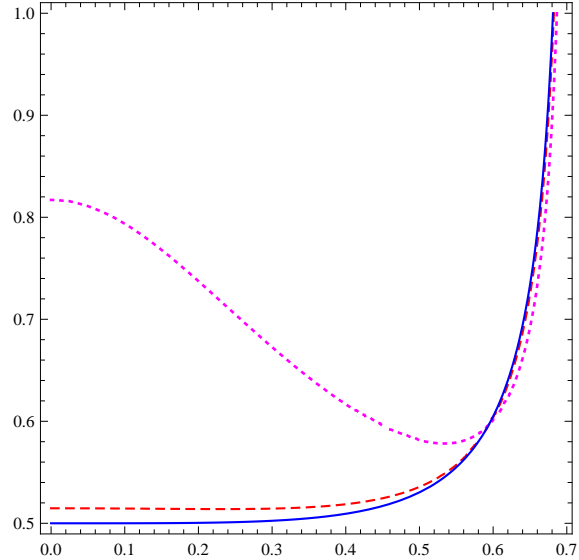}
\caption[fig4]{\label{r-mn1}(Color online)  Graphic illustrations
of the function $r(c)$ for the three- and four-parameter W states.
The solid curve is the approximate solution \eqref{4.r}. The
dashed curve is the joined numerical solution of Eqs.
\eqref{4.mk+1} and  \eqref{4.mk-1}. All remaining coefficients are
well away from the first critical value $(\approx0.58)$ when $c$
varies within the range of definition in this case. Accordingly,
the state is in the symmetric region when $0<c<0.58$ and in the
asymmetric region when $0.58<c<0.707$. The dotted line is the
numerical solution for the state \eqref{3.cm+1-1}. Now another
coefficient may exceed the first critical value. Therefore there
are two first critical values, for the last and the preceding
coefficients, respectively. The first critical value for the next
to last coefficient $c$ is $\approx0.606$ and for the last
coefficient $d$ is $\approx0.59$ which is attained at $c=0.45657$.
Thus the state is in the symmetric region when $0.45657<c<0.606$
and in the asymmetric region otherwise. Remarkably, the three
curves coincide when $c>0.606$.}
\end{center}
\end{figure}

\subsection{Four-parameter W states}
However, there are W state that are outside the realm of the model
sketched in the previous subsection. These are states with few (at
most three) coefficients close to the first critical value
$r_1\sim1/\s{3}$. In this case these coefficients are not small
and the resulting $r$ should has a different behavior.

Notice, two coefficients cannot exceed the first critical value
simultaneously. But we can construct W states whose coefficients
depend on a free parameter in such a way that at one value of the
free parameter the last coefficient exceeds the first critical
value and at another value of the free parameter the preceding
coefficient exceeds the first critical value. Below we construct
an illustrative example of a such state and analyze its
entanglement diameter.

An example is the 19-qubit four-parameter W state with
coefficients
\begin{equation}\label{3.cm+1-1}
 c_1=\cdots=c_7\equiv a,\;c_{8}=\cdots=c_{17}\equiv b,\; c_{18}\equiv c,\;c_{19}\equiv d.
\end{equation}
For the normalized states we can use free parameters $\v,\;k$ and
$c$  as follows
$$a^2 = \frac{\cos^2\v}{7k}(1-c^2),\; b^2 = \frac{\sin^2\v}{10k}(1-c^2),\;d^2 = \frac{k-1}{k}(1-c^2).$$

Now we analyze the function $r(c)$ at $k=1.8, \v=\pi/4$.

\begin{enumerate}
 \item The  next to last coefficient $c$ coincides with its first
critical value $r_1(c)$  at $c\approx0.606$ , that is, the
solution of the system
$$7\s{r_1^2-a^2}+10\s{r_1^2-b^2}+\s{r_1^2-d^2}=17r_1\;\, {\rm
and}\;\, r_1=c$$ is $r_1=c\approx 0.606$. Then $r(c)$ should range
from $r_1(c)$ to infinity when $c$ ranges from $r_1(c)$ to 1/2 and
should has a vertical asymptote at $c^2\to1/2$.
 \item The last
coefficient $d$ coincides with its first critical value $r_1(d)$
at $d\approx0.593$, that is, the solution of the system
$$7\s{r_1^2-a^2}+10\s{r_1^2-b^2}+\s{r_1^2-c^2}=17r_1\;\, {\rm and}\;\, r_1=d$$
is $r_1 = d \approx 0.593$. Note that at this point
$c\approx0.45657$. Then $r$ should increase when $d$ ranges from
$r_1(d)$ to $d_{\max}$. But the maximum value of $d$ is less than
the second critical value since
$d_{\max}^2=d^2(c=0)=(k-1)/k=4/9<1/2$. Therefore $r$ should be
bounded above in the interval $[r_1(d),d_{max}]$ and attain a
maximum at $d_{max}$. As $d$ is a decreasing function on $c$, $r$
should attain a maximum at $c=0$ and then decrease when $c$ ranges
from 0 to 0.45657.
 \item The state is in the symmetric region when
$d<r_1(d)$ and $c<r_1(c)$. Hence $r(c)$ should be minimal and
nearly constant when $0.45657<c<0.606$.
\end{enumerate}

The dotted line in Fig.\ref{r-mn1} represents the $c$ dependence
of the function $r(c)$. It agrees completely with the above
analyze.

The main point is that all the three curves coincide when
$c>r_1(c)$. In the next section we will show that this is not
accidental and the curves must coincide. In this context the
equation \eqref{4.r} is a surprising result. The quantity $r$, as
well as the maximal product overlap $\Lambda_{max}$, depends from
$c$ only. The rest of the state parameters appear in \eqref{4.r}
in the combination $|\psi|^2-c^2$ and drop out by the
normalization condition!
\begin{figure}[ht!]
\begin{center}
\includegraphics[width=7cm]{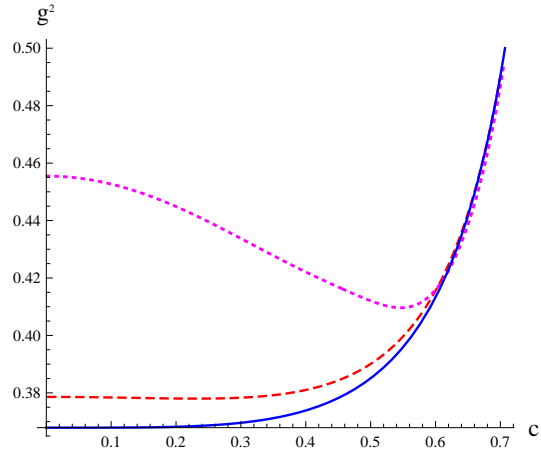}
\caption[fig5]{\label{g-mn1}(Color online) The plots of the
function $\Lambda_{max}(c)$. The solid line is the approximate
solution \eqref{4.g-as-ap}, the dashed and dotted lines are the
numerical solutions for the states \eqref{3.cmk1} and
\eqref{3.cm+1-1}, respectively. The curves may have different
behaviors when $c_N<r_1$, but coincide when $c_N\geq r_1$.}
\end{center}
\end{figure}

Furthermore, we can derive an analytic expression for the maximal
overlap. Using approximations \eqref{3.app-c} one obtains
\begin{equation}\label{4.g-as-ap}
\lm(c)=(1-c^2)e^{-(1-2c^2)/(1-c^2)}.
\end{equation}

The behavior of the function $\Lambda_{max}(c)$ is shown in
Fig.\ref{g-mn1}. The solid line is the curve \eqref{4.g-as-ap},
the dashed curve is the numerical solution for the state
\eqref{3.cmk1} and the dotted line is the numerical computation
for the state \eqref{3.cm+1-1}. They all coincide when $c>r_1(c)$.

For highly entangled states the maximal product overlap ranges
from its lower to the upper bound  when $c$ ranges from $r_1$ to
$r_2$. On the other hand the Bloch vector ${\bm b}$ of $N$th qubit
is collinear with axis $z$ and $b_z=1-2c^2$.  Thus $\Lambda_{max}$
is a one-variable function on $b_z$ and one can vary the
entanglement of the multiqubit W state by altering the Bloch
vector of a single qubit. The remaining qubits should be present
in order to create an entanglement, but their individual
characteristics do not play any role within the domain
$-1<b_z<1-2r_1^2,N\gg1$. These qubits are just spectators, they
should appear in the W state, but have no influence on the
entanglement of the state.

\section{General case}

The results of the previous sections are based on the fact that
the entanglement diameter $r$ is bounded below. In the symmetric
region it is rigidly bound by the following theorem.

{\bf Theorem 1.} If $r$ is a solution of Eq.\eqref{2.symr}, then
\begin{equation}\label{4.bounds}
\frac{1}{4}\leq r^2\leq \frac{1}{2}.
\end{equation}

{\bf Proof.} Note that
$$\frac{c_i^2}{r^2}=\left(1+\s{1-\frac{c_i^2}{r^2}}\right)\left(1-\s{1-\frac{c_i^2}{r^2}}\right) \le 2\left(1-\s{1-\frac{c_i^2}{r^2}}\right).$$
By summing over $i$ the above inequality and using \eqref{2.symr}
and the normalization condition one obtains
$$\frac{1}{r^2}\leq 2(n-n+2)=4.$$
Hence $r^2\geq1/4$. Next, from $x \le \sqrt{x}$ for $0 \le x \le
1$ it follows that
$$\sum_{i=1}^n\left(1-\frac{c_i^2}{r^2}\right)\le\sum_{i=1}^n\s{1-\frac{c_i^2}{r^2}},\quad
{\rm or}\quad n-\frac{1}{r^2}\leq n-2,$$ that is, $r^2<1/2$.

The inequalities \eqref{4.bounds} allow us to understand the
behavior of $\Lambda_{max}$ of arbitrary N-qubit W states in the
symmetric region. Indeed, in this region $c_i^2\sim1/N$ and
therefore $c_i^2/r^2\ll1$. Then one can expand the radicals in
\eqref{2.symr} and obtain
$$N-\frac{1}{2r^2}\approx N-2,$$
which generalizes \eqref{2.symr} and \eqref{2.symg} to arbitrary W
states with $c_N\ll1$.

In Eq.\eqref{3.cmkl} we have chosen equal coefficients in order to
reduce the number of independent parameters and make it easier the
analyze. Now Theorem\;1 states that it is irrelevant whether some
coefficients coincide. Decisive factor is that the coefficients
$c_i$ are small($\sim1/\s{N}$). Then the ratios $c_i/r$ are small
since $r$ is bounded below $(\sim1/2)$ and we can keep first
nonvanishing orders of these ratios. Surprisingly, all these
ratios are combined in such a way that they yield the Euclidean
norm of the state function and the final answer becomes
independent on the state parameters as well as the number of
particles involved.

In the asymmetric region the entanglement diameter $r$ should has
a lower bound but has not an upper bound since $r\to\infty$ at
$c_2\to r_2$. One may expect that the lower bound of $r$ in the
asymmetric region coincides with the upper bound of $r$ in the
symmetric region. But the following theorem shows that this is not
the case.

{\bf Theorem 2.} If $r$ is a solution of Eq.\eqref{2.asymr}, then
\begin{equation}\label{bounda}
r^2\geq \frac{1}{3}.
\end{equation}

{\bf Proof.} We use the same technique, namely
$$\frac{1}{r^2} = \sum_i^{N-1}\frac{c_i^2}{r^2} + \frac{c_N^2}{r^2} \le \sum_i^{N-1}2\left(1-\s{1-\frac{c_i^2}{r^2}}\right) + \frac{c_N^2}{r^2},$$
or
$$\frac{1}{r^2} \le 2-2\s{1-\frac{c_N^2}{r^2}} + \frac{c_N^2}{r^2} \le 3\quad {\rm since}\quad c_N\le r.$$
This bound, as well as bounds \eqref{4.bounds}, is tight, for
example, $r^2\to1/3$ at $c^2\to1/3$ in \eqref{4.r}.

Theorem\,2 explains why the asymmetric approximate solution
\eqref{4.r} fits the numerical date more quickly ($N\sim10$) than
the symmetric one \eqref{3.rapp}($N\sim20$). First, the lower
bound of $r$ is greater in this case. Second, since $c_N$ is
greater($c_N>r_1$) the remaining coefficients should be smaller
due to the normalization condition. These two factors together
make the ratio $c_i/r$ smaller. Hence the approximate solution
should has a better agreement with the exact one. Aside from that,
$r$ is a fast increasing function and goes to the infinity unlike
to the symmetric case. Hence the values of the  coefficients $c_i$
become irrelevant when $r\gg1$.

In fact  there is no W state in the asymmetric region that differs
markedly from the above model when many qubits are involved. The
following theorem completes the proof that in the asymmetric
region the maximal product overlap is a one-variable function.

{\bf Theorem 3.} If $c_N=r_1$, then
\begin{equation}\label{bound-cr1}
r_1^2=\frac{1}{3}+O(\frac{1}{N})
\end{equation}

{\bf Proof.} Note that on the boundary of the symmetric and
asymmetric regions $r=r_1=c_N$ and therefore $r_1^2\geq1/3$.
Expanding the radicals in \eqref{2.r1} in powers of $c_i^2/r_1^2$
one obtains
$$N-1-\frac{1-c_N^2}{2c_N^2}+O(\frac{1}{N})=N-2,$$
which gives \eqref{bound-cr1}.

Now we are ready to explain what is happening in the asymmetric
region.
\begin{enumerate}
 \item When many qubits $(N\gg1)$ are involved the first critical
value depends neither the number of qubits nor the state
parameters and is a constant, $r_1\approx1/\s{3}$.
 \item Regardless what is happening in the interval $0<c_N<r_1$ all
functions $r(c)$ must  converge to the point
$r(1/\s{3})\approx1/\s{3}$. This is the effect of the first
critical value.
 \item All functions $r(c)$ have the the same
vertical asymptote, namely, $r(c)\to\infty$ at $c\to1/\s{2}$. This
is the effect of the second critical value.
\end{enumerate}

These statements together give no chance to differ markedly exact
and approximate solutions in the asymmetric region. In conclusion,
when $N\gg1$, everywhere the maximal product overlap of W states
is governed by the smallest $b_z$ among the $z$ components of the
Bloch vectors. Using approximations
$$\frac{1}{2}\left(1+\s{1-\frac{c_i^2}{r^2}}\right)\approx e^{-c_i^2/4r^2},\quad i=1,2,\cdots,N-1$$
and equations \eqref{2.slight} and \eqref{4.r}  one obtains
\begin{equation}\label{5.g}
\lm(N\gg1)=
\begin{cases}
 \enskip \frac{1+b_z}{2}\,e^{-\frac{2b_z}{1+b_z}}, &{\rm if}\quad 0<b_z<\frac{1}{3}\cr
 \enskip \frac{1-b_z}{2}, & {\rm if}\quad b_z<0
\end{cases}
\end{equation}
 Graphic comparison of the interpolating formula and numerical computation of $\Lambda_{max}$ is shown in Fig.~\ref{int},
 where the $b_z$ dependence of $\Lambda_{max}$ is plotted for $N=10$. The solid and dashed lines represent the interpolating function \eqref{5.g} and numerical computation, respectively.

\begin{figure}[ht!]
\begin{center}
\includegraphics[width=8cm]{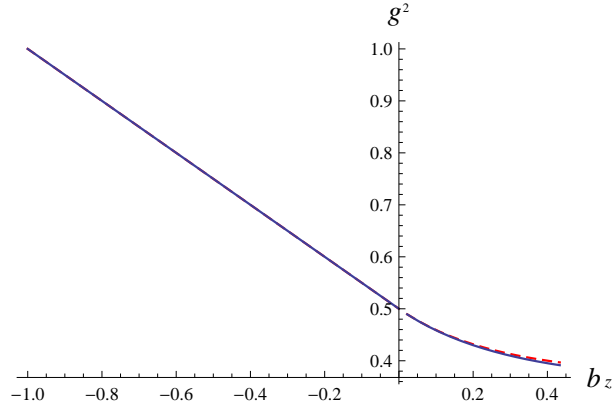}
\caption[fig6]{\label{int}(Color online) The maximal product
overlap $\Lambda_{max}$ as a function of $z$ component of the
Bloch vector $b_z$. The solid line is the interpolating formula
\eqref{5.g}. The dashed line is the numerical computation for a
10-qubit W state.}
\end{center}
\end{figure}

We did not plotted numerical results for different states because
different curves overlap and become indistinguishable. We failed
to find the states for which the numerical results markedly differ
from the plotted one provided $N\gg1$ holds.

\section{Discussion}

The main result of this work is the formula \eqref{5.g}. First, it
shows that sometimes the characterization and manipulation of the
entanglement of many qubit states is a simple task, while the case
of few or several qubits is a complicated problem. Second, it
states that when $N\gg1$ the maximal product overlap of W states
is universal in the asymmetric and slightly entangled regions and
the only exceptions are W states in the symmetric region that are
almost maximally entangled states. Then a question arises: Why do
the maximal product overlaps of the different W states far apart
from the exceptional points have the same behavior? Perhaps the
reason is that these states are all W-class states. Classification
of entangled states explains that pure states can be
probabilistically converted to one another within the same class
by stochastic local operations and classical
communication~\cite{Chir,four,bast}. And one can assume that
large-scale systems within the same class have the feature, aside
from the interconvertibility, that their entanglement is
universal. An argument in favor of this assumption is that the
geometric measure of entanglement~\cite{wei-03}, the relative
entropy of entanglement~\cite{vedr-97} and the logarithmic global
robustness~\cite{robust-universal} are related by bounding
inequalities and, moreover, the relative entropy of entanglement
is an upper bound to entanglement of distillation. Hence it is
unlikely that these measures may exhibit contradicting results and
each of them predicts its own and very different entanglement
behavior of large-scale W-states. If this argument is true, then
entanglement of large-scale states within the same class is
universal. However, states from the different classes may exhibit
different behaviors. By no means it is obvious, and probably not
true, that the maximal product overlap of GHZ-class states should
have a behavior similar to that of W states.

Another possible explanation is that the universality of the
maximal overlap of large scale W states is the inherent feature of
the geometric entanglement measure rather than the inherent
feature of quantum states. If it is indeed the case, then a
reasonable question is the following: do the exceptional points
really exist or they are just the fabrication of the geometric
entanglement measure? In this context the second exceptional point
is a fundamental quantity. Indeed, there are states applicable for
the perfect teleportation and dense coding and these states all
should possess the same amount of entanglement. Hence there is an
specific entanglement point(infinite entanglement diameter in the
case of the geometric measure) that can be associated with the
exceptional point. And one can assume that the second exceptional
point is a property of quantum states rather than a property of
the maximal product overlap.  And how about the first exceptional
point? Unfortunately, we do not know any strong arguments in favor
of it. In order to clarify the existence or nonexistence of the
first exceptional point, as well as the second exceptional point,
one has to analyze another reliable entanglement measure, say
relative entropy of entanglement~\cite{park}, and see whether it
possesses exceptional points.


\chapter[Generalized Schmidt Decomposition]{Non-strict inequality for Schmidt coefficients of three-qubit states}\label{schmidt}

\renewcommand{\p}{{\prime}}
\newcommand{\lz}{\lambda_0}
\newcommand{\lf}{\lambda_1}
\newcommand{\ls}{\lambda_2}
\newcommand{\lt}{\lambda_3}
\newcommand{\lh}{\lambda_4}
\renewcommand{\ra}{\rangle}
\renewcommand{\la}{\langle}
\newcommand{\h}{\mathcal{H}}
\renewcommand{\ket}[1]{\vert #1 \ra}
\renewcommand{\bra}[1]{\la #1 \vert}
\renewcommand{\bk}[2]{\la #1 \vert #2 \ra}
\renewcommand{\kb}[2]{\vert #1 \ra \la #2 \vert}
\renewcommand{\ov}[2]{\left\la #1 | #2 \right\ra}

In this chapter we analyze generalized Schmidt decomposition for three-qubit states and establish a relation between those Schmidt coefficients~\cite{inequal}.

Tripartite entanglement is a difficult subject for physicists. Essential results were obtained in this field~\cite{hig,acin,Chir,coff}, but fundamental problems remain unsolved. Two of them are the main obstacles to understand tripartite entanglement so well as bipartite entanglement.

The first problem is the entanglement transformation problem. Its essence is the set of necessary and sufficient conditions for transforming a given pure tripartite state to another pure tripartite state by local operations and classical communication. This problem is solved for bipartite systems~\cite{niels} and therefore the conditions for bipartite entanglement transformation based on majorization give a concise answer to the questions: among given states which ones are more/less entangled and which ones are incomparable? Unfortunately these problem is a puzzle in the case of tripartite systems.

The second problem, closely related to the first one, is the notion of maximally entangled states. This problem also is solved for bipartite systems and maximally entangled two-qubit states are the  Einstein-Podolsky-Rosen state~\cite{epr} and its local unitary(LU) equivalents known as  Bell states~\cite{bel}. However, there is no clear and unique definition of a maximally entangled state in multipartite settings. Consequently it is impossible to introduce operational entanglement measures based on optimal rates of conversion between arbitrary states and maximally entangled states~\cite{ben-conc,woot-98,ben-error}.

For bipartite systems these problems have been solved with the help of the Schmidt decomposition~\cite{schmidt1907,ek-sch}. Therefore its generalization to multipartite states can solve difficult problems related to multipartite entanglement. This generalization for three qubits is done by Ac\'in {\it et al}~\cite{acin}, where it is shown that an arbitrary pure state can be written as a linear combination of five product states. Independently, Carteret {\it et al} developed a method for such a generalization for pure states of arbitrary multipartite system, where the dimensions of the individual state spaces are finite but otherwise arbitrary~\cite{hig}. The main idea of this method is the following. First one finds the product vector which gives maximal overlap with a given quantum state vector. Then one considers product vectors orthogonal to the first product vector and finds among them the product vector that gives maximal overlap with the state vector. Continuing in this way, one finds a set of orthogonal product states and presents the state function as a linear combination of these product vectors. Since the first product vector is a stationarity point, the resulting canonical form contains a minimal set of state parameters.

Just as in bipartite case, the largest coefficient of this canonical form is the maximal product overlap which is an increasing entanglement monotone~\cite{barn-01}. Just as in bipartite case, the second largest coefficient is the maximal overlap over product states orthogonal to the nearest product state and so on. Additionally, this generalization of the Schmidt decomposition(GSD) gives insight into the nature of the maximally entangled three-qubit states~\cite{maxim} and is a good tool to extend Nielsen's theorem and operational entanglement measures to multipartite cases. Hence we accept that the amplitudes of GSD proposed in~\cite{hig} are multipartite Schmidt coefficients.

However, for a given quantum state the canonical form is not unique and the same state can have different canonical forms and therefore different sets of such amplitudes. The reason is that the stationarity equations defining stationarity points are nonlinear equations and in general have several solutions of different types. For instance, three-parameter W type states have four stationary points that create four equivalent canonical forms for the same W type state. This point is explained in detail in section III and now we focus on the question which of amplitude sets should be treated as Schmidt coefficients and which ones should be treated as insignificant mathematical solutions. A criterion should exist that can distinguish right Schmidt coefficients from false ones and we need such a criterion. It is unlikely that we can solve problems of three-qubit entanglement without knowledge of what quantities are the relevant entanglement parameters.

The canonical form whose largest coefficient is the maximal product overlap, as in bipartite case, presents GSD and others are irrelevant solutions of stationarity equations. Then our task is to single out the canonical form whose largest coefficient is the maximal product overlap and this requirement gives rise to a nontrivial relation between Schmidt coefficients of three-qubits. This situation differs from the bipartite case, where each set of positive numbers satisfying the normalization condition presents Schmidt coefficients of a some quantum state and its LU-equivalents. In contrary, in three-qubit case four positive and one complex coefficients satisfying the normalization condition are Schmidt coefficient if they satisfy an equality (derived in section V), otherwise they do not present relevant entanglement parameters at all. This is the main result of this work.

It is clear how do we single out the canonical form whose largest coefficient is the maximal product overlap. We should single out the closest product state of a given quantum state that gives a true maximum for overlap.  Of course, we cannot find closest product states of generic three-qubit states because there is no method to solve generic stationarity equation so far. Hence to distinguish the true maximum from other stationary points we require that the second variation of the maximal product overlap is negative everywhere and this condition yields the desired inequality.

However, the derived non-strict inequality is a necessary but not a sufficient condition for specifying uniquely the Schmidt coefficients. It establishes an upper bound for the three middle coefficients and this upper bound is defined by the largest coefficient. But it does not give an upper bound for the last coefficient which also should have an upper bound conditioned by four previous coefficients. The existence of an upper bound for the last coefficient is clarified in section IV which means that an additional inequality is needed to distinguish clearly the right Schmidt coefficients from the false ones.

This chapter is organized as follows. In Sec. II we repeat the derivation of GSD for three-qubit systems. In Sec.III we present an illustrative example showing that the canonical form is not unique. In Sec. IV we compute the second variation of the maximal product overlap. In Sec.V we derive the non-strict inequality for three-qubit Schmidt coefficients and analyze particular cases. In Sec. VI we show that another inequality is needed to specify uniquely Schmidt coefficients. In Sec. VII we discuss our results.

\section{Generalized Schmidt decomposition\\ for three-qubits}

In this section we derive GSD for three-qubit pure states in detail since the derivation method is used in Sec.IV to compute the second variation of the maximal product overlap.

For a three-qubit pure state $\ket{\psi}$ the maximal product overlap $\lz(\psi)$ is defined as
\begin{equation}\label{2.lzdef}
\Lambda_{max}\equiv\lz(\psi)=\max|\la u_1u_2u_3|\psi\ra|,
\end{equation}
\noindent where the maximum is over all tuples of vectors $|u_k\ra$ with $\|u_k\|=1, (k=1,2,3)$. Note, hereafter the labels within each ket refer to qubits 1, 2 and 3 in that order.

To find the maximum of $\lz(\psi)$ with constraints $\|u_k\|=1$ we form the auxiliary function $\Lambda$ given by
\begin{equation}\label{2.Lzdef}
\Lambda=|\la u_1u_2u_3|\psi\ra|^2 + \alpha_1(\la u_1|u_1\ra-1 )+ \alpha_2(\la u_2|u_2\ra-1) +\alpha_3(\la u_3|u_3\ra-1),
\end{equation}
where the Lagrange multipliers $\alpha_k$ enforce unit nature of the local vectors $|u_k\ra$.

Now we consider small variation of $|u_k\ra$ and $\alpha_k$, that is $|u_k\ra\to|u_k\ra+|\delta u_k\ra;\,\alpha_k\to\alpha_k+\delta\alpha_k$, and compute the resulting variation of $\Lambda$.
Hereafter $\delta\Lambda$  and $\delta^n\Lambda$ mean the full and the $n$th variation of $\Lambda$, respectively.

First we consider the first variation and require that $\delta^1\Lambda=0$. Then the vanishing of the partial derivatives of $\Lambda$ with respect to the Lagrange multipliers $\alpha_k$ gives
\begin{equation}\label{2.real}
\la u_1|u_1\ra-1=\la u_2|u_2\ra-1=\la u_3|u_3\ra-1=0,
\end{equation}
which are constraints on the local states $\ket{u_i}$.

The vanishing of the partial derivatives of $\Lambda$ with respect to these local states gives
\begin{eqnarray}\label{2.stp}
  \nonumber
  \ov{\psi}{u_1u_2u_3}\ov{u_2u_3}{\psi}+\alpha_1\ket{u_1} &=& 0 \\
  \ov{\psi}{u_1u_2u_3}\ov{u_1u_3}{\psi}+\alpha_2\ket{u_2} &=& 0 \\
  \nonumber
  \ov{\psi}{u_1u_2u_3}\ov{u_1u_2}{\psi}+\alpha_3\ket{u_3} &=& 0
\end{eqnarray}
and their Hermitian conjugates. From \eqref{2.stp} it follows that $\alpha_1=\alpha_2=\alpha_3=-\lz^2$ and therefore we can adjust phases of $\ket{u_k}$ so that stationarity equations \eqref{2.stp} become
\begin{equation}\label{2.stf}
\ov{u_2u_3}{\psi}=\lz\ket{u_1},\quad\ov{u_1u_3}{\psi}=\lz\ket{u_2},\quad \ov{u_1u_2}{\psi}=\lz\ket{u_3}.
\end{equation}

In the case of three-qubit states these equations are sufficient to construct GSD as follows. For each single-qubit state $|u_k\ra$ there is, up to an arbitrary phase, a unique single-qubit
state $|v_k\ra$ orthogonal to it. Then from \eqref{2.stf} it follows that the product states $$\ket{u_1u_2v_3},\quad\ket{u_1v_2u_3},\quad\ket{v_1u_2u_3}$$
are orthogonal to $\ket{\psi}$ and \eqref{2.stf} can be written as
\begin{eqnarray}\label{2.stlup}
  \nonumber
  \ov{u_1}{\psi} &=& \lz\ket{u_2u_3}+\lf\ket{v_2v_3}, \\
  \ov{u_2}{\psi} &=& \lz\ket{u_1u_3}+\ls\ket{v_1v_3}, \\
  \nonumber
  \ov{u_3}{\psi} &=& \lz\ket{u_1u_2}+\lt\ket{v_1v_2}.
\end{eqnarray}
 We choose the phases of $\ket{v_k}$ such that $\lf,\ls,\lt\ge0$. Note that after this choice the collective sign-flip of $\ket{v_k}$'s does not change anything and we will use this freedom in a little while.

 The state $\psi$ can be written as a linear combination of five product states as follows
\begin{equation}\label{2.gsd}
\ket{\psi}=\lz\ket{u_1u_2u_3} + \lf\ket{u_1v_2v_3} + \ls\ket{v_1u_2v_3} + \lt\ket{v_1v_2u_3} + \lh\ket{v_1v_2v_3},
\end{equation}
where $\lh$ is a complex number. It has two constraints, first $\lz\ge|\lh|$ and second $-\pi/2\leq {\rm Arg}(\lh) \leq\pi/2$, which can be achieved by the simultaneous change of the signs of the local states $\ket{v_k}$.

Sometimes one relabels $\ket{u}\to\ket{0},\;\ket{v}\to\ket{1}$ for the simplicity. We leave \eqref{2.gsd} as is and refer to as GSD for three-qubits.

\section{Illustrative example}

In this section we show that for a given state $\ket{\psi}$ the canonical form \eqref{2.gsd} is not unique except rare cases and additional relations are needed to single out the Schmidt decomposition from the useless canonical forms.

Consider a three-parameter family of W type states~\cite{Chir} given by
\begin{equation}\label{3.w}
|w(a,b,c)\ra=a|100\ra+b|010\ra+c|001\ra,
\end{equation}
where parameters $a,b,c$ are all positive since their phases can be eliminated by appropriate LU transformations. Stationarity equations \eqref{2.stf} of this state have three simple solutions and one special solution which exists if and only if parameters $a,b,c$ can form a triangle~\cite{analyt}.

The three simple solutions are
\begin{eqnarray}\label{3.soltr}
  \ket{u_1(1)}=\ket{1},\quad \ket{u_2(1)}=\ket{0},\quad\ket{u_3(1)}=\ket{0},\quad\lz(1)&=& a;\label{3.soltr1} \\
  \ket{u_1(2)}=\ket{0},\quad \ket{u_2(2)}=\ket{1},\quad\ket{u_3(2)}=\ket{0},\quad\lz(2)&=& b;\label{3.soltr2} \\
  \ket{u_1(3)}=\ket{0},\quad \ket{u_2(3)}=\ket{0},\quad\ket{u_3(3)}=\ket{1},\quad\lz(3)&=& c;\label{3.soltr3}
\end{eqnarray}
where numbers within brackets mark solutions.

The fourth nontrivial solution is
\begin{eqnarray}\label{3.solnon}
  \nonumber
  \ket{u_1(4)} &=& \frac{a\sqrt{2r_a}\,\ket{0}+\sqrt{r_br_c}\,\ket{1}}{4S},\quad
  \ket{u_2(4)} \frac{b\sqrt{2r_b}\,\ket{0}+\sqrt{r_ar_c}\,\ket{1}}{4S} \\
  \ket{u_3(4)} &=& \frac{c\sqrt{2r_c}\,\ket{0}+\sqrt{r_ar_b}\,\ket{1}}{4S},\quad \lz(4)=\frac{abc}{2S}
\end{eqnarray}
where
\begin{eqnarray}\label{3.bloch}
\nonumber
  & & r_a=b^2+c^2-a^2,\\
  && r_b=a^2+c^2-b^2,\\ \nonumber
  & & r_c=a^2+b^2-c^2
\end{eqnarray}

and $S$ is the area of the triangle $(a,b,c)$.

At $r_ar_br_c=0$ the special solution reduces to a trivial solution. Note that absolute values of these quantities $|r_a|,|r_b|,|r_c|$ are magnitudes of Bloch vectors of the first, second and third qubits, respectively and $r_ar_br_c=0$ means that some of one-particle reduced densities is a multiple of the unit matrix. In other words, the states with a completely mixed subsystems appear at the edge of the special solution and viceversa.

These four solutions of \eqref{2.stf} give the following four canonical forms for the state \eqref{3.w}

\begin{eqnarray}\label{3.cf}
  \ket{w(a,b,c)}&=&\\ \nonumber
  & &\!\!\!\lz(1)\ket{u_1(1)u_2(1)u_3(1)} + b\ket{v_1(1)u_2(1)v_3(1)} + c\ket{v_1(1)v_2(1)u_3(1)};\label{3.cf1}\\ \nonumber
 & &\\
  \ket{w(a,b,c)}&=&\\ \nonumber
  & &\!\!\!\lz(2)\ket{u_1(2)u_2(2)u_3(2)} + c\ket{u_1(2)u_2(2)u_3(2)} +  a\ket{v_1(2)v_2(2)u_3(2)};\label{3.cf2}\\ \nonumber
  & &\\
  \ket{w(a,b,c)}&=&\\ \nonumber
  & &\!\!\!\lz(3)\ket{u_1(3)u_2(3)u_3(3)} + b\ket{u_1(2)v_2(2)v_3(2)} + a\ket{v_1(3)u_2(3)v_3(3)};\label{3.cf3}
\end{eqnarray}
\begin{eqnarray}\label{3.cf4}
 \ket{w(a,b,c)} &=&\\ \nonumber
  & &\!\!\!\!\!\lz(4)\ket{u_1(4)u_2(4)u_3(4)} +i\frac{\sqrt{2r_ar_br_c}}{4S}\ket{v_1(4)v_2(4)v_3(4)}+ \\ \nonumber
 & & \!\!\!\!\!\!\frac{ar_a}{4S}\ket{u_1(4)v_2(4)v_3(4)}+\frac{br_b}{4S}\ket{v_1(4)u_2(4)v_3(4)}+  \frac{cr_c}{4S}\ket{u_1(4)u_2(4)v_3(4)}.
\end{eqnarray}

Now which of these canonical forms is a right decomposition?

It is easy to clarify this question in this particular case since we have all solutions of the stationarity equations \eqref{2.stf} and can single out the one whose largest coefficient is the dominant eigenvalue of \eqref{2.stf}.

The answer is~\cite{analyt}:
\begin{enumerate}
  \item if $r_a<0$ then only $\lz(1)$ is the maximal eigenvalue of \eqref{2.stf}, but $\lz(2)$, $\lz(3)$, $\lz(4)$ are not.
  \item if $r_b<0$ then only $\lz(2)$ is the maximal eigenvalue of \eqref{2.stf}, but $\lz(1)$, $\lz(3)$, $\lz(4)$ are not.
  \item if $r_c<0$ then only $\lz(3)$ is the maximal eigenvalue of \eqref{2.stf}, but $\lz(1)$, $\lz(2)$, $\lz(4)$ are not.
  \item otherwise only $\lz(4)$ is the maximal eigenvalue of \eqref{2.stf}, but $\lz(1)$, $\lz(2)$, $\lz(3)$ are not.
\end{enumerate}
However, we are unable to solve \eqref{2.stf} for generic states and single out the maximal eigenvalue in this way. Also we are not forced to compare all eigenvalues of \eqref{2.stf} to see whether the largest coefficient of a given decomposition is the maximal product overlap. We can just require that it is a truly maximum of the product overlap instead and obtain a criteria which shows whether the largest coefficient of a given canonical form is the maximal product overlap of the state. This will be done in next sections.

\section{The second variation\\ of the maximal product overlap}

In this section we compute the second variation of the maximal product overlap.

We compute it at stationary points to single out truly maximums and therefore we use the results coming from the vanishing of the first variation. Straightforward calculation gives
\begin{eqnarray}\label{4.var1}
  \delta^2\Lambda &=& \lz^2\left|\ov{\delta u_1}{u_1} + \ov{\delta u_2}{u_2} +\ov{\delta u_3}{u_3} \right|^2 \\ \nonumber
  &-& \lz^2\left(\|\delta u_1\|^2 +\|\delta u_2\|^2 +\|\delta u_3\|^2 \right)\\ \nonumber
  &+& \lz^2\left( \ov{\delta u_1}{u_1}\ov{\delta u_2}{u_2} + \ov{\delta u_1}{u_1}\ov{\delta u_3}{u_3} + \ov{\delta u_2}{u_2}\ov{\delta u_3}{u_3} + {\rm cc}\right)\\
  \nonumber
  &+&\lz\big(\lt\ov{\delta u_1}{v_1}\ov{\delta u_2}{v_2} + \ls\ov{\delta u_1}{v_1}\ov{\delta u_3}{v_3}\\ \nonumber
  &+&\lf\ov{\delta u_2}{v_2}\ov{\delta u_3}{v_3} + {\rm cc}\big)\\
  \nonumber
  &+&\delta\alpha_1\delta||u_1||^2 + \delta\alpha_2\delta||u_2||^2 + \delta\alpha_3\delta||u_3||^2,
\end{eqnarray}
where cc means complex conjugate.

Using the identity $\|\delta u_k\|^2\equiv\left|\ov{\delta u_k}{u_k}\right|^2 + \left|\ov{\delta u_k}{v_k}\right|^2 $ it can be rewritten as
\begin{eqnarray}\label{4.var2}
  \delta^2\Lambda &=& -\lz^2\left(\left|\ov{\delta u_1}{v_1}\right|^2 +  \left|\ov{\delta u_2}{v_2}\right|^2 +\left|\ov{\delta u_3}{v_3}\right|^2 \right) \\ \nonumber
   &+& \lz\big(\lt\ov{\delta u_1}{v_1}\ov{\delta u_2}{v_2} + \ls\ov{\delta u_1}{v_1}\ov{\delta u_3}{v_3} \\ \nonumber
   &+& \lf\ov{\delta u_2}{v_2}\ov{\delta u_3}{v_3} + {\rm cc}\big) \\ \nonumber
   &+& \lz^2\left(\delta||u_1||^2\delta||u_2||^2 + \delta||u_1||^2\delta||u_3||^2) + \delta||u_2||^2\delta||u_3||^2\right)\\ \nonumber
   &+&\delta\alpha_1\delta||u_1||^2 + \delta\alpha_2\delta||u_2||^2 + \delta\alpha_3\delta||u_3||^2.
\end{eqnarray}
From \eqref{2.real} it follows that terms containing $\delta||u_k||^2$ vanish and the second variation takes the form
\begin{eqnarray}\label{4.var3}
  \delta^2\Lambda &=& -\lz^2\left(\left|\ov{\delta u_1}{v_1}\right|^2 +  \left|\ov{\delta u_2}{v_2}\right|^2 +\left|\ov{\delta u_3}{v_3}\right|^2 \right) \\\nonumber
   &+& \lz\big(\lt\ov{\delta u_1}{v_1}\ov{\delta u_2}{v_2} + \ls\ov{\delta u_1}{v_1}\ov{\delta u_3}{v_3}\\ \nonumber
   &+& \lf\ov{\delta u_2}{v_2}\ov{\delta u_3}{v_3} + {\rm cc}\big).
\end{eqnarray}
From $\ov{\delta u_i}{v_i}\ov{\delta u_j}{v_j}\leq|\ov{\delta u_i}{v_i}\ov{\delta u_j}{v_j}|$ it follows that
\begin{equation}\label{4.est}
   \delta^2\Lambda\leq-\lz\sum_{i,j=1}^3|\ov{\delta u_i}{v_i}||\ov{\delta u_j}{v_j}|A_{ij},
\end{equation}
where the real and symmetric matrix $A$ is given by

\begin{equation}\label{4.mat}
A=
\begin{pmatrix}
\lz & -\lt & -\ls\\
-\lt & \lz & -\lf\\
-\ls & -\lf & \lz
\end{pmatrix}
.
\end{equation}
Note that the inequality \eqref{4.est} can be saturated when vectors $\ket{\delta u_k}$ are all multiples of vectors $v_k$ and therefore \eqref{4.est} gives the least upper bound of $\delta^2\Lambda$.

\section{A non-strict inequality\\ for the Schmidt coefficients}

In this section we derive a non-strict inequality for the Schmidt coefficients.

The condition $\delta^2\Lambda\leq0$ holds everywhere if and only if the matrix $A$ is positive which means that
\begin{equation}\label{5.mpos}
   {\rm tr}(A)\geq0,\quad ({\rm tr}(A))^2- {\rm tr}(A^2)\geq0,\quad \det(A)\geq0,
\end{equation}
where tr and det mean the trace and the determinant of a matrix, respectively.

The first condition ${\rm tr}(A)=3\lz>0$  is satisfied and does not give anything. Similarly, the second condition $({\rm tr}(A))^2- {\rm tr}(A^2)=6\lz^2-2(\lf^2+\ls^2+\lt^2)>0$ is a triviality since $\lz$ is the largest coefficient. But the third condition $\det(A)\geq0$ gives
\begin{equation}\label{5.ineq}
   \lz^2\;\geq\;\lf^2+\ls^2+\lt^2+2\;\frac{\lf\ls\lt}{\lz}.
\end{equation}
This is a new and unexpected relation which says that nondiagonal coefficients all together are bounded above by the quantity depending only on the largest coefficient and therefore they should be small.

Let us consider some particular cases. First consider the case when some of nondiagonal coefficients, namely $\lf$, vanishes. Then \eqref{5.ineq} reduces to the
\begin{equation}\label{5.inl10}
   \lz^2\;\geq\;\ls^2+\lt^2,\quad\lf=0.
\end{equation}
 The solution \eqref{3.soltr1} and the canonical form \eqref {3.cf1} present this case. This happens when a quantum state is a linear combination of three product states and its amplitudes in a computational basis satisfy \eqref{5.inl10}. Then the largest amplitude is the largest Schmidt coefficient and GSD is achieved by a simple flipping of local states. Similarly, the solution \eqref{3.soltr2} with the form \eqref {3.cf2} and solution \eqref{3.soltr3} with the form \eqref {3.cf3} are the cases $\ls=0$ and $\lt=0$, respectively.

 Conversely, when amplitudes of a three-term state in a computational basis do not satisfy \eqref{5.inl10} it appears a special solution \eqref{3.solnon} which creates a new factoizable basis. In this basis new amplitudes of the state given by \eqref {3.cf4} satisfy \eqref{5.ineq}.  Indeed,
 \begin{equation}\label{5.intri}
 4(abc)^2\geq (ar_a)^2 + (br_b)^2 + (cr_c)^2 +r_ar_br_c,
 \end{equation}
 which can be checked using triangle inequalities. This means that if amplitudes of the state were not satisfying \eqref{5.ineq} in the initial basis from product states then it appears a special solution giving rise to a new basis from product states and in this final basis amplitudes do satisfy \eqref{5.ineq}.

 In conclusion, \eqref{5.ineq} clearly  indicates whether a given canonical form is GSD or not and this is its main advantage.

 Another particular case which we would like to elucidate is the following. We want to find a quantum state for which \eqref{5.ineq} is saturated and nondiagonal coefficients have the maximal value. We equate all nondiagonal coefficients for the simplicity and \eqref{5.ineq} reduces to
\begin{equation}\label{5.insym}
   \lz\;\geq\;2\lambda,\quad\lf=\ls=\lt\equiv\lambda
\end{equation}
and we are looking for the states with $ \lz=2\lambda$. The W state is a such state, this is easy to see by setting $a=b=c$ in \eqref {3.cf4}. These substitutions yield
\begin{equation}\label{5.inw}
   \lz(W)=2\lambda(W)=\sqrt{2}|\lh(W)|,
\end{equation}
which shows that \eqref{5.ineq} is indeed a non-strict inequality and gives the least upper bound for the nondiagonal coefficients.

\section{Missed inequality}

In this section we show that another inequality is needed to specify uniquely the Schmidt coefficients of three-qubit states. To prove this statement let us assume the converse. Then \eqref{5.ineq} is a necessary and sufficient condition and GSD coefficients should satisfy only \eqref{5.ineq} and $\lz\ge|\lh|$. Consider symmetric states  and put $\lf=\ls=\lt=\lambda$ which yields $\lz\ge 2\lambda$. Then it exists a state such that $\lz=|\lh|=2\lambda$ and its GSD is given by
\begin{equation}\label{6.wrongpsi}
   \ket{\psi_{contr}} =\frac{1}{\sqrt{11}} \left(2\ket{000} + \ket{011} + \ket{101} + \ket{110} + 2\ket{111}\right).
\end{equation}
This is a wrong GSD. Indeed,
$$\lz^2(wrong)=\frac{4}{11},$$
but it is shown in Ref.\cite{maxim} that absolute minimum of $\lz^2$ over three-qubit pure states is 4/9 and this minimum is reached at the W state. Hence no three qubit state exists for which $\lz^2<4/9$. For the sake of clarity we present the maximal product overlap and nearest product state for the state \eqref{6.wrongpsi}
\begin{equation}\label{6.right}
   \lz^2(right)=\frac{14+3\sqrt{2}}{22},\; \ket{u_1u_2u_3} = (\cos\theta\ket{0}+\sin\theta\ket{1})^{\otimes3},\; \tan\theta=1+\sqrt{2},
\end{equation}
which can be derived by usual maximization tools.

 This example shows that conditions $\lz\ge 2\lambda$ and $\lz\ge |\lh|$ are insufficient and another relation should exist and this new relation should give bounds for the last Schmidt coefficient. We know that when all nondiagonal coefficients vanish the upper bound is $|\lh(\max)|=\lz$(known as GHZ state) and when all nondiagonal elements are maximal given by \eqref{5.inw} the upper bound is $|\lh(\max)|=\lz/\sqrt{2}$(at the W state). Hence for $|\lh|$ it exists an upper bound depending on the remaining coefficients and this upper bound gives those particular bounds at GHZ and W states, respectively.

 We can derive this upper in some simple cases, for instance, when $\ls=\lt=0$ and the state is
 \begin{equation}\label{6.simple}
   \ket{\psi_{simple}}=\lz\ket{000}+\lf\ket{011}+\lh\ket{111},
 \end{equation}
 where $\lh$ is positive as its phase is meaningless in this case.

 The stationarity equations \eqref{2.stf} for the state \eqref{6.simple} have a relevant solution given by
 \begin{equation}\label{6.simsol}
    \ket{u_1}=\frac{\lf\ket{0}+\lh\ket{1}}{\sqrt{\lf^2+\lh^2}},\quad \ket{u_2}=\ket{1}, \quad \ket{u_3}=\ket{1}, \quad \lz^\p=\sqrt{\lf^2+\lh^2}.
 \end{equation}
 From this solution it follows that \eqref{6.simple} is a right decomposition if and only if $\lz\ge\lz^\p$, that is
 \begin{equation}\label{6.simboun}
   \lz^2\,\ge\lf^2+\lh^2, \quad \ls=\lt=0.
 \end{equation}
 This inequality gives the least upper bound for the last Schmidt coefficient when two nondiagonal coefficients vanish. Unfortunately the tools used in this work were unable to find the least upper bound of $|\lh|$ for generic states.

\section{Summary}

The main result of this work is the inequality \eqref{5.ineq}. Its role is to separate out three-qubit Schmidt coefficients from the set of four positive and one complex numbers. As is explained in above section, it is a necessary but not a sufficient condition and another inequality should exist to complete the task.

It is likely that the three nondiagonal elements together define bounds for the last Schmidt coefficients in the missed inequality. Then the nondiagonal coefficients are not just extra terms in GSD, but the ones which can show some important features of tripartite entanglement unknown so far.

Another application of the derived non-strict inequality is that it can give us a hint how do we extend Nielsen's protocol or operational entanglement measures to three-qubit states. For instance, in bipartite case the protocol relies on inequalities quadratic on Schmidt coefficients. In three-qubit case such a theorem should include cubic relations as is evident from \eqref{5.ineq}.


\chapter[Summary]{Summary}\label{summary}

In this chapter we list the main results of the thesis.

\bigskip

{\bf 1.} We have  developed a method to derive algebraic equations for the geometric measure of entanglement of three-qubit pure states. Owing to it we have presented the first calculation of the geometric measure of entanglement in a wide range of three-qubit systems, including the general class of W states and states which are symmetric under the permutation of two qubits. Additionally, we have shown that the nearest separable states are not necessarily unique, and highly entangled states are surrounded by a one-parametric set of equally distant separable states.

\medskip

{\bf 2.} We have derived an explicit expression for the geometric measure of entanglement for three-qubit states that are linear combinations of four orthogonal product states and thus have Schmidt rank 4. Any pure three-qubit state can be written in terms of five preassigned orthogonal product states via Schmidt decomposition. Thus the states discussed here are more general states compared to the well-known Greenberger--Horne--Zeilinger and W states that have less rank. In fact, just a single step is needed to compute analytically the geometric measure for five-parameter states and thereby to get the answer for arbitrary three-qubit states.

\medskip

{\bf 3.} Using derived analytic expressions we have established that the geometric measure for three-qubit states has three different expressions depending on the range of definition in parameter space. Each expression of the measure has its own geometrically meaningful interpretation.  The states that lie on joint surfaces separating different ranges of definition, designated as shared states, are quantum channels for perfect teleportation and dense coding. Hence we have found a criterion which shows whether or not a given state can be be applied as a dual quantum channel.

\medskip

{\bf 4.} The Groverian measures are analytically computed in various types of three-qubit states and the final results are also expressed in terms of local-unitary invariant quantities in each type. Hence we use this relations to classify entangled regions of the Hilbert space as follows: in each region some of polynomial invariants are important and define uniquely Groverian measure of entanglement, while the remaining polynomial invariants are irrelevant. In this way we obtained six different entangled
regions for three-qubit pure states.

\medskip

{\bf 5.} We have developed a powerful method to compute analytically entanglement measures of multipartite systems. The method is based on duality which consists of two bijections. The first one creates a map between highly entangled n-qubit quantum states and n-dimensional unit vectors. The second one does the same between n-dimensional unit vectors and n-part product states. In this way we have obtained a double map or duality. The main advantage of the map is that, if one knows any of the three vectors, then one instantly finds the other two.

\medskip

{\bf 6.} We have found the nearest product states for arbitrary generalized W states of n qubits, and shown that the nearest product state is essentially unique if the W state is highly entangled. It is specified by a unit vector in Euclidean n-dimensional space. We have used this duality between unit vectors and highly entangled W states to find the geometric measure of entanglement of such states. The duality  map reveals two critical values for quantum state parameters. The first critical value separates symmetric and asymmetric entangled regions of highly entangled states, while the second one separates highly and slightly entangled states.

\medskip

{\bf 7.} We have shown that when N $\gg 1$ the geometric entanglement measure of general N-qubit W states, except maximally entangled W states, is a one-variable function and depends only on the Bloch vector with the minimal z component. Hence one can prepare a W state with the required maximal product overlap by altering the Bloch vector of a single qubit. Also we have computed analytically the geometric measure of large-scale W states by describing these systems in terms of very few parameters. The final formula relates two quantities, namely the maximal product
overlap and the Bloch vector, that can be easily estimated in experiments.

\medskip

{\bf 8.} We have derived an interpolating formula for the geometric measure of entanglement even if neither the number of particles nor the most of state parameters are known. The importance of the interpolating formula in quantum information is threefold. First, it connects two quantities, namely the Bloch vector and maximal product overlap, that can be easily estimated in experiments. Second, it is an example of how we compute entanglement of a quantum state with many unknowns. Third, if the Bloch vector varies within the allowable domain then maximal product overlap ranges from its lower to its upper bounds. Then one can prepare the W state with the given maximal product overlap  bringing into the position the Bloch vector.

\medskip

{\bf 9.} We have derived a non-strict inequality between three-qubit Schmidt coefficients, where the largest coefficient defines the least upper bound for the three nondiagonal coefficients or, equivalently, the three nondiagonal coefficients together define the greatest lower bound for the largest coefficient. In addition, we have shown the existence of another inequality which should establish an upper bound for the remaining Schmidt coefficient. The role of the inequalities is to separate out three-qubit Schmidt coefficients from the set of four positive and one complex numbers. Another application of the derived non-strict inequality is that it can give us a hint how do we extend entanglement transforming protocols or operational entanglement measures to three-qubit states.


\chapter*{Acknowledgments}

{}\hspace{0.45cm} I want to thank my coauthors Prof. Anthony
Sudbery, Prof. Dae-Kil Park and Dr. Sayatnova Tamaryan for the
collaboration.

I am grateful to my supervisor Dr. Lekdar Gevorgian who helped and
supported my research.

Also I would like to thank Mrs. Eleonora Tameyan who was always
near to my office and always ready to help.

\begin{appendices}


\chapter{Matrix ${\cal O}$}

\setcounter{equation}{0}
\renewcommand{\theequation}{A.\arabic{equation}}

One can easily show that the elements of ${\cal O}$ defined in
Eq.(\ref{lu1}) are given by
\begin{eqnarray}
\label{app1} & &{\cal O}_{11} = \frac{1}{2} \left( u_{11} u_{22}^*
+ u_{11}^* u_{22} +
               u_{12} u_{21}^* + u_{12}^* u_{21} \right)
                                                \\   \nonumber
& &{\cal O}_{22} = \frac{1}{2} \left( u_{11} u_{22}^* + u_{11}^*
u_{22} -
               u_{12} u_{21}^* - u_{12}^* u_{21} \right)
                                                \\   \nonumber
& &{\cal O}_{33} = |u_{11}|^2 - |u_{12}|^2    \\   \nonumber &
&{\cal O}_{12} = \frac{i}{2} \left( u_{12} u_{21}^* + u_{11}
u_{22}^* -
               u_{12}^* u_{21} - u_{11}^* u_{22} \right)
                                                \\   \nonumber
& &{\cal O}_{21} = \frac{i}{2} \left( u_{12} u_{21}^* + u_{11}^*
u_{22} -
               u_{12}^* u_{21} - u_{11} u_{22}^* \right)
                                                \\   \nonumber
& &{\cal O}_{13} = u_{11} u_{12}^* + u_{11}^* u_{12}
                                                     \\   \nonumber
& &{\cal O}_{31} = u_{11} u_{21}^* + u_{11}^* u_{21}
                                                      \\   \nonumber
& &{\cal O}_{23} = -i \left( u_{11} u_{12}^* + u_{21}^* u_{22}
\right)
                                                    \\   \nonumber
& &{\cal O}_{32} = i \left( u_{11} u_{21}^* + u_{12}^* u_{22}
\right)
\end{eqnarray}
where $u_{ij}$ is element of the unitary matrix defined in
Eq.(\ref{lu1}). It is easy to prove ${\cal O} {\cal O}^T = {\cal
O}^T {\cal O} = \openone$, which indicates that ${\cal O}_{\alpha
\beta}$ is an element of O(3).


\chapter{Bloch representation}

\setcounter{equation}{0}
\renewcommand{\theequation}{B.\arabic{equation}}
If the density matrix associated from the pure state $|\psi
\rangle$ in Eq.(\ref{state1}) is represented by Bloch form like
Eq.(\ref{density1}), the explicit expressions for $\vec{v}_i$ are
\begin{eqnarray}
\label{appb1} & &\vec{v}_1 = \left(     \begin{array}{c}
                         2 \lambda_0 \lambda_1 \cos \varphi  \\
                         2 \lambda_0 \lambda_1 \sin \varphi  \\
         \lambda_0^2 - \lambda_1^2 - \lambda_2^2 - \lambda_3^2 - \lambda_4^2
                          \end{array}
                                                    \right)
\hspace{1.0cm} \vec{v}_2 = \left(     \begin{array}{c}
            2 \lambda_1 \lambda_3 \cos \varphi + 2 \lambda_2 \lambda_4 \\
                        - 2 \lambda_1 \lambda_3 \sin \varphi  \\
         \lambda_0^2 + \lambda_1^2 + \lambda_2^2 - \lambda_3^2 - \lambda_4^2
                          \end{array}
                                                    \right)
                                                               \\  \nonumber
& &\hspace{3.0cm} \vec{v}_3 = \left(     \begin{array}{c}
            2 \lambda_1 \lambda_2 \cos \varphi + 2 \lambda_3 \lambda_4 \\
                        - 2 \lambda_1 \lambda_2 \sin \varphi  \\
         \lambda_0^2 + \lambda_1^2 - \lambda_2^2 + \lambda_3^2 - \lambda_4^2
                          \end{array}
                                                    \right)
\end{eqnarray}
and the components of $h^{(i)}$ are
\begin{eqnarray}
\label{appb2} & &h^{(1)}_{11} = 2 \lambda_2 \lambda_3 + 2
\lambda_1 \lambda_4 \cos \varphi, \hspace{1.0cm} h^{(1)}_{22} = 2
\lambda_2 \lambda_3 - 2 \lambda_1 \lambda_4 \cos \varphi
                                                                   \\  \nonumber
& &h^{(1)}_{33} = \lambda_0^2 + \lambda_1^2 - \lambda_2^2 -
\lambda_3^2
                  + \lambda_4^2,
\hspace{1.0cm} h^{(1)}_{12} = h^{(1)}_{21} = -2 \lambda_1
\lambda_4 \sin \varphi
                                                                   \\  \nonumber
& &h^{(1)}_{13} = -2 \lambda_2 \lambda_4 + 2 \lambda_1 \lambda_3
\cos \varphi, \hspace{1.0cm} h^{(1)}_{31} = -2 \lambda_3 \lambda_4
+ 2 \lambda_1 \lambda_2 \cos \varphi
                                                                    \\  \nonumber
& &h^{(1)}_{23} = -2 \lambda_1 \lambda_3 \sin \varphi,
\hspace{1.0cm} h^{(1)}_{32} = -2 \lambda_1 \lambda_2 \sin \varphi
                                                                    \\  \nonumber
& &h^{(2)}_{11} = - h^{(2)}_{22} = 2 \lambda_0 \lambda_2,
\hspace{1.0cm} h^{(2)}_{33} = \lambda_0^2 - \lambda_1^2 +
\lambda_2^2 - \lambda_3^2
                  + \lambda_4^2
                                                                     \\  \nonumber
& &h^{(2)}_{12} = h^{(2)}_{21} = 0, \hspace{1.0cm} h^{(2)}_{13} =
2 \lambda_0 \lambda_1 \cos \varphi
                                                                      \\  \nonumber
& &h^{(2)}_{31} = -2 \lambda_3 \lambda_4 - 2 \lambda_1 \lambda_2
\cos \varphi, \hspace{1.0cm} h^{(2)}_{23} = 2 \lambda_0 \lambda_1
\sin \varphi
                                                                   \\  \nonumber
& &h^{(2)}_{32} = 2 \lambda_1 \lambda_2 \sin \varphi.
\end{eqnarray}
The matrix $h^{(3)}_{\alpha \beta}$ is obtained from
$h^{(2)}_{\alpha \beta}$ by exchanging $\lambda_2$ with
$\lambda_3$. The non-vanishing components of $g_{\alpha \beta
\gamma}$ are
\begin{eqnarray}
\label{appb3} & &g_{111} = -g_{122} = -g_{212} = -g_{221} = 2
\lambda_0 \lambda_4
                                                         \\  \nonumber
& &g_{113} = -g_{223} = 2 \lambda_0 \lambda_3, \hspace{1.0cm}
g_{131} = -g_{232} = 2 \lambda_0 \lambda_2
                                                         \\  \nonumber
& &g_{133} = 2 \lambda_0 \lambda_1 \cos \varphi, \hspace{1.0cm}
g_{233} = 2 \lambda_0 \lambda_1 \sin \varphi
                                                         \\  \nonumber
& &g_{312} = g_{321} = 2 \lambda_1 \lambda_4 \sin \varphi,
\hspace{1.0cm} g_{311} = -2 \lambda_2 \lambda_3 - 2 \lambda_1
\lambda_4 \cos \varphi
                                                         \\   \nonumber
& &g_{313} = 2 \lambda_2 \lambda_4 - 2 \lambda_1 \lambda_3 \cos
\varphi, \hspace{1.0cm} g_{322} = -2 \lambda_2 \lambda_3 + 2
\lambda_1 \lambda_4 \cos \varphi
                                                           \\  \nonumber
& &g_{323} = 2 \lambda_1 \lambda_3 \sin \varphi, \hspace{1.0cm}
g_{331} = 2 \lambda_3 \lambda_4 - 2 \lambda_1 \lambda_2 \cos
\varphi
                                                               \\  \nonumber
& &g_{332} = 2 \lambda_1 \lambda_2 \sin \varphi, \hspace{1.0cm}
g_{333} = \lambda_0^2 - \lambda_1^2 + \lambda_2^2 + \lambda_3^2 -
\lambda_4^2.
\end{eqnarray}


\chapter[Geometrical interpretation]{Geometrical interpretation of the duality}

\begin{figure}[ht!]
\begin{center}
\includegraphics[height=7cm]{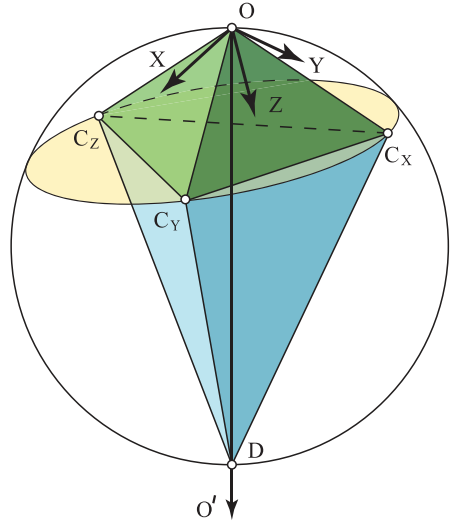}
\caption[fig7]{\label{vec}(Color online) The geometrical
interpretation of the duality for three-qubit W states. Mutually
perpendicular bold lines $OX$, $OY$ and $OZ$ are coordinate axes
and  $\overrightarrow{OO^\prime}$ is an arbitrary direction.
$OC_X,\,OC_Y$ and $OC_Z$ are mirror images of the line $OO^\prime$
in respect to the three axes. The points $C_X,\,C_Y$ and $C_Z$ are
intersections of these lines with the sphere uniquely defined by
the two conditions: its center lies on the line $OO^\prime$  and
its diameter $OD\equiv r$ is the sum of the lateral sides of the
upper pyramid (with the apex $O$ and base $C_XC_YC_Z$). Now the
direction cosines (and sines) of the vector $\overrightarrow{OD}$
are coefficients of the local states $\ket{u_i}$ in a
computational basis. And the lateral sides of the lower pyramid
(with the apex $D$ and base $C_XC_YC_Z$) are the coefficients of a
3-qubit W-state in the same basis. Thus each direction singles out
a product state and a W state and thereby establishes a
correspondence among them.}
\end{center}
\end{figure}
The nearest product state
$\ket{u_1}\o\ket{u_2}\o\cdots\o\ket{u_N}$  of the W state
\eqref{2.w} can be parameterized as follows
\begin{equation}\label{0.near}
 \ket{u_k}=\sin\t_k\ket{0}+\cos\t_k\ket{1},\, 0\leq\t_k\leq\frac{\pi}{2}, \, k=1,2,...,N,
\end{equation}
where
\begin{equation}\label{0.dircos-universal}
\cos^2\t_1+\cos^2\t_2+\cdots+\cos^2\t_N=1.
\end{equation}
Thus the angles $\cos\t_k$ define a unit N-dimensional vector in
Euclidean space. They satisfy the equalities
\begin{equation}\label{0.rmod-universal} \frac{1}{r} \equiv
\frac{\sin2\t_1}{c_1} = \frac{\sin2\t_2}{c_2} = \cdots =
\frac{\sin2\t_N}{c_N}.
\end{equation}
These equalities can be interpreted as trigonometric relations for
the right triangles with hypotenuses $r$, angles $2\t_k$, opposite
legs $c_k$ and adjacent legs $\s{r^2-c_k^2}$. If $2\t_k>\pi/2$,
then one takes the angle $\pi-2\t_k$ instead. All of these
triangles has the same hypotenuse $r$ and therefore can be
circumscribed by a single sphere with the diameter $r$. The final
picture represents two inscribed N-dimensional pyramids with a
common base and lateral sides $c_1,c_2,...,c_N$ and
$\s{r^2-c_1^2},\s{r^2-c_2^2},\s{r^2-c_N^2}$, respectively. The
case $N=3$ is illustrated in Fig.\ref{vec}.

\end{appendices}


\end{document}